\documentclass[12pt]{iopart}
\usepackage{iopams} 

\usepackage{graphicx}
\graphicspath{{../figs/}{../Fig/}}  

\usepackage{color} 
\usepackage[colorlinks]{hyperref} 
\usepackage{amsthm}

\usepackage[
    backend=biber,  
    sorting=nyt,
    style=numeric, 
    natbib=true,
    style=phys, 
    biblabel= brackets, 
    articletitle=true,  
    chaptertitle=true,  
    pageranges = true , 
    sortlocale=en_US,
    giveninits=true,
    url=false, 
    doi=false, 
    eprint=false
            ]{biblatex}

\DeclareFieldFormat
  [article,
   inbook,
   incollection,
   inproceedings,
   patent,
   thesis, 
   unpublished,
   report, 
   misc,
  ]{title}{\href{\thefield{url}}{#1}}

\newbibmacro{string+doiurlisbn}[1]{%
  \iffieldundef{doi}{%
    \iffieldundef{url}{%
      \iffieldundef{isbn}{%
        \iffieldundef{issn}{%
          #1%
        }{%
          \href{http://books.google.com/books?vid=ISSN\thefield{issn}}{#1}%
        }%
      }{%
        \href{http://books.google.com/books?vid=ISBN\thefield{isbn}}{#1}%
      }%
    }{%
      \href{\thefield{url}}{#1}%
    }%
  }{%
    \href{http://dx.doi.org/\thefield{doi}}{#1}%
  }%
}

\DeclareFieldFormat{title}{\usebibmacro{string+doiurlisbn}{\mkbibemph{#1}}}
\DeclareFieldFormat[article,incollection]{title}%
    {\usebibmacro{string+doiurlisbn}{\mkbibquote{#1}}}

\newtheorem{theorem}{Theorem}[section]
\newtheorem{corollary}{Corollary}[theorem]

\newcommand{\twot}{invariant 2-torus} 
\newcommand{\twots}{invariant 2-tori}

\newcommand{\spt}{spatiotemporal}

\newcommand{\catlatt}{spatiotemporal cat}     
\newcommand{\catLatt}{Spatiotemporal cat}     

\newcommand{\Msr}{{\mu}}                
\newcommand{\dMsr}{{d\mu}}              
\newcommand{\PV}{Percival-Vivaldi}

\newcommand{\GO}{Gutkin-Osipov}

\newcommand{\Aa}{\ensuremath{\bar{\A}}}
\newcommand{\A}{\ensuremath{\mathcal{A}}}       
\newcommand{\Ai}{\ensuremath{\mathcal{A}_0}}    
\newcommand{\Ae}{\ensuremath{\mathcal{A}_1}}    
\newcommand{\B}{\mathcal{B}}
\newcommand{\K}{\mathcal{K}}
\newcommand{\D}{\mathcal{D}}
\newcommand{\R}{\ensuremath{\mathcal{R}}}
\newcommand{\Pol}{\mathcal{P}}                  
\newcommand{\gd}{\mathsf{g}}
\newcommand{\gp}{\mathsf{g}^0}
\newcommand{\m}{\ensuremath{m}}     
\newcommand{\Mm}{\ensuremath{\mathsf{M}}}
\newcommand{\Xx}{\ensuremath{\mathsf{X}}}
\newcommand{\MmR}{\ensuremath{{\Mm_\R}}}
\newcommand{\XxR}{\ensuremath{{\Xx_\R}}}
\newcommand{\Spshift}{\ensuremath{\mathcal{S}}}  
\newcommand{\Tshift}{\ensuremath{\mathcal{T}}}   
\newcommand{\reals}{\ensuremath{\mathbb{R}}}

\newcommand{\integers}{\ensuremath{\mathbb{Z}}}
\newcommand{\Zz}{\ensuremath{\mathbb{Z}^2}}
\newcommand{\ZLT}{\ensuremath{\R_{\scriptscriptstyle\mathrm{LT}}}} 

\newcommand{\beq}{\begin{equation}}
\newcommand{\continue}{\nonumber \\ }

\newcommand{\eeq}{\end{equation}}
\newcommand{\ee}[1] {\label{#1} \end{equation}}
\newcommand{\bea}{\begin{eqnarray}}

\newcommand{\eea}{\end{eqnarray}}
\newcommand{\barr}{\begin{array}}
\newcommand{\earr}{\end{array}}

\newcommand{\rf}     [1] {~\cite{#1}}
\newcommand{\refref} [1] {\cite{#1}}

\newcommand{\refrefs}[1] {\cite{#1}}

\newcommand{\refeq}  [1] {(\ref{#1})}

\newcommand{\reffig} [1] {figure~\ref{#1}}
\newcommand{\reffigs} [2] {figures~\ref{#1} and~\ref{#2}}

\newcommand{\reftab} [1] {table~\ref{#1}}

\newcommand{\refsect}[1] {section~\ref{#1}}
\newcommand{\refsects}[2] {sections~\ref{#1} and \ref{#2}}
\newcommand{\refSect}[1] {Section~\ref{#1}}

\newcommand{\refappe}[1] {\ref{#1}} 

\newcommand\period[1]{{\ensuremath{T_{#1}}}}         
\newcommand{\cl}[1]{{\ensuremath{|#1|}}}  
\newcommand\speriod[1]{{\ensuremath{L_{#1}}}}  

\newcommand{\statesp}{state space}
\newcommand{\Statesp}{State space}
\newcommand{\stateDsp}{state-space}
\newcommand{\StateDsp}{State-space}
\renewcommand{\statesp}{phase space}
\renewcommand{\Statesp}{Phase space}
\renewcommand{\stateDsp}{phase-space}
\renewcommand{\StateDsp}{Phase-space}

\newcommand{\admissible}{admissible}

\newcommand{\inadmissible}{inadmissible}
\newcommand{\obser}{\ensuremath{A}} 

\newcommand{\ExpaEig}{\ensuremath{\Lambda}}
\newcommand{\Lyap}{\ensuremath{\lambda}}            

\newcommand{\mod}{\mbox{\rm mod}\,}

\newcommand{\block}[1]{\ensuremath{#1}} 
\newcommand{\brick}{block}
\newcommand{\Brick}{Block}
\newcommand{\prune}[1]{\ensuremath{\_{#1}\_}}        

\newcommand{\Ssym}[1]{{\ensuremath{m_{#1}}}}    

\newcommand{\po}{periodic orbit}

\newcommand{\pS}{\ensuremath{{\cal M}}}          
\newcommand{\ssp}{\ensuremath{x}}                

\newcommand{\Dn}[1]{\ensuremath{\textrm{D}_{#1}}}              

\newcommand{\dmn}{-dimensional}  

\newcommand{\ie}{{i.e.}}        

\begin{document}
    \title[Spatiotemporal cat]
{Linear encoding of the spatiotemporal cat}
    \author{
B Gutkin$^1$,
P Cvitanovi{\'c}$^2$,
R Jafari$^2$,
A K Saremi$^2$
         and
L Han$^2$
    }\address{
$^1$Department of Applied Mathematics,
Holon Institute of Technology, 58102 Holon, Israel\\
$^2$School of Physics,
Georgia Institute of Technology,
Atlanta, GA 30332-0430, USA
    } \ead{borisgu@hit.ac.il}
    \vspace{10pt}
    \begin{indented}
    \item[]
November 15, 2020
    \end{indented}

\begin{abstract}
The dynamics of an extended, spatiotemporally chaotic system might appear
extremely complex. Nevertheless, the local dynamics, observed through a
finite spatiotemporal window, can often be thought of as a visitation
sequence of a finite repertoire of finite patterns. To make statistical
predictions about the system,  one needs to know how often a given
pattern  occurs. Here we address this fundamental question within a
spatiotemporal cat, a 1-dimensional spatial lattice of coupled cat maps
evolving in time. In spatiotemporal cat, any spatiotemporal state is
labeled by a unique 2-dimensional lattice of symbols from a finite
alphabet, with the lattice states and their symbolic representation
related linearly (hence ``linear encoding''). We show that the state of
the system over a finite spatiotemporal domain can be described with
exponentially increasing precision  by a finite pattern of symbols, and
we provide a systematic, lattice Green's function methodology to
calculate the frequency (i.e., the measure) of such states.
\end{abstract}


\vspace{2pc}

\submitto{\NL}

\section{Introduction}
\label{sect:intro}

While the technical novelty of this paper is in working out details of
the {\catlatt}, an elegant, but very special  model of many-particle
dynamics (or discretization of a classical $d$\dmn\ field theory), the
vision that motivates it is much broader. We address here the long
standing problem of how to describe, by means of discrete symbolic
dynamics, the {\spt} chaos (or turbulence) in spatially extended,
strongly nonlinear field theories.

One way to capture  the essential features of turbulent motions of a
physical flow is offered by coupled map lattice models, in which the
spacetime is coarsely discretized, with the dynamics of domains that
capture important small-scale spatial structures modeled by discrete time
maps (Poincar\'e sections of a single ``particle'' dynamics) attached to
lattice sites, and the coupling to neighboring sites consistent with the
translational and reflection symmetries of the problem. Here we shall
follow this path  by investigating the Gutkin and Osipov\rf{GutOsi15}
coupled cat maps lattice (``{\catlatt}'' for short, in what follows),
built from the {Thom-}Anosov-Arnol'd-Sinai cat maps (modeling the
Hamiltonian dynamics of individual ``particles'') at sites of a $1$\dmn\
spatial lattice, linearly coupled to their nearest neighbors.

The key insight (which applies to coupled-map lattices in general, and
PDEs modeled by them, not only the system considered here) is that a
$2$\dmn\ {\spt} pattern is best described by the corresponding $2$\dmn\
{\spt} symbol lattice rather than by  a one-dimensional temporal symbol
sequence, as one usually does when describing a finite coupled
``$N$-particle'' system. The remarkable feature of the {\catlatt}  is
that its every solution is uniquely encoded by a linear transformation to
the corresponding finite alphabet $2$\dmn\ symbol lattice, a {\spt}
generalization of the {linear code} for temporal evolution of a {cat
map}, introduced in the beautiful 1987 paper by Percival and
Vivaldi\rf{PerViv}.

Within the {\catlatt}, a  window into system dynamics is provided by
a finite {\brick}  of  symbols, and the central question is to understand
which symbol \brick s are {\admissible}, and what is the likelihood of a
given {\brick}'s occurrence. It was already noted in \refref{GutOsi15}
that two {\spt} orbits that share the same sub-{\brick} {shadow} each
other exponentially well within the  corresponding {\spt} window. This is
the key property of hyperbolic {\spt} dynamics that we explore in detail
in this paper. The  linearity of the {\catlatt} enables us to go far
analytically; lattice Green's function methods enable us to compute
explicitly the measures of a large set of {\spt}ly finite \brick s, and
give an algorithm for exact computation of the rest (which is
computationally feasible  for small \brick s).

We start by formulating our ``{\catlatt}'' and stating the
main  results of the paper.

\section{Model and overview of the main results}
\label{sect:overview}

\subsection*{\catLatt}

Consider a linear, \statesp\ (area) preserving map of a 2-torus
$\mathbb{T}^2=\reals^2/\integers^2$
onto itself
 \beq
 \left(\begin{array}{c}
   x_{t+1}  \\
   p_{t+1}
  \end{array} \right )=
  A \left(\begin{array}{c}
   x_t  \\
   p_t
  \end{array} \right )\quad \mod 1
\,,\qquad
A = \left (
\begin{array}{cc}
s-1 & 1 \\
s-2 & 1 \\
\end{array}
    \right )
\,,
\ee{eq:CatMapIntr}
where both $x_t$ and $p_t$ belong to the unit interval.  For integer
$s=\tr{A} > 2$ the map is referred to as a cat map\rf{ArnAve}.
It is a
fully chaotic Hamiltonian dynamical system, which, rewritten as a 2-step
difference equation in $(\ssp_{t-1},\ssp_{t})$  takes a particularly
simple form\rf{PerViv}
\beq
\ssp_{t+1}  -  s \, \ssp_{t} + \ssp_{t-1}
    =
-\Ssym{t}
\,,
\ee{eq:CatMapNewton1}
with the unique integer ``winding number'' $\Ssym{t}$  at every time step
$t$ ensuring that $\ssp_{t+1}$ lands in
the given covering partition of the unit torus.
While the dynamics
is linear, the nonlinearity comes through the $(\mod 1)$ operation,
encoded in $\Ssym{t}\in  \A$, where  \A\ is finite alphabet of possible
values for $\Ssym{t}$.

The {cat map} is generalized to the {\em {\spt}} cat map by
considering a 1\dmn\ spatial lattice, with field $\ssp_{n,t}$
at site $n$. If each site
couples only to its nearest neighbors $\ssp_{n\pm1,t}$, and if we require
(1) invariance under spatial translations, (2) invariance under spatial
reflections, and (3) invariance under the space-time exchange, we arrive
at the 2\dmn\ Euclidean {coupled} cat map lattice:
\beq
\ssp_{n,t+1} + \ssp_{n,t-1} - {2s}\,\ssp_{n,t} + \ssp_{n+1,t} + \ssp_{n-1, t}
     = -\Ssym{n,t}
\,.
\ee{eq:CatMapNewton2}
The temporal cat map \refeq{eq:CatMapNewton1}, and the \catlatt\
\refeq{eq:CatMapNewton2} can be brought into uniform notation by
converting the {\spt} differences to discrete derivatives. This yields
the discrete screened Poisson equation\rf{Dorr70,HuCon96} for the
{$2$}\dmn\ {\em \catlatt}
\bea
 (-\Box +{2(s-2)})\,\ssp_{z} &=& \Ssym{z}
\,,\qquad
        {
  \ssp_{z} \in  \mathbb{T}^{1}  \,,\quad
  \Ssym{z} \in \A               \,,\quad
          z\in \integers^{2}
        }
\,,
\continue
 && \A = \{-3, -2,\cdots,{2s}-2,{2s}-1\}
\,,
\label{LinearConn}
\eea
where the Euclidean spacetime  Laplacian is given by
\bea
\Box\,\ssp_t &\equiv& \ssp_{t+1} - 2\,\ssp_{t} + \ssp_{t-1}
    \label{LaplTime}\\
\Box\,\ssp_{n,t}
     &\equiv&
\ssp_{n,t+1} + \ssp_{n+1,t} - 4\,\ssp_{n,t} + \ssp_{n,t-1} + \ssp_{n-1, t}
    \label{LaplSpaceTime}
\eea
in $d=1$ and $2$ dimensions, respectively,
{
As in the $d=1$ case \refeq{eq:CatMapNewton1}, the alphabet $\A$ ranges
over all ``winding number'' $\Ssym{n,t}$ values needed
to ensure that the field $\ssp_{n,t}$ on every lattice site is
confined to the $\mod 1$ interval $[0,1)$.
The \catlatt\ is smooth and fully hyperbolic for
integer ${s}>2$.
    }

The key insight  is that  {\em {$2$}\dmn} {\spt} lattice of integers
\(
\{\Ssym{z}\} = \{\Ssym{z}, z\in \integers^{{2}}\}
\)
is the natural encoding of a {$2$}\dmn\ {\spt} state.
As the relation \refeq{LinearConn} between the lattice state
 \(
\{\ssp_{z}\} = \{\ssp_{z},  z\in \integers^{{2}}\}
 \)
and its encoding   \(\{\Ssym{z}\}\) is linear, we refer to \(\Ssym{z}\) as
the ``linear code'' {following} \refref{PerViv}, both for the cat map
\refeq{eq:CatMapNewton1} in one dimension, and for the \catlatt\
\refeq{LinearConn} in {two} dimensions.
Given a set of $\{\Ssym{z}\}$, the linearity of \refeq{LinearConn}
enables us to determine the corresponding lattice state $\{\ssp_{z}\}$ by
Green's function methods.

This paper builds explicit 2\dmn\ {\catlatt}   symbolic dynamics  using
winding numbers  \(\Ssym{z}\) and Green's functions with Dirichlet
boundary conditions.
The companion paper\rf{CL18} formulates the \po\ theory for $d$\dmn\
\catlatt\ using Green's functions with periodic boundary conditions.

\subsection*{Statement of the problem}

As \catlatt\ \refeq{LinearConn} is a Hamiltonian system, it possesses the
natural measure $\Msr$ (see \refappe{sect:HamiltonCatLatt}) invariant
under space and time translations.
Let
\(
\Xx=\{\ssp_{z} \in  \mathbb{T}^{1},\, z\in \integers^2\}
\)
be a {\spt}ly infinite  solution of $d=2$ \catlatt, 
    {
generated by initial conditions generic with respect to $\Msr$
    }
in the fully hyperbolic  regime  $s>2$.
    {
By the linear relation between \Xx\ and
\Mm, symbolic representation of \Xx\ is given
by a unique {\brick}
\(
\Mm= \{\Ssym{z} \in \A \,,\; z\in \integers^2 \}
\,.
\)
    }
Assuming  now that only partial information is available,  over a finite
lattice domain $\R\subset \integers^2$, we would like to know  the
probability that    \Mm\   has  prescribed set of symbols  $\Mm_{\R}$
over  $\R$. Slightly rephrased, the  central question
studied  in this paper are the relative  frequencies $f(\Mm_{\R})$
of symbol {\brick}s within the symbolic representation of a generic
{\spt} pattern:

\bigskip

\textbf{Q1.} \textit{
  How often does a prescribed finite symbol {\brick} $\Mm_{\R}$  occur  in the
  symbolic representation \Mm\ of a generic state \Xx?
}

\bigskip

Our second question is about information stored in a finite symbol
{\brick} $\Mm_{\R}$. For the standard cat map symbolic dynamics based on
a finite Markov partition of the cat map \statesp, a {1\dmn} {\brick} of
symbols defines the corresponding trajectory up to an error which
decreases exponentially with the length of the symbol
{\brick}\rf{bowen,Rue76,BSTV97}. We would like to know whether the
corresponding result holds for $2$\dmn\ linear encoding:

\bigskip

 \textbf{Q2.}
\textit{
To what precision does the symbol \brick\ $\Mm_{\R}$ define
the local {\spt} pattern
\(
\XxR=\{\ssp_{z} \in  \mathbb{T}^{1},\, z\in\R\}
\)
?
 }

\bigskip

\subsection*{Main results}

\medskip
\noindent
Let   $\R=\{(n,t)\,|\,n=1,\dots,\ell_1,t=1,\dots\ell_2\}$  be a rectangular
domain of the lattice (an $[\ell_1\!\times\!\ell_2]$ \spt\ window), and
let \(\Mm_{\R}\) be a $[\ell_1\!\times\!\ell_2]$ \brick\ of symbols
from the alphabet $\A$. Given a generic solution \Xx\ of the equation
\refeq{eq:CatMapNewton2} and its symbolic representation
$\Mm=\{\Ssym{nt}\,|\,(n,t)\in \integers^2\}$  the  space (resp. time) shift
action on it  is given by
\beq
\Spshift \cdot\Mm =\{\Ssym{n+1,t}\,|\,(n,t)\in \integers^2\}
\,, \qquad   \Tshift \cdot \Mm =\{\Ssym{n,t+1}\,|\,(n,t)\in \integers^2\}
\,.
\ee{Shifts}
How often $\Mm_{\R}$ occurs within $\Mm$ is  then  defined
by the double limit
\beq
f(\Mm_{\R})= \lim_{T,L\to \infty} \frac{1}{L T} \sum_{t=1}^T \sum^L_{n=1}
\chi(\Spshift^n  \, \Tshift^t \cdot \Mm\,|\,\Mm_{\R}),
\ee{relFreqNew}
where
\(\chi(\Spshift^n\,\Tshift^t\cdot\Mm\,|\,\Mm_{\R})\)
is the characteristic function that equals to $1$ if symbols of
$\Spshift^n  \, \Tshift^t \cdot\Mm$   coincide with  $\Mm_{\R}$ over $\R$,
and equals to $0$ otherwise.

The $d=2$  \catlatt\  is
{
fully hyperbolic for ${s>2}$, see \refeq{qudreq}.
On the basis of our numerical simulations, we conjecture that the natural
measure $\Msr$, invariant under spatial and temporal shifts  $\Spshift$,
$\Tshift$, is uniquely ergodic, with the initial conditions for \Xx\
chosen to be generic with respect to $\Msr$.
    }
In this case the limiting frequencies  $f(\Mm_{\R})$ for
generic solutions $\Xx$ of \refeq{eq:CatMapNewton2}   are equal to the
\emph{measures} $\Msr( \pS_{\R})$ of the cylinder sets $ \pS_{\R}$, defined as
sets of {\spt} states $\Xx$ with the same $\Mm_{\R}$ \brick\ over
the domain \R. For this reason, we shall refer to the frequencies
estimated by \refeq{relFreqNew} as {measures} of $\Mm_{\R}$,
and denote them by $\Msr(\Mm_{\R})$ in what follows.
\bigskip

\noindent \textbf{Answer to Q1.}
The \catlatt\ admits a natural 2\dmn\ linear encoding with a
finite alphabet.
We find it helpful to split the alphabet into two parts,
\[
\A=\Ai\cup \Ae
\,,
\]
where the number of symbols in the exterior alphabet \Ae\ is fixed, and
the interior alphabet  \Ai\ is a full shift, with the number of symbols
in  \Ai\ growing linearly with $s$.
The following holds:

\begin{itemize}
\item   Any {\brick}  of symbols from  \Ai\ is {\admissible}.
\end{itemize}

For any general {\spt} ergodic system, relative frequencies $f(\Mm_{\R})$
defined by \refeq{relFreqNew} provide a numerical way to
estimate measures $\Msr(\Mm_{\R})$, by generating solutions on finite
$[\speriod{}\!\times\!\period{}]$ domains, compatible everywhere locally with the defining
equation \refeq{LinearConn}, and counting the number of times $\Mm_{\R}$
occurs within each such solution.
However, due to the linear relation between a {\spt} state \Xx\ and its
symbolic encoding \Mm, for the \catlatt\ one can do much better, and compute
measures $\Msr(\Mm_{\R})$ \emph{analytically and explicitly:}

\begin{itemize}

\item   Measures of \brick s $ \Mm_{\R}$ are given by rational numbers
and factorize  into products of constant and geometrical parts:
\[
\Msr(\Mm_{\R})= d_{\R}|\Pol(\Mm_{\R})|
\,.
\]
The constant $d_{\R}$ depends only {on domain \R}, independent of  the
symbolic content $M_{\R}$. The factor  $|\Pol(\Mm_{\R})|$ admits
geometrical interpretation as  the volume of polytope $\Pol(\Mm_{\R})$  in
the $|\partial \R|$\dmn\ Euclidean space, where $|\partial \R|$ is the
number of boundary points of the domain $\R$.   The  {polytope}
$\Pol(\Mm_{\R})$ is determined by the content of $ \Mm_{\R}$. For small
{$|\R|$}, $\Msr(\Mm_{\R})$ can be evaluated analytically, see
\refsects{sect:catFreqEval}{sect:catLattFreqEval}.

\item   If $\Mm_{\R}$  is composed only of  symbols from  \Ai,
then  $|\Pol(\Mm_{\R})|=1$ and $\Msr(\Mm_{\R})=d_{\R}$.
\end{itemize}

\bigskip

\noindent \textbf{Answer to Q2.}  The {\brick}  of  symbols $\Mm_{\R}$ defines the
{\spt} state \Xx\ over \R\ up to an error which decreases exponentially
with the size of the domain \R:

\begin{itemize}
\item The difference between any solution  $\ssp_{z_0} $ of \refeq{LinearConn}
for $z_0 \in \R$ and the  ``average coordinate''  $ \bar{x}({\Mm_\R})$,
determined solely by $\Mm_{\R}$, is bounded by
 \beq
  |\ssp_{z_0}-\bar{x}({\Mm_\R})| \leq C e^{-\nu \ell (z_0, \partial \R)}
  \,, \qquad \nu >0,
\ee{BGapprox}
where $\ell(z_0, \partial \R)$ is the minimal Euclidean distance between
$z_0$ and the  boundary $\partial \R$.
For explicit formulas for $\bar{x}({\Mm_\R})$ in terms of the {\brick} of
symbols ${\Mm_\R}$ see \refeq{inverseq} and \refeq{BGaverCatLattPt}.
\end{itemize}

\noindent Two remarks are in order.
First, all of  the above results hold also for the $d=1$ temporal cat map
if $s>2$, see \refsect{sect:catMap}. In this case, the domain \R\ is just
an interval in $\integers$ with two endpoints, $|\partial \R|=2$.
Second,  it is plausible that our results  hold for \catlatt\
\refeq{LinearConn} on  a lattice of dimension $d$,
{
provided that  {$s>2$}, \ie, the system is in the chaotic regime,
    }
but, in order to streamline the exposition, we discuss here only the 1-
and 2\dmn\ cases.

In principle, having answers to Q1 and Q2  allows  for a calculation
of the expectation values of
observables by means of symbolic dynamics. As $\Mm_{\R}$
defines positions of points in the center $z_0$ of domain \R\ with exponential
precision, any  observable $\obser(z)$  can be viewed as  a
function of $\Mm_{\R}$, $\obser(z_0) \approx \obser\left( \Mm_\R
\right)$, where  the  quality of  the approximation increases
exponentially  with the size of \R. In the limit of large domain size
$|\R|$,  one approximates the sum over states of the lattice with
exponentially increasing accuracy, and has for the average of $\obser$
\beq
\langle \obser \rangle
    =
\lim_{|\R| \to\infty} \sum_{\Mm_{\R}}
      \Msr(\Mm_{\R}) \obser\left( \Mm_\R \right)
\,,
\ee{averObserv}
where the sum is over all {\admissible} \brick s of symbols within \R. In
particular, for $\obser=-\frac{1}{|\R|}\log  \Msr(\Mm_{\R})$ the above
expression defines the {\spt} metric entropy of the system.

\bigskip

The paper is organized as follows:
\refSect{sect:catMap} is devoted to the {cat map} linear encoding,
introduced   in \refsect{sect:catLinSymDyn}.  Its  basic properties   and
the properties of the  corresponding phase space  partitions are
established in \refsect{sect:catLinPartit} and
\refsect{sect:catLinGreen}.
In \refsect{sect:catMfreq} we investigate the measures of {\admissible}
finite symbol \brick s, and show their factorization  into a geometric
part, and a constant part which depends only on the length of the
{\brick}. We evaluate explicitly the geometrical part of measures   for
short \brick s of symbols in \refsect{sect:catFreqEval}.
The key conceptual ingredient that underpins this calculation, hidden in
much algebra in what follows, are the transformations
\refeq{SingleCatJacobian} and \refeq{CoupledFactorization} from the
Hamiltonian initial state formulation to
the Lagrangian end points formulation. This replaces the exponentially
unstable Hamiltonian time evolution problem by a robustly convergent
Lagrangian boundary value problem.

In \refsect{sect:CCMs}, we extend these results to the \catlatt. We show
in \refsect{sect:CCMlinSymDyn}  that the system admits  a natural
2\dmn\ linear encoding with a finite alphabet,
and then  compute the measures of finite
{\spt} symbol \brick s  in \refsect{sect:CCMmeasBrick}.
In \refsect{sect:twots},  we use these  results to construct
sets of {\spt} \twots\ that fully  shadow each other.
Implementing this program requires extensive use of lattice Green's
functions, whose properties  are derived in \refappe{sect:Green}.
The  Hamiltonian formulation of the {\spt} cat map and the
metric entropy estimation are provided in  \refappe{sect:HamiltonCatLatt}
and \refappe{sect:metricEntropy}, respectively.
The results are summarized and some open questions discussed in the
\refsect{sect:summary}.

\section{Cat map}
\label{sect:catMap}

Before turning to the ``many-particle'' case, it is instructive to motivate
our formulation of the {\catlatt} by investigating the {cat map}
(\ie, a ``spatial lattice'' with only one site).
We start by a brief review of the physical origin of cat maps.

Phase space area-preserving maps that describe kicked rotors subject to a
discrete time sequence of angle-dependent impulses $F(x_{t})$,
$t\in\integers$,
\bea
x_{t+1} &=& x_{t}+p_{t+1} \qquad  \mod 1, \label{PerViv2.1b}
    \\
p_{t+1} &=& p_{t} + F(x_{t})             \label{PerViv2.1a}
\,,
\eea
play important role in the theory of chaos in Hamiltonian systems, from the
Taylor, Chirikov and Greene  standard map\rf{Lichtenberg92,Chirikov79}, to
the cat maps that we study here. Here $2\pi x$ is the  angle of the rotor,
$p$ is the momentum conjugate to the configuration coordinate $x$,
$F(x)=F(x+1)$ is periodic with period $1$, and the time step has been set to
$\Delta t= 1$.
Eq.~\refeq{PerViv2.1b} says that in one time step $\Delta t$ the
configuration trajectory starting at $\ssp_{t}$ reaches $\ssp_{t+1} =
\ssp_{t}+p_{t+1}\Delta t$, and \refeq{PerViv2.1a} says that at each kick
the angular momentum $p_{t}$ is accelerated to $p_{t+1}$ by the force
pulse $F(x_{t})\Delta t$. As the values of $x$ differing by integers are
identified, and the momentum $p$ is unbounded, the \statesp\ is a
cylinder. However, to analyse the dynamics, one can just as well
compactify the \statesp\ by folding the momentum dynamics onto a circle,
by adding ``$\mod 1$'' to \refeq{PerViv2.1a}. This reduces the dynamics
to a toral automorphism acting on a  $[0,1)\times [0,1)$ square of unit
area, with the opposite sides identified.

The simplest example of (\ref{PerViv2.1b},\ref{PerViv2.1a}) is a rotor subject to a
force
\(
 F(x) = Kx
\) 
linear in the displacement $x$. The $\mod 1$ added to \refeq{PerViv2.1a}
makes this a discontinuous ``sawtooth,'' unless $K$ is an integer. In that
case the map (\ref{PerViv2.1b},\ref{PerViv2.1a}) is a continuous automorphism
of the torus, or a ``cat map''\rf{ArnAve}, a linear symplectic
map on the unit 2-torus \statesp,
\( 
(x \mapsto A x
    \,
|\,
x
  \in \mathbb{T}^2 =  \reals^2/\integers^2
    \,;\; 
A \in SL(2,\integers)
)
\,,
\) 
with coordinates $x =(x_t,p_t)$
interpreted as the angular position variable and its conjugate momentum
at time instant $t$. Explicitly:
 \beq
 \left(\begin{array}{c}
 x_{t+1}  \\
   p_{t+1}
  \end{array} \right )=
  A \left(\begin{array}{c}
 x_t  \\
   p_t
  \end{array} \right )\quad \mod 1
    \,,  \qquad
 {A} =\left(\begin{array}{cc}
 a & c \\
 d & b
  \end{array} \right)
\,,
\ee{eq:CatMap}
where $a,b,c,d$ are any integers that satisfy
$\det A=1$, so that the map is symplectic (area preserving).

A cat map is a fully chaotic Hamiltonian dynamical system if its
stability multipliers
$(\ExpaEig\,,\; \ExpaEig^{-1})$, where
\beq
\ExpaEig=(s+\sqrt{(s-2)(s+2)})/2
\,,\qquad
\ExpaEig=e^{\Lyap}
\,,
\ee{StabMtlpr}
are real, with a positive Lyapunov exponent $\Lyap >0$. The eigenvalues are
functions of a single parameter $s=\tr{A}=\ExpaEig+\ExpaEig^{-1}$,
and the map is chaotic if and only if $|s| > 2$.
We shall refer here to the least unstable of the cat
maps \refeq{eq:CatMap}, with $s=3$, as the ``Arnol'd'', or
``Arnol'd-Sinai cat map''\rf{ArnAve,deva87}, and to general maps with integer
$s\geq 3$ as ``cat maps''.
Cat maps have been extensively analyzed as particularly simple examples
of chaotic Hamiltonian dynamics. They  exhibit ergodicity, mixing,
exponential sensitivity to variation of the initial  conditions (the
positivity of the Lyapunov exponent), and  the
positivity of  the  Kolmogorov-Sinai  entropy\rf{StOtWt06}.
Detailed
understanding of dynamics of cat maps is important also for the much richer
world of nonlinear hyperbolic toral automorphisms,
see \refrefs{Creagh94,GarSar04,SlBaJu16} for examples.

\subsection{Linear encoding}
\label{sect:catLinSymDyn}
\renewcommand{\statesp}{state space}
\renewcommand{\Statesp}{State space}
\renewcommand{\stateDsp}{state-space}
\renewcommand{\StateDsp}{State-space}

Eqs.~(\ref{PerViv2.1b},\ref{PerViv2.1a}) are the discrete-time Hamilton's
equations, which induce temporal evolution on the 2-torus
$(\ssp_{t},p_{t})$ {\em phase space}. For the problem at hand, it pays to
go from the Hamiltonian $(\ssp_{t},p_{t})$ phase space formulation to the
Newtonian $(\ssp_{t-1},\ssp_{t})$ {\em \statesp} formulation\rf{PerViv},
with $p_t$ replaced by
\(
p_t = (\ssp_{t} - \ssp_{t-1})/\Delta t \,.
\)
Eq.~(\ref{PerViv2.1a}) then takes the 2-step difference form (the discrete
time Laplacian $\Box$ formula for the second order time derivative
${d^2}/{dt^2}$, with the time step set to $\Delta t=1$),
\beq
\Box \, \ssp_t \equiv \ssp_{t+1} - 2\ssp_{t} + \ssp_{t-1} = F(\ssp_{t})  \qquad  \mod 1
\,,
\ee{PerViv2.2}
\ie, Newton's Second Law: ``acceleration equals force.''
For a cat map, with force $F(x)$ linear in the displacement $x$, the
Newton's equation of motion \refeq{PerViv2.2} takes the form
\beq
(\Box  +2 - s)\,\ssp_{t} =-\Ssym{t}
\,,
\ee{OneCat}
with $\mod 1$ enforced by $\Ssym{t}$'s, integers from  the alphabet
\beq
\A=\{\underline{1},0,\dots s\!-\!1\}
\,,
\ee{catAlphabet}
necessary  to keep $\ssp_{t}$ for all times $t$ within the unit interval
$[0,1)$.
For the sake of notational convenience, we have introduced here the
symbol
{$\underline{m_t}$}
to denote $m_t$ with the negative sign, \ie, `$\underline{1}$' stands for
the symbol `$-1$'.

  \begin{figure}
  \begin{center}  
  \setlength{\unitlength}{0.65\textwidth}
  \begin{picture}(1,0.81984366)%
    \put(0,0){\includegraphics[width=\unitlength]{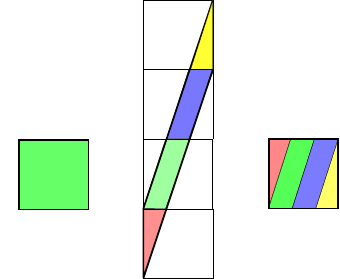}}%
    \put(-0.025,0.15){\color[rgb]{0,0,0}\makebox(0,0)[lb]{\smash{$(0,0)$}}}%
    \put(0.26669086,0.17307744){\color[rgb]{0,0,0}\makebox(0,0)[lb]{\smash{B}}}%
    \put(0.26669086,0.41839762){\color[rgb]{0,0,0}\makebox(0,0)[lb]{\smash{C}}}%
    \put(0.02137069,0.41839762){\color[rgb]{0,0,0}\makebox(0,0)[lb]{\smash{D}}}%
    \put(0.38935094,0.00135332){\color[rgb]{0,0,0}\makebox(0,0)[lb]{\smash{B}}}%
    \put(0.64284833,0.60647641){\color[rgb]{0,0,0}\makebox(0,0)[lb]{\smash{C}}}%
    \put(0.64284833,0.79455521){\color[rgb]{0,0,0}\makebox(0,0)[lb]{\smash{D}}}%
    \put(0.02,0.27){\color[rgb]{0,0,0}\rotatebox{90.0}{\makebox(0,0)[lb]{\smash{$\ssp_{0}$}}}}%
    \put(0.39,0.27){\color[rgb]{0,0,0}\rotatebox{90.0}{\makebox(0,0)[lb]{\smash{$\ssp_{1}$}}}}%
    \put(0.76,0.27){\color[rgb]{0,0,0}\rotatebox{90.0}{\makebox(0,0)[lb]{\smash{$\ssp_{1}$}}}}%
    \put(0.1469525,0.173){\color[rgb]{0,0,0}\makebox(0,0)[lb]{\smash{$\ssp_{-1}$}}}%
    \put(0.51493279,0.174){\color[rgb]{0,0,0}\makebox(0,0)[lb]{\smash{$\ssp_{0}$}}}%
    \put(0.86655824,0.175){\color[rgb]{0,0,0}\makebox(0,0)[lb]{\smash{$\ssp_{0}$}}}%
    \put(0.21762677,0.55482852){\color[rgb]{0,0,0}\rotatebox{43.35476392}{\makebox(0,0)[lb]{\smash{stretch}}}}%
    \put(0.74915379,0.61465381){\color[rgb]{0,0,0}\rotatebox{-46.94301089}{\makebox(0,0)[lb]{\smash{wrap}}}}%
  \end{picture}%
\end{center}
   \caption{ \label{fig:CatMapStatesp}
(Color online)
The Newtonian $s=3$  Arnol'd cat map matrix $A'$ \refeq{eq:StateSpCatMap}
keeps the origin $(0,0)$ fixed, but otherwise stretches the unit square
into a parallelogram. Translations by $\Ssym{0}$ from alphabet
$\A=\{\underline{1},0,1,2\}=$
\{%
{\color{red}red},
{\color{green}green},
{\color{blue}blue},
{\color{yellow}yellow}%
\}
bring stray regions back onto the torus.
   }
 \end{figure}

\subsection{Percival-Vivaldi linear encoding partition of the \statesp}
\label{sect:catLinPartit}

To interpret $\Ssym{t}$'s, consider the action of
the Newtonian cat map \refeq{OneCat} on a 2\dmn\ \statesp\ point
$(\ssp_{t-1},\ssp_{t})$,
\beq
 \left(\begin{array}{c}
 \ssp_{t}  \\
 \ssp_{t+1}
 \end{array} \right )=
 A' \left(\begin{array}{c}
 \ssp_{t-1}  \\
 \ssp_{t}
 \end{array} \right ) 
 - \left(\begin{array}{c}
 0  \\
 \Ssym{t}
 \end{array} \right )
 \,,  \qquad
 {A'} =\left(\begin{array}{cc}
 0 & 1 \\
 -1 & s
 \end{array} \right)
\,.
\ee{eq:StateSpCatMap}
In Percival and Vivaldi\rf{PerViv}, this representation of cat map is
referred to as ``the two-configuration representation''.
As illustrated in
\reffig{fig:CatMapStatesp}, in one time step the area preserving
map $A'$ stretches the unit square into a parallelogram, and a
point $(\ssp_{0},\ssp_{1})$ within the initial unit square
in general lands outside it, in another unit square $\Ssym{t}$
steps away. As they shepherd such stray points back into the unit
torus, the integers $\Ssym{t}$ can be interpreted as ``winding
numbers''\rf{Keating91}, or ``stabilising impulses''\rf{PerViv}.
The $\Ssym{t}$ translations reshuffle the \statesp, thus partitioning it into
$|\A|$ regions $\pS_\Ssym{}$, $\Ssym{}\in\A$.

As illustrated by \reffig{fig:CatMapStatesp}, there are the two kinds of
pieces within the state  space partition: the parallelograms
$\pS_0,\dots,\pS_{s-2}$, and the two exterior half sized  triangles
$\pS_{\underline{1}}$, $\pS_{s-1}$, labeled by the $(s\!-\!1)$-letter
\emph{interior} alphabet $\Ai$, and the two-letter \emph{exterior}
alphabet $\Ae$, respectively. For integer $s\geq2$ these alphabets are
\beq
\A=\Ai\cup\Ae
    \,,\qquad
\Ai=\{0, \cdots, s\!-\!2\}
    \,,\qquad
\Ae=\{\underline{1}, s\!-\!1\}
\,.
\ee{1dCatAlphs}

Refinements of these partitions work very much like they do for the
baker's map and the Smale horseshoe, by peering further into the future
and the past, and constructing the intersections of the future and past
partitions\rf{DasBuch}. The ``$\ell$-th level'' of partition $\pS = \cup
\pS_b$ is labeled by the set of all {\admissible} {\brick s} $b$ of length
$\ell$, composed of the past $\ell-t-1$ steps, and future $t$ steps, with
`decimal point' denoting the present,
\[
b  =
 \Ssym{t-\ell+1}\cdots\Ssym{-1}\Ssym{0}.\Ssym{1}\Ssym{2}\cdots\Ssym{t}
\,.
\]
    {
For the cat map symbol \brick s $\Mm_{\R}$ is 1\dmn, and a domain \R\
consists of $\ell$ consecutive temporal lattice sites, so in this section
we shall denote $\Mm_{\R}$ by a \brick\ $b$ of length $\ell$, and refer to
the infinite length symbol \brick\ as `itinerary'.
    }

While an {\admissible} infinite itinerary defines a unique point in the
\statesp, a finite {\brick} $b$ determines
a \emph{cylinder set}
$\pS_b$, the set of all points
in $(x_0,x_1)$ plane having itineraries of the form
\[
\cdots a_{t-\ell-1} a_{t-\ell}\Ssym{t-\ell+1}\cdots\Ssym{-1}\Ssym{0}
.
\Ssym{1}\Ssym{2}\cdots\Ssym{t} a_{t+1}a_{t+2}\cdots
\,,
\]
with fixed $\Ssym{i}$'s, and arbitrary $a_{i}\in\A$.
How these {\brick s} partition the \statesp\ is best understood by
inspecting \reffig{fig:SingleCatPartit}.

\begin{figure}
	\centering
	\includegraphics[width=0.95\textwidth]{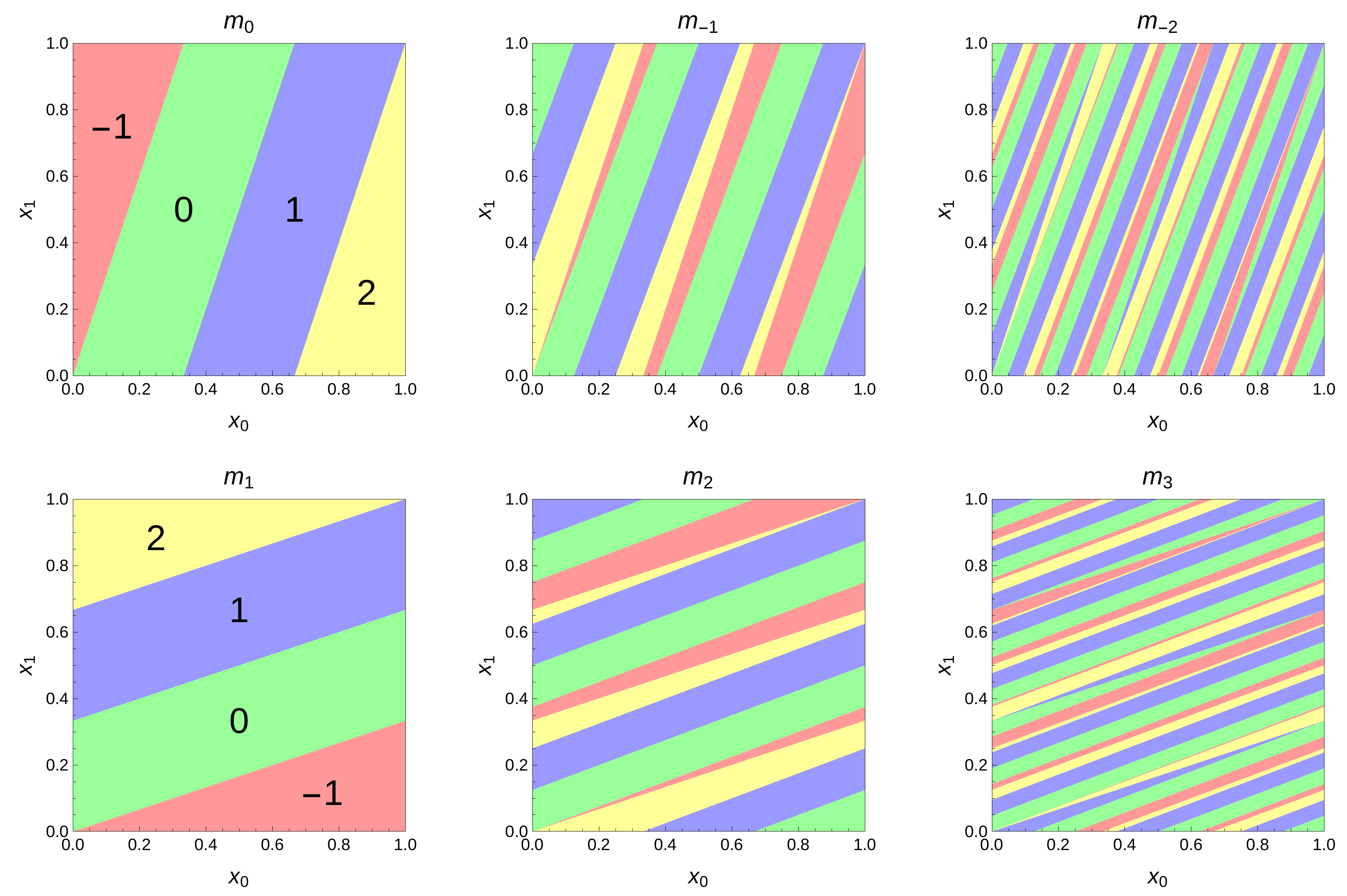}
\\
	\includegraphics[width=0.97\textwidth]{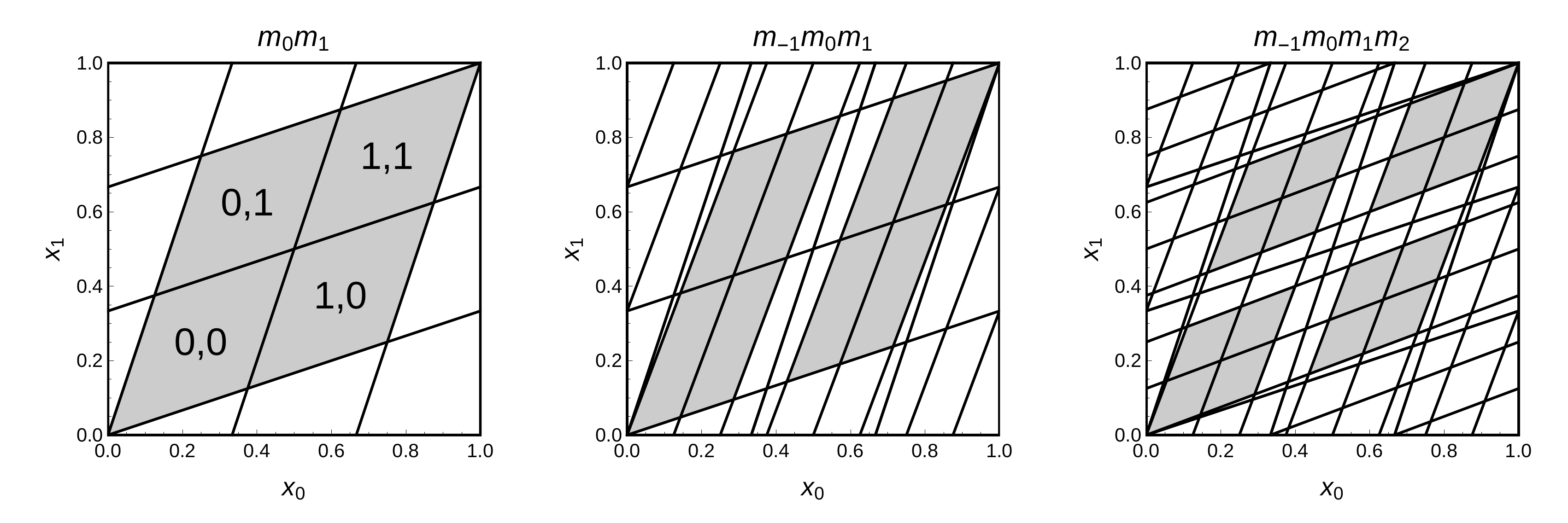}
	\caption{\label{fig:SingleCatPartit}
(Color online)
Arnol'd cat map $(\ssp_{0},\ssp_{1})$  \statesp\ partition into
(a) 4 regions labeled by $\Ssym{0}.$ , obtained from
$(\ssp_{-1},\ssp_{0})$ \statesp\ by one iteration (the same as
\reffig{fig:CatMapStatesp}).
(b) 14 regions labeled by past {\brick} $\Ssym{-1}\Ssym{0}.$, obtained from
$(\ssp_{-2},\ssp_{-1})$ \statesp\ by two iterations.
(c) 44 regions, past {\brick} $\Ssym{-2}\Ssym{-1}\Ssym{0}.$
(d) 4 regions labeled by $.\Ssym{1}$ , obtained from
$(\ssp_{2},\ssp_{1})$ \statesp\ by one backward iteration.
(e) 14 regions labeled by future {\brick} $.\Ssym{1}\Ssym{2}$ , obtained from
$(\ssp_{3},\ssp_{2})$ \statesp\ by two backward iterations.
(f) 44 regions, future {\brick} $\Ssym{3}\Ssym{2}\Ssym{1}.$
Each color has the same total area ($1/6$ for $\Ssym{t} = \underline{1},
2$, and $1/3$ for $\Ssym{t} = 0, 1$).
\Statesp\ partition into
(g) 14 regions labeled by {\brick} $b = \Ssym{0}.\Ssym{1}$, the
intersection of one past (a) and
one future iteration (d).
(h) {\brick} $b = \Ssym{-1}\Ssym{0}.\Ssym{1}$, the intersection of two
past (b) and one future iteration (d).
(i) {\brick} $b = \Ssym{-1}\Ssym{0}.\Ssym{1}\Ssym{2}$, the intersection
of two past (b) and two future iterations (e).
Note that while some regions involving external alphabet (such as
\prune{22} in (g)) are pruned, the interior alphabet labels a horseshoe,
indicated by the shaded regions.
The first three covers of the horseshoe have areas (g) $4 \times
1/8$, (h) $8 \times 1/21$, and (i) $16 \times 1/55$.
All boundaries are straight lines with rational slopes.
	}
\end{figure}

 \begin{table}
\begin{center}
\begin{tabular}{ c|cccc }
 & \underline{1} & 0 & 1 & 2\\
  \hline
 \underline{1} &  0  & 0.0208 &  0.0625 &  0.0833\\
 0 &  0.0208 &  0.1250 &  0.1250 &  0.0625 \\
 1 & 0.0625 & 0.1250 & 0.1250 &  0.0208\\
 2 & 0.0833 & 0.0625 &  0.0208&  0\\
\end{tabular}
\quad
\begin{tabular}{ c|cccc }
 & \underline{1} & 0 & 1 & 2\\
  \hline
 \underline{1} &  0    & 1/48 &  1/16 &  1/12\\
 0 &  1/48 &  1/8 &  1/8 &  1/16 \\
  1 &  1/16 & 1/8  &  1/8 &  1/48\\
  2 & 1 /12 & 1/16 &  1/48 &  0\\
\end{tabular}
\end{center}
  \caption{\label{tab:RJ2letFreq}
The measures $\Msr(\Ssym{i}\Ssym{i+1})$ of the 16 distinct 2-symbol \brick s
$\Ssym{i} \Ssym{i+1}$ for the  $s=3$ Arnol'd cat map,
    (left)
obtained from a long-time ($\sim10^9$ iterations) numerical simulation
rounded off to  four significant  digits;
    (right)
the exact values given by \refeq{PairsFreq}, or read off as sub-partition
areas in \reffig{fig:SingleCatPartit}\,(g).
Column: \ $\Ssym{i}$.  Row: \ $\Ssym{i+1}$.
See \reffig{fig:SingleCatPartit} for a geometric picture of why {\brick
s}  $\underline{1}\underline{1}$ and $22$ are pruned.
  }
\end{table}

\subsection{From itineraries to orbits and back}
\label{sect:catLinGreen}

The power of the linear encoding for a cat map\rf{PerViv} is that one can
use  integers $\Ssym{t}$ to encode its \statesp\ trajectories. Since the
connection \refeq{OneCat}  between  sequences of $\Ssym{t}$ and
$\ssp_{t}$ is linear, it is straightforward  to go back and forth between
\statesp\ and symbolic representation of an orbit. In particular, if \(
\{\Ssym{t}\} \) is an {\admissible} itinerary, the corresponding
\statesp\ point  at the  $t$ time instant is given by the inverse of
\refeq{OneCat},
\beq
  \ssp_{t}=\sum_{t'=-\infty}^\infty \gd_{tt'} \Ssym{t'}
  \,, \qquad
  \gd_{tt'} =
       \left(\frac{1}{-\Box -2 +s}\right)_{tt'}
 \,.
\ee{Coord}
The matrix $\gd_{tt'}$ is the  Green's function for 1\dmn\
discretized heat equation\rf{PerViv,varcyc} given explicitly by
$\gd_{tt'}=\ExpaEig^{-|t-t'|}/(\ExpaEig-\ExpaEig^{-1})$,
$s=\ExpaEig+\ExpaEig^{-1}$, see \refeq{StabMtlpr} and \refappe{sect:Green}.

Although the recovery  of \statesp\ \po s from finite symbol \brick s is
straightforward for the linear encoding, it is not easy to
describe the grammar rules for which symbol \brick s are
{\admissible}\rf{PerViv87b}. For the linear encoding presented
here, there is no finite set of short pruned {\brick} grammar rules, in
contrast to linear encoding for the Adler--Weiss Markov {generating}
partition of the  cat map \statesp\ given in \refref{DasBuch}. An
itinerary $\dots \Ssym{-1}\Ssym{0}\Ssym{1}\dots$  is {\admissible} if and
only if each  of the corresponding \statesp\ orbit points $\ssp_t$ in
\refeq{Coord} is in the unit interval $[0,1)$. Therefore, there is an
infinite number of conditions to satisfy. All  these conditions, however,
are automatically satisfied if the symbols $\Ssym{t}$ belong to the
interior alphabet  $\Ai$ \refeq{1dCatAlphs}. Indeed, if
$0\leq\Ssym{t}\leq s\!-\!2$ for all $t$, then
\beq
0\leq \sum_{t=-\infty}^{\infty}
      \frac{\Ssym{t}\ExpaEig^{-|t|}}{\ExpaEig-\ExpaEig^{-1}}
 \leq  \sum_{t=-\infty}^{\infty}
      \frac{(\ExpaEig^{-1}+\ExpaEig-2) \ExpaEig^{-|t|}}{\ExpaEig-\ExpaEig^{-1}}
 =1
 \,,
\ee{innerShift}
and  all  $\ssp_{t}$  generated by \refeq{Coord} fall into $[0,1)$.
As a result,  the interior part of the lattice states,  $\Ai^\integers$ is
a full shift,  with any infinite sequence of $\Ssym{t}\in\Ai$ being
{\admissible}. All grammar rules (``pruning'' of {\admissible} \brick s)
necessarily involve symbols from the {exterior} alphabet $\Ae$.

\subsection{Measures of cylinder sets}
\label{sect:catMfreq}

\subsubsection{Numerics.}

The \statesp\ coordinates $(\ssp_{0},\ssp_{1})$ are, up to the linear
shift $p_t = \ssp_{t} - \ssp_{t-1} \to \ssp_{t}$, equivalent to the
Hamiltonian $(\ssp_{0},p_{0})$ {phase space coordinates}, and as the cat map is
invertible, ergodic, and area preserving, with the invariant measure
$\dMsr=dx_t dp_t=d\ssp_{t}\,d\ssp_{t-1}$ and uniform invariant density
$\rho(\ssp_0, \ssp_1) = 1$, the measure $\Msr({b})$ corresponding to a
{\brick} $b$ 
equals to the  area $|\pS_b|$ of a {\statesp} region $\pS_b$.
The $(\ssp_{0},\ssp_{1})$  \statesp\ is composed of a disjoint union of
regions $\pS = \cup \pS_b$ labelled by all
{\admissible} {\brick s} of a fixed length $|{b}|=\ell$, so the sum of all
measures $\Msr(b)$ equals the total area of the \statesp\ $|\pS|=1$,
\beq
   \sum_{|{b}|=\ell} \Msr({{b}})=1
  \,.
\ee{totMeasure} Area sums over
subpartitions, such as
\bea
 \sum_{\Ssym{1}} \Msr(\Ssym{1} \Ssym{2} \cdots \Ssym{\ell})
     &=& \Msr(\Ssym{2} \cdots \Ssym{\ell})
\continue
 \sum_{\Ssym{\ell}} \Msr(\Ssym{1} \Ssym{2} \cdots \Ssym{\ell})
     &=& \Msr(\Ssym{1} \cdots \Ssym{\ell-1})
     \,,
\label{MarginalFreq}
\eea
provide consistency checks for  computations.

By the ergodic  theorem, the relative frequency of  appearances of a
{\brick} ${b}$ in a generic ergodic trajectory equals $\Msr(b)$. This
allows for numerical estimates of $\Msr(b)$ by long ergodic
trajectories, as illustrated in \reftab{tab:RJ2letFreq}. For the problem
at hand, there is, in principle,  no need for such simulations, as the
areas $|\pS_{b}|$ of partitions for short \brick s $b$ can be evaluated
exactly using, for example, Mathematica geometric computation
tools\rf{Mathematica}. We have computed such tables for partitions up to
\brick\ length $|{b}|=12$, but the results are quickly too unwieldy and
unilluminating to tabulate here. We visualize instead the measures by
their areas in the $(\ssp_{0},\ssp_{1})$ plane, as illustrated in
\reffig{fig:SingleCatPartit}.

\subsubsection{Analytics.}

The number of $\pS_{b}$   grows exponentially with $|b|$, while  their
areas  $|\pS_{b}|$  shrink exponentially.
Furthermore, for  larger $|b|$, the domains $\pS_{b}$  split into
disjoint sets, making  it hard to determine their areas and the pruning
rules for longer \brick s.
Because of this, for the analytical  calculation of measures $\Msr(b)$,
the Lagrangian reformulation of the problem, with  fixed
boundary points $ \ssp_0$,  $\ssp_{\ell+1}$, turns out to be  more
powerful. Moreover, as we show in \refsect{sect:CCMmeasBrick},  the
Lagrangian formulation generalizes  in a natural way to the \catlatt\ in
any spatial dimension.

Let $\{\ssp_t\}$ be   a trajectory  generated by the cat map
and let  $\{\Ssym{t}\} $ be its symbolic  representation.
As we show in
\refappe{sect:Green}, the state $\ssp_t$ at time
$t\in \{1,\dots \ell\}$ can be
expressed through the {\brick}  ${b}=\Ssym{1}\Ssym{2}\dots
\Ssym{\ell}$ at the times  $1,\dots \ell$, and  the boundary
coordinates $(\ssp_0, \ssp_{\ell+1})$:
\beq
  \ssp_t=\sum_{t'=1}^{\ell} \gd_{tt'}\Ssym{t'}
          +\gd_{t1}\ssp_0+\gd_{t\ell}\ssp_{\ell+1}
  \,, \qquad t=1,\dots \ell,
\ee{inverseq}
where  $\gd$ is  the discrete Green's function  with the Dirichlet boundary
conditions at $t=0$ and $t=\ell+1$.

Explicitly,  $\gd_{tt'}$ can be expressed in terms of Chebyshev
polynomials of the second kind
$U_n(s/2)={\sinh[(n+1)\Lyap]}/{\sinh \Lyap}$:
\[
 \gd_{ij}= \left\{\begin{array}{lr}
        \frac{U_{i-1}(s/2)U_{\ell-j}(s/2)}{U_{\ell}(s/2)}, & \mbox{for } i\leq j\\
        \frac{U_{j-1}(s/2)U_{\ell-i}(s/2)}{U_{\ell}(s/2)}, & \mbox{for } i> j .
        \end{array}\right.
\]
The first term on the right hand side of  \refeq{inverseq},
\beq
  \bar{x}_i({b})=\sum_{j=1}^{\ell} \gd_{ij}\Ssym{j}
\,,
\ee{catMapAverCoord}
can be thought of as the ``approximate state'' at time $i$. Indeed, by
\refeq{inverseq} we have
\beq
 |\ssp_i-\bar{x}_i({b})|
 = \left|\frac{U_{\ell-i}(s/2)}{U_{\ell}(s/2)}\ssp_0
   + \frac{U_{i-1}(s/2)}{U_{\ell}(s/2)} \ssp_{\ell+1}\right|
   \leq \frac{\cosh (\frac{1}{2}(\ell+1)-i)\Lyap }{\cosh \frac{1}{2}(\ell+1)\Lyap}
\,.
\ee{error}
Hence the {\brick} ${b}$ determines  the lattice state at
the center
$i=\lfloor\ell/2\rfloor$  of the \brick\,  up to an exponentially small
error in $\ell$, of the order $e^{-\ell\Lyap/2}$.

The following theorem allows  for evaluation of symbol \brick s measures.
\begin{theorem}\label{catTheorem}
Let $b$ be a finite \brick\ of $\ell$ symbols.  The corresponding
measure is given by the product
\beq
 \Msr({b}) =d_\ell |\Pol_{b}|,  \qquad  d_\ell =
  {1}/{U_{\ell}({s}/{2})},
\ee{FreqDecomp}
where
 $|\Pol_{b}|$ is  the area of the polygon $\Pol_{b}$
defined  by the inequalities
\begin{eqnarray}
 & 0&\leq \bar{x}_i({b})
     +\frac{U_{\ell-i}(s/2)}{U_{\ell}(s/2)}\ssp_0
     +\frac{U_{i-1}(s/2)}{U_{\ell}(s/2)} \ssp_{\ell+1}<1
     \,,\qquad i=1,\dots ,\ell,
\label{SquareCut1} \\
 & 0&\leq \ssp_0 <1, \qquad  0\leq \ssp_{\ell+1} <1
 \,
\label{SquareCut2}
\end{eqnarray}
in the plane $(\ssp_0,\ssp_{\ell+1})$.
\end{theorem}

\begin{proof}
From   \refeq{inverseq}  the element  $\pS_b$ of the partition $\pS$   is
defined  by the inequalities \refeq{SquareCut1} and \refeq{SquareCut2}.
 In general, the inequalities  (\ref{SquareCut1}) ''cut out''  a polygon $\Pol_{{b}}$ of the unit square (\ref{SquareCut2}) in the $(\ssp_0,\ssp_{\ell+1})$ plane.  As a
result, the measure of ${b}$ is  given by the product
\beq
 \Msr({b})=|\pS_b| =d_\ell |\Pol_{b}|
\ee{FreqDecomp1}
of  the area $|\Pol_{b}|$ of the polygon $\Pol_{b}$ and the Jacobian
$d_\ell$ of the transformation of the invariant phase space measure $\dMsr
=d\ssp_0 dp_0$ to the Lagrangian end points measure $d\ssp_0 d
\ssp_{\ell+1}$. Since the Jacobian  of the transformation from $(\ssp_0
,p_0)$ to $(\ssp_0,\ssp_1)$ equals $1$, the value of  $d_\ell$  can be
evaluated as  the Jacobian  of the transformation from $d\ssp_0 d\ssp_1$
to $d\ssp_0 d \ssp_{\ell+1}$. By \refeq{inverseq}, we therefore get
\bea
  d_\ell
= {|\partial (\ssp_0, \ssp_1)}/{\partial (\ssp_0, \ssp_{\ell+1})|}=\gd_{1\ell}  =
  {1}/{U_{\ell}({s}/{2})}
\,.
\label{SingleCatJacobian}
\eea
\end{proof}

For \brick s composed of interior symbols only  the theorem yields a simple corollary:

\begin{corollary}
If  $\Ssym{i}\in \Ai
\,, \quad  i=1,\dots |{b}|$.
The corresponding measure is given by
\[
\Msr({b})
= {1}/{U_{|{b}|}({s}/{2})}, \qquad  \Ssym{i}\in \Ai
\,, \quad  i=1,\dots |{b}|
\]
and depends  only  on the length of the \brick\ ${b}$.
\end{corollary}

\begin{proof}
If  all $\Ssym{i}$  belong to $ \Ai$, the inequalities
(\ref{SquareCut1})  are always satisfied, and  $\Pol_{{b}}$ are unit
squares of area $1$.  The proof then follows immediately   by eq.~\refeq{FreqDecomp}.
\end{proof}

In general,   the area  $|\Pol_{{b}}|$ in \refeq{FreqDecomp}  depends on the
particular \brick\ $b$; the Jacobian $d_\ell$ is the same for all ${b}$'s
of length $\ell$.
The view from $(\ssp_0, \ssp_{\ell+1})$ state space has a natural
interpretation of area as the relative measure to \brick\ of all interior
symbols, \ie, the geometrical factor $|\Pol_{b}|$. In other words, as
illustrated by comparing \reffig{fig:SingleCatPartit} with
\reffig{fig:PairSymbol}, the Hamiltonian and the Lagrangian partition areas
are the same  up to the overall Jacobian factor $d_\ell$.
The power of the Lagrangian reformulation is now evident: in contrast to the
exponentially shrinking and disjoint (for sufficiently large $|b|$)  $\pS_b$
of \reffig{fig:SingleCatPartit}, $\Pol_{{b}}$ are always  simply connected
convex polygons cut out of the unit square, see \reffig{fig:PairSymbol}.

\begin{figure}
\begin{center}
(a) \includegraphics[height=0.29\textwidth]{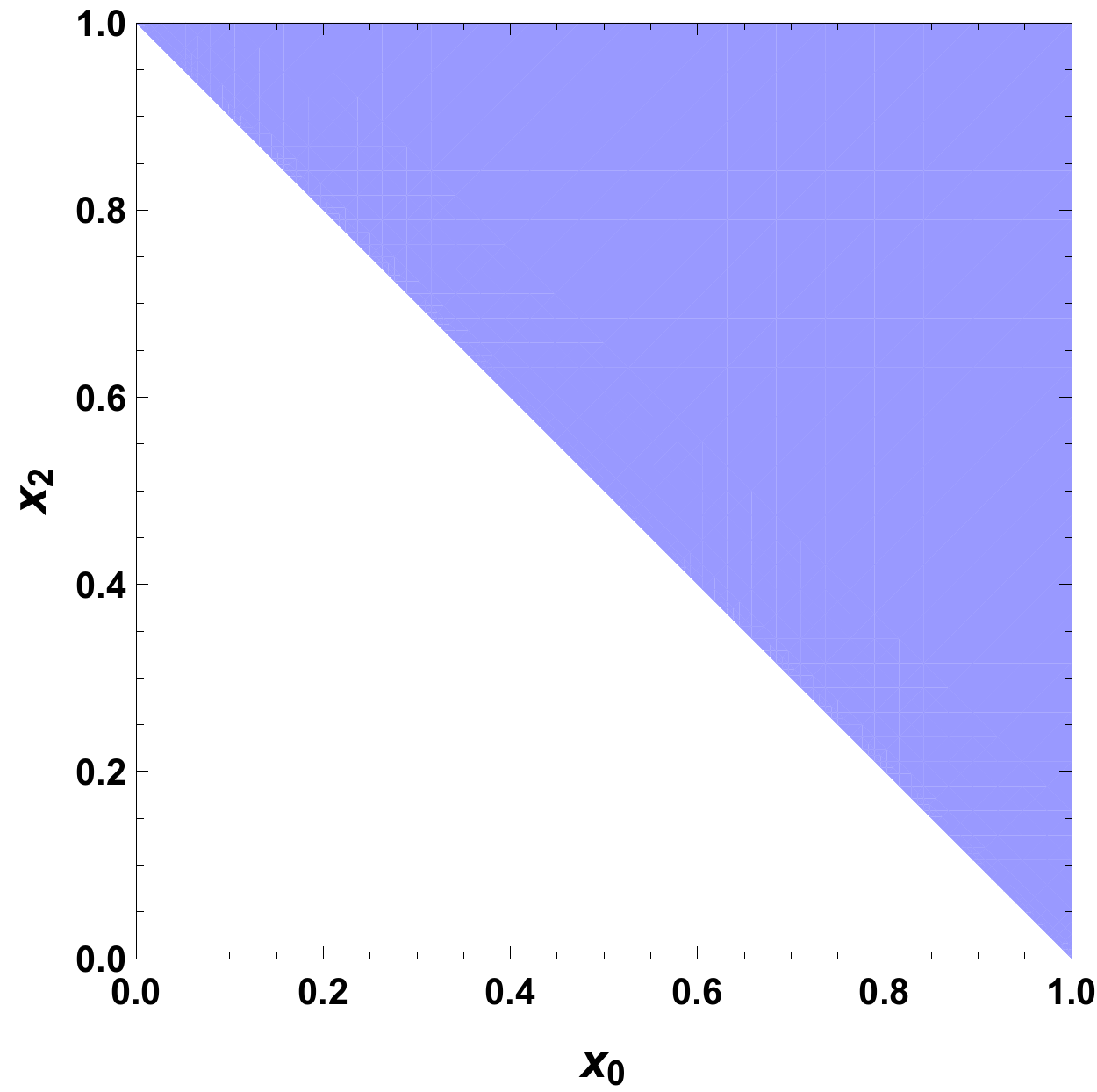}
    \hspace{0.05\textwidth} \includegraphics[height=0.29\textwidth]{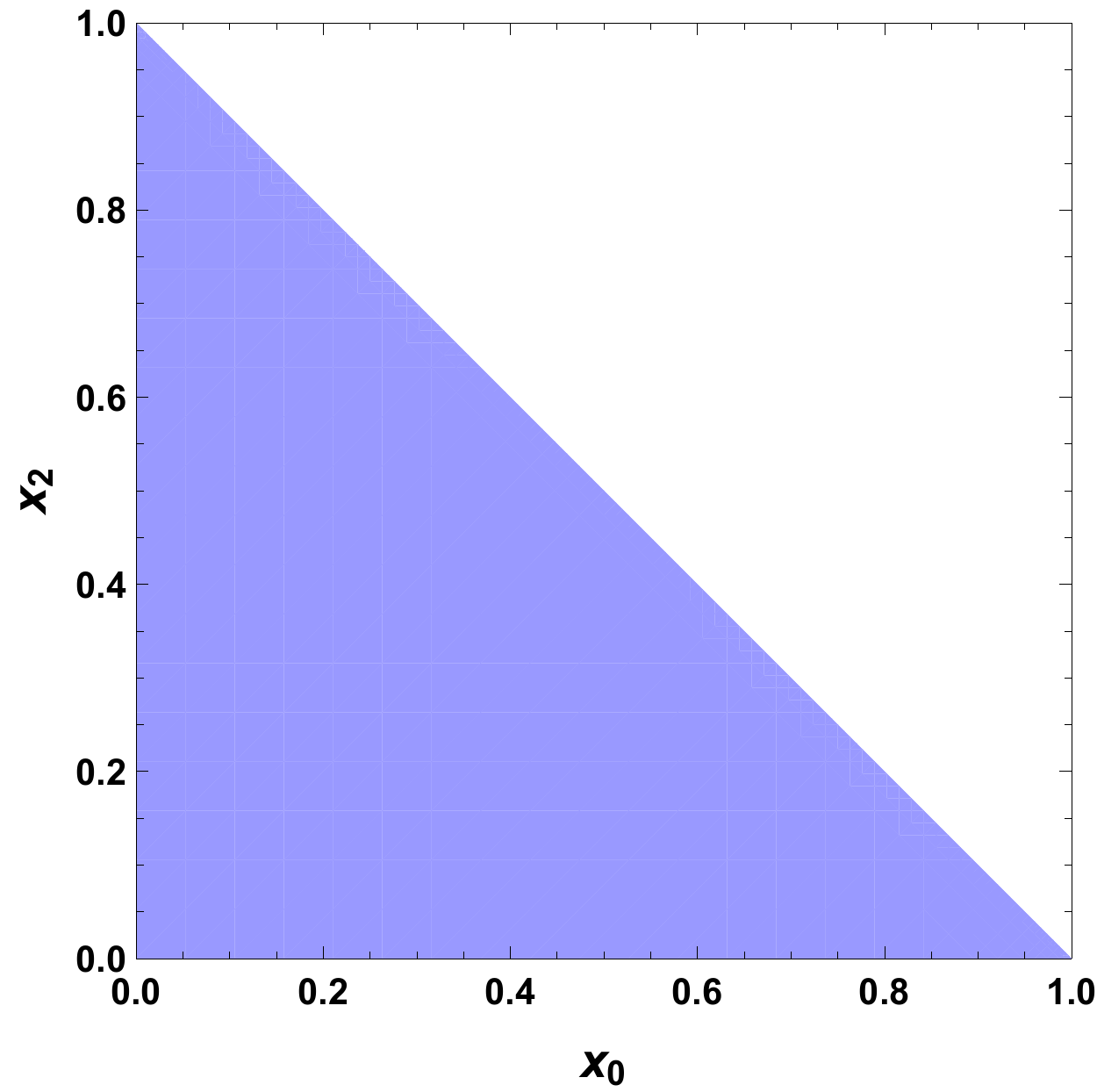}

\medskip

(b) \includegraphics[height=0.29\textwidth]{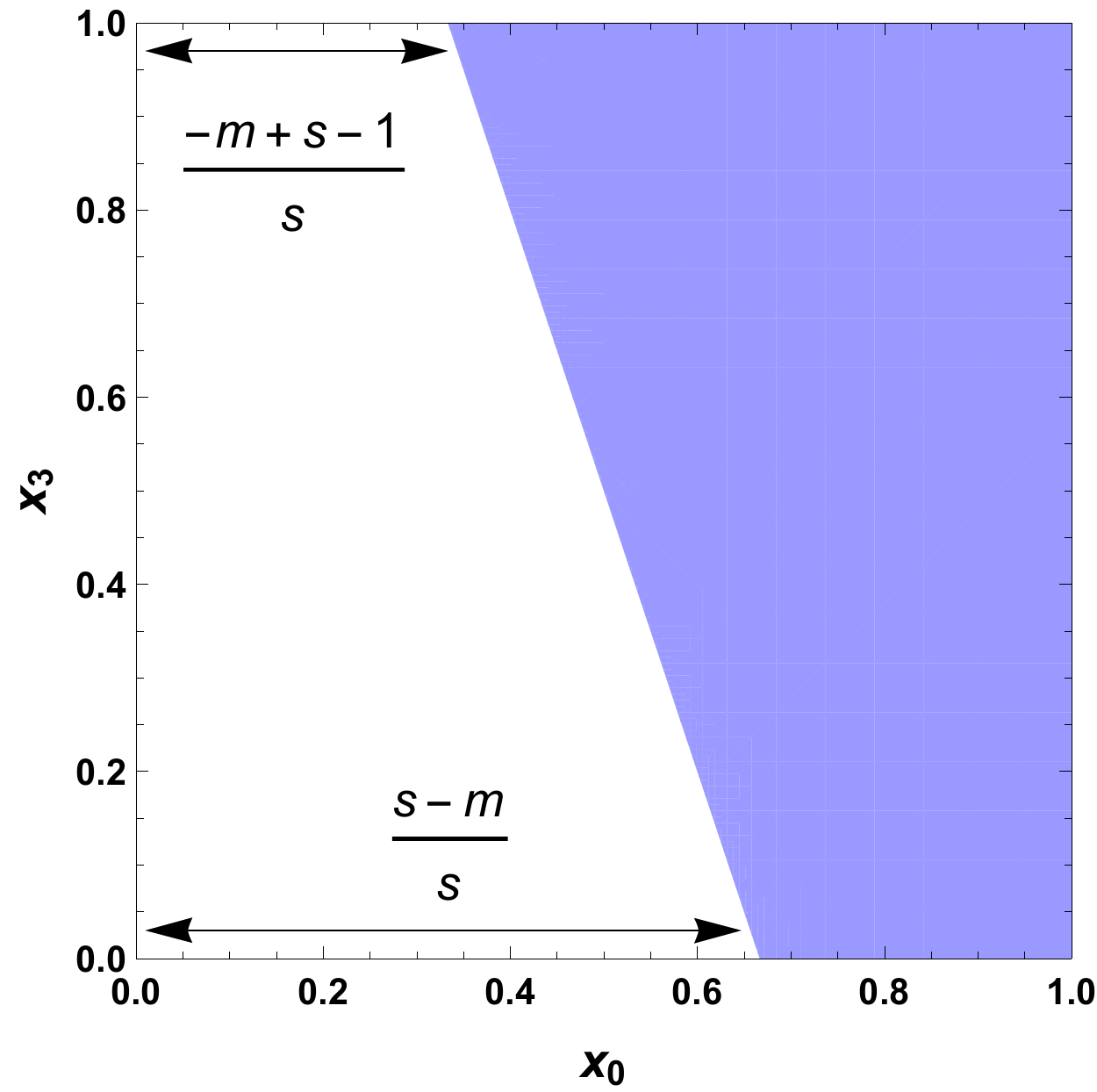}
    \hspace{0.02\textwidth} \includegraphics[height=0.29\textwidth]{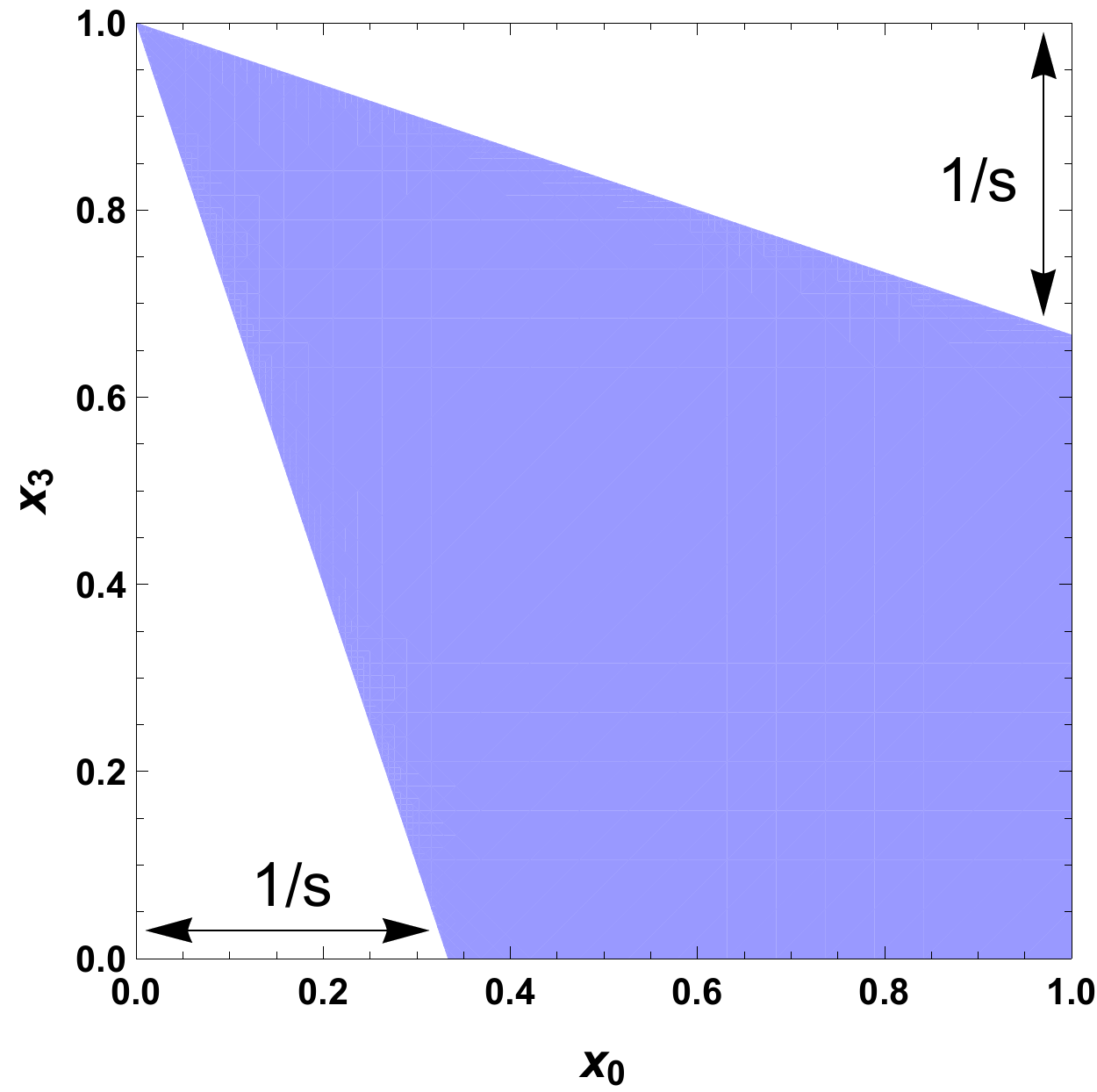}
    \hspace{0.02\textwidth} \includegraphics[height=0.29\textwidth]{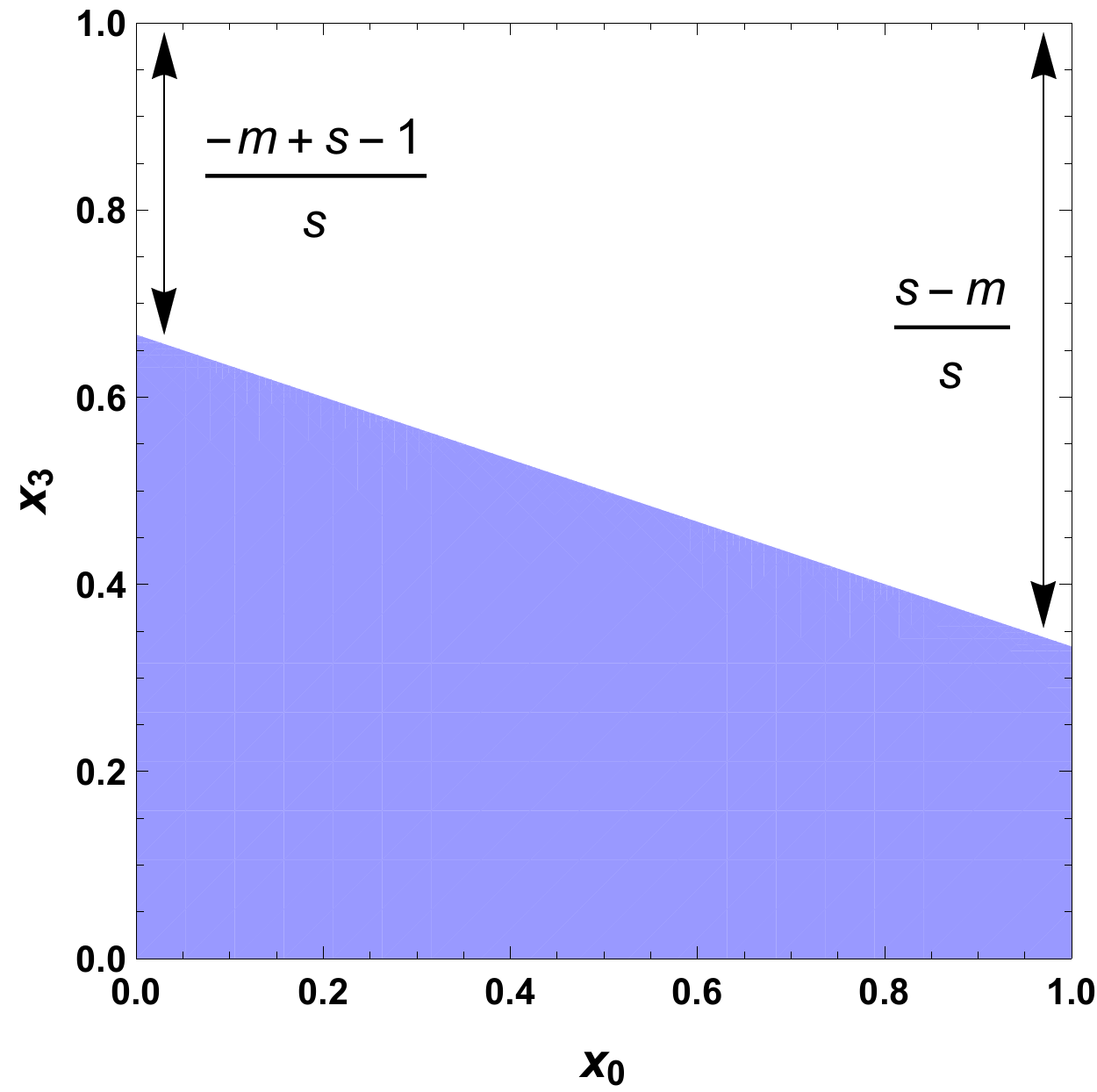}
\end{center}
\caption[]{\label{fig:PairSymbol}
(a) Polygons $\Pol_{m}$ (shaded areas)  for a single symbol,
     Lagrangian $(\ssp_0, \ssp_{2})$ plane, exterior letters $m\in \Ae$:
     (left)  $m=\underline{1}$,
     (right) $m=s\!-\!1$.
(b) Polygons $\Pol_{\Ssym{1} \Ssym{2}}$ (shaded areas) in Lagrangian
    $(\ssp_0, \ssp_{3})$ plane for {\brick s} $\block{\Ssym{1} \Ssym{2}}$
    of length 2,
    $\Ssym{1} =\underline{1},  \Ssym{2}= m\in \Ai$ (left);
    $\Ssym{1} =\underline{1},  \Ssym{2}=s-1$ (middle);
    $\Ssym{1} =s-1, \Ssym{2}=m \in \Ai$ (right).
Repeated exterior alphabet symbols are pruned,
$\Pol_{\Ssym{}\Ssym{}}= \emptyset$ if $\Ssym{}\in\Ae$.
For \brick s composed of only interior symbols $m_1,m_2\in \Ai$,  $\Pol_{\Ssym{1} \Ssym{2}}$ is the
entire unit square (full shift, no pruning).
Note that $\Msr({{b}})= \Msr({\bar{{b}}})$ by the symmetry \refeq{symmetry}.
Compare also with the \statesp\ representation \reffig{fig:SingleCatPartit}.
}
\end{figure}

\subsection{Evaluation of measures}
\label{sect:catFreqEval}

Since all coefficients in  (\ref{SquareCut1})  are given by rational
numbers, the polygon areas $|\Pol_{{b}}|$ are rational too. The same
holds for the $d_\ell$ factor. As a result,   measures  $\Msr({{b}})$ are
always rational (see, for example, \reftab{tab:RJ2letFreq}). This allows
for their exact evaluation by integer arithmetic.  As the factor $d_\ell$
in \refeq{FreqDecomp} is known explicitly, the evaluation   of
$\Msr({b})$ relies on the knowledge of the areas $|\Pol_{{b}}|$ which can
be easily found  analytically for small $\ell$. Before working out
specific examples, we list {symmetry} properties of measures
$\Msr({b})$ valid for any \brick\ length $\ell$.

\paragraph{Symmetries.}
Symmetries of the cat map induce  invariance with respect to
corresponding symbol exchanges. Define $\bar{m}=s\!-\!2\!-\!m$ to be the
conjugate of symbol $m \in \A$. For example, the two exterior
alphabet \Ae\ symbols are conjugate to each other, as illustrated by
\reffig{fig:CatMapStatesp}.
If ${b}=\Ssym{1} \Ssym{2} \dots \Ssym{\ell}$ is a
\brick, and  $\bar{{b}}=\bar{m}_1 \bar{m}_2 \dots
\bar{m}_\ell$ its conjugate, then by  reflection symmetry of the cat
map we have  $|\Pol_{{b}}|= |\Pol_{\bar{{b}}}|$. Similarly, if
$b^*=\Ssym{l}\Ssym{l-1}\dots \Ssym{1}$, the time reversal invariance implies
$|\Pol_{{b}}|=|\Pol_{{b}^*}|$. Accordingly,
\brick s $ {b}$, $ {b}^*$, and $\bar{{b}}$ have the same
measure,
 \beq
 \Msr({{b}})= \Msr({\bar{{b}}})=\Msr({{b}^*})
 \,.
 \ee{symmetry}

\subsubsection{Example: Measure of \brick s of length $\ell=1$.}

The defining, single symbol \brick\ Lagrangian equation
\refeq{eq:CatMapNewton1} is the simplest example of the Lagrangian,
two-point boundary values formulation,
\[
  \ssp_1=\gd_{11}\Ssym{1}
          +\gd_{11}\ssp_0+\gd_{11}\ssp_{2}
\,,
\]
with $\ell=1$,  $\gd_{11}=1/s$
verifying the general Green's function formula \refeq{inverseq}.
For  single symbol $\Ssym{1}\equiv m$ the set of inequalities
(\ref{SquareCut1}) thus reduces to
 \[
   -m \leq \ssp_{0}+\ssp_{2} < s-m
 \,.
\]
This constraint is always fulfilled for interior symbols $m\in \Ai$. For
$m\in \Ae$, polygon $\Pol_{m}$ is the upper and lower triangle, respectively,
shown in \reffig{fig:PairSymbol}\,(a).
As a result, we have   $|\Pol_{m}|=1$ if $m\in \Ai$ and $|\Pol_{m}|=1/2$
if $m\in \Ae$, giving measures
\[
 \Msr(m)= \left\{\begin{array}{ll}
       1/s, & \mbox{for } m\in \Ai \\
        1/2s, & \mbox{for } m\in \Ae \,,
        \end{array}\right.
\]
which indeed add up to
{one after summation over all letters of the alphabet $\A$.}

 \subsubsection{Example: Measure of \brick s of length $\ell=2$:}
 \label{1Dblock2}

For {\brick s} $\Ssym{1} \Ssym{2}$, bounds (\ref{SquareCut1}) give four
inequalities
\bea
  -s\Ssym{1}-\Ssym{2} &\leq& s\ssp_{0}+\ssp_{3} < s^2-1 - s\Ssym{1}-\Ssym{2}\\
  -s\Ssym{2}-\Ssym{1} &\leq& s\ssp_{3}+\ssp_{0} < s^2-1 - s\Ssym{2}-\Ssym{1}
  \,,
  \nonumber
\eea
where we have used $U_2(s/2)=s^2-1$.
A constraint arises  whenever at least  one of the symbols belongs to
the exterior alphabet
\Ae. By the symmetry (\ref{symmetry}), it is sufficient to analyze the
case when $\Ssym{1}=\underline{1}$. If $\Ssym{2}\in
\Ai$, the {polygon} $\Pol_{\underline{1}m_2}$ is determined by (see
\reffig{fig:PairSymbol}\,(b)):
 \[
  s-\Ssym{2}\leq s\ssp_{0}+\ssp_{3},   \qquad 0\leq \ssp_{0},\ssp_{3} <1.
 \]
The area of the  resulting polygon is equal  to   $|\Pol_{\underline{1}
m_2}|=(1+2m_2)/2s$, where   $m_2\in \Ai$.
If both $\Ssym{1} \neq \Ssym{2}$ belong to \Ae, i.e.,
$\Ssym{1}=\underline{1}$, $\Ssym{2}=s\!-\!1$, then the corresponding
polygon is determined by the conditions:
 \[
  1\leq s\ssp_{0}+\ssp_{3},   \qquad  s\ssp_{3}
        +\ssp_{0} \leq s,   \qquad 0\leq \ssp_{0},\ssp_{3} <1.
 \]
with the corresponding area  $|\Pol_{\Ssym{1} \Ssym{2}}|=1-1/s$. Finally,
if both  $\Ssym{1}=\Ssym{2}$ belong to \Ae\  and they are equal, then
\brick\ is pruned,
$\Pol_{\Ssym{1} \Ssym{2}}=\emptyset$; there are two pruned \brick s of
length 2. In summary,
\beq
 \Msr(\Ssym{1} \Ssym{2})= \left\{\begin{array}{llll}
       1/s^2-1 & \mbox{for } \Ssym{1}, \Ssym{2}\in \Ai \\
        (1+2\Ssym{2})/2s(s^2-1) & \mbox{for } \Ssym{1}=\underline{1}, \Ssym{2}\in \Ai\\
        1/s(s+1) & \mbox{for } \Ssym{1}=\underline{1}, \Ssym{2}=s\!-\!1\\
        0 & \mbox{for } \Ssym{1}=\underline{1}, \Ssym{2}=\underline{1}.
        \end{array}\right.
\ee{PairsFreq}
The measures for the remaining  symbol combinations  can be obtained by the
symmetries, see  \refeq{symmetry} and \reftab{tab:RJ2letFreq} for the $s=3$
case.

\subsubsection{Pruning.}
\label{sect:catPruning}

As shown in \refeq{innerShift}, any \brick\ of \Ai\ symbols is
{\admissible}. If, on the other hand, one or more symbols from $\Ssym{}$ belong
to \Ae, such a \brick\ might be forbidden, with the {polygon} $\Pol_{\Ssym{}}$
defined  by (\ref{SquareCut1},\ref{SquareCut2}) empty, and thus $\Msr(\Ssym{})=0$.
An example is the pruned \brick s \prune{\underline{1}\underline{1}} and
\prune{22}, missing from \reffig{fig:SingleCatPartit}\,(g). While here we
do not attempt to solve the number-theoretic problem of determining the
number of pruned \brick s for arbitrary $\ell$, the count of pruning
rules given in \reftab{tab:RJpruning} indicates that for the linear {encoding}
the number of pruned \brick s grows exponentially with their length. Thus
the linear {encoding} is not a subshift of finite type, as  its grammar
consists of an infinity of arbitrarily long pruned (\ie, {\inadmissible})
{\brick s}.
While the shaded areas of \reffig{fig:SingleCatPartit}\,(g-h) are
accounted for by the complete Smale-horseshoe grammar of the interior
alphabet, the admissibility  rules for  {\brick}s involving  letters
from $\Ae=\{\underline{1},2\}$  are not known.
In \refref{CL18}, an {Adler--Weiss Markov {generating}
partition} symbolic dynamics for the \PV\ cat map
\refeq{eq:StateSpCatMap} is constructed, with complete, finite subshift
grammar. That, however, has no bearing on  the main thrust of this paper.

\Table{\label{tab:RJpruning}   
$N_n$ is the total number of pruned {\brick s} of length  $n=\cl{b}$ for
the $s=3$ Arnol'd cat map. $\tilde{N}_n$ is the number of \emph{new}
pruned {\brick s} of length  $\cl{b}$, with all length  $\cl{b}$ {\brick
s} that contain shorter pruned {\brick s} already eliminated. Empirically
there is a single new pruning rule for each prime-number \brick\ length (it is
listed as two rules, but by the reflection symmetry there is only one).
}
\begin{tabular}{rrr} 
\br
  $n$ & $N_n$ & $\tilde{N}_{n-1}$ \\
\hline  
  2 & 2 & 0 \\
  3 & 22 & 2 \\
  4 & 132 & 8 \\
  5 & 684 & 2 \\
  6 & 3164 & 30 \\
  7 & 13894 & 2 \\
  8 & 58912 & 70 \\
  9 & 244678 & 16 \\
  10 & 1002558 & 198 \\
  11 & 4073528 & 2 \\
  12 & 16460290 & 528 \\
  13 &   & 2 \\
\br
\end{tabular}
\endTable

\section{\catLatt}
\label{sect:CCMs}

We now turn to the study of the {\catlatt} \refeq{LinearConn}, with cat maps
on sites (``particles'') coupled isotropically to their nearest neighbors on
a 2\dmn\ {\spt}ly infinite $\integers^2$ lattice.
The coupled map lattices (CML) were introduced in the mid 1980's as
models\rf{Kaneko83,Kaneko84} for studies of spatio-temporal chaos in
discretizations of dissipative PDEs. Later on, chains of coupled Anosov
maps were investigated in mathematically rigorous
settings\rf{BunSin88,PesSin88}. The conventional CML models start out
with chaotic on-site dynamics weakly coupled to neighboring sites, with
strong spacetime asymmetry. In order to establish the desired
statistical properties of CML, such as the continuity of their {SRB}
measures, \refrefs{BunSin88,PesSin88} and most of the subsequent
mathematical literature rely on the structural stability of Anosov
automorphisms under small perturbations.
Contrast this with the non-perturbative $2$\dmn\ \GO\rf{GutOsi15}
{\em \catlatt} \refeq{LinearConn}.
While this model has a Hamiltonian formulation (see
\refappe{sect:HamiltonCatLatt}), as in the {cat map} case of
\refsect{sect:catLinSymDyn}, it is instructive to write down its
equations of motion in the Lagrangian form:
\bea
 (-\Box +{2(s-2)})\,\ssp_{z} &=& \Ssym{z}
\,,\qquad
 z=(n,t)\in \integers^{2}
\,,
\continue
  \ssp_{z}\in [0,1), \quad  && \Ssym{z}\in \A
  = \{-3, -2,\cdots,{2s}-2,{2s}-1\}
\,,
\label{CoupledCats}
\eea
with $\Box $ being the discrete spacetime Laplacian \refeq{LaplSpaceTime}
on $\integers^2$. The map is space $\leftrightarrow$ time symmetric and
has the temporal and spatial dynamics strongly coupled. Furthermore, it is
smooth and fully hyperbolic for any integer $|s|>{2}$. In what follows we
will assume positive $s>{2}$.

In this paper, we focus on learning how to enumerate
{\admissible} \catlatt\ {\spt} patterns, compute their measures, and identify
their recurrences (shadowing of a large \twot\ by smaller \twots).

\subsection{Linear encoding}
\label{sect:CCMlinSymDyn}

The symbols $\Ssym{z}$ from the set
$\A=\{\underline{3},\underline{2}, \cdots,{2s}\!-\!2, {2s}\!-\!1\}$
on the right hand side of \refeq{CoupledCats} are necessary to keep
$\ssp_{z}$ within the interval $[0,1)$, with {$\underline{m_z}$}
standing here for $m_z$ with the negative sign, \ie, `$\underline{3}$'
stands for symbol `$-3$'. As we now show,
\(
{
\Mm= \{\Ssym{z} \in \A \,,\; z\in \integers^2 \}
     }
\)
can be used as a 2\dmn\ symbolic representation (code) of the
lattice system states.

Since \refeq{CoupledCats} is a linear equation, any of its solutions
\(\Xx= \{x_{z} \in [0,1) \,,\; z\in \integers^2 \}\) can be
uniquely recovered from the corresponding code \Mm.
By inverting \refeq{CoupledCats} we obtain
\begin{equation}
  \ssp_{z}=\sum_{z'\in\integers^2}\gd_{z z'} \Ssym{z'}, \qquad  \gd_{z z' }
       =\left(\frac{1}{-\Box +{2(s-2)}}\right)_{zz'}
       \,,
\label{GreenFuncCoupled}
 \end{equation}
where $\gd_{z z'}$ is the Green's function for the 2\dmn\ discretized heat
equation, see \refappe{sect:Green}. A symbol \brick\ $\Mm$ is
{\admissible} if and only if all $\ssp_{z}$ given by
\refeq{GreenFuncCoupled} fall into the interval $[0,1)$.

As for the {cat map}, we split the ${2s}+3$ letter alphabet $\A=\Ai\cup\Ae$
into the interior \Ai\ and exterior \Ae\ alphabets
\beq
  \Ai=\{0,\dots,{2({s}-2)}\},   \quad
  \Ae=
\{\underline{3},\underline{2},\underline{1}\}\cup
\{{2s}\!-\!3,{2s}\!-\!2,{2s}\!-\!1\}
\,.
\ee{2dCatLattAlph}
For example, for $s={5/2}$ the interior, respectively exterior alphabets are
\beq
  \Ai=\{0,1\},   \quad
  \Ae=\{\underline{3},\underline{2},\underline{1}\}\cup \{2,3,4\}
\,.
\ee{2dCatLattAlph5}
If all $\Ssym{z}\in \Mm$ belong to \Ai, $\Mm$ is
{\admissible}, i.e., $\Ai^{\integers^2}$ is a full shift.
Indeed, by the positivity of Green's function (see \refappe{sect:Green})
it follows immediately that $0 \leq \ssp_{z}$, while the condition
$\sum_{z'\in \Zz}\gd_{zz'}={1/2({s}-2)}$ implies that $\ssp_z\leq 1$,
with the equality saturated only if $\Ssym{z}={2({s}-2)}$, for all
$z\in \Zz$.

The key advantage of linear encoding is illustrated already by the $d=2$
case. While the size of the alphabet \Aa\ based on a Markov partition
grows exponentially with the ``particle number'' $L$, the number of
letters \refeq{LinearConn} of the linear encoding \A\ is finite and the same
for any $L$, including the $L\to\infty$ \catlatt. For the linear encoding an
\twot\ is encoded by a doubly periodic $d=2$ {\brick} $\Mm$ of symbols
from a small alphabet, rather then by a $1$\dmn\ temporal string of
symbols from the exponentially large (in $L$) alphabet \Aa.

\subsection{Finite symbol \brick s}
\label{sect:CCMmeasBrick}

Let $\R\subset\Zz$ be a {rectangle} on $\Zz$ and
let $\MmR=\{\Ssym{z}| z\in \R\}$ be a symbol \brick\ defined on \R. We
now show that $\MmR$ determines approximate positions of the points
$\ssp_z$, $z\in\R$, within the domain \R. To start with we define the
(exterior) boundary $\partial\R$ of $\R$ as a set of points adjacent to
$\R$. More precisely, $z=(n,t)$ belongs to $\partial\R$ if and only if
$z\notin \R$ but one of the four neighboring points $(n\pm 1,t)$, $(n,
t\pm 1)$ belongs to $\R$, see \reffig{fig:block2x2}(a). Let then
$\gd_{zz'}$ be the corresponding Dirichlet Green's function which
vanishes at the boundary $\partial\R$.
By the lattice Green's identity (see \refappe{sect:Green2Dident}) any
solution of the equation \refeq{CoupledCats} satisfies
\beq
 \ssp_z=\sum_{z'\in \R}\gd_{zz'} \Ssym{z'}
 + \sum_{z''\in\partial \R}\gd_{z\bar{z}''}\ssp_{z''}
 \,, \qquad  z\in \R
 \,,
\ee{DirichletGreenEquation}
with $\bar{z}''$ being the unique adjacent point of $z''\in \partial \R$
within the domain $\R$. Here, the first term
\beq
\bar{x}_z:= \sum_{z'\in \R}\gd_{zz'} \Ssym{z'}
\ee{BGaverCatLattPt}
can be viewed as the ``approximate {\spt} state''
$\bar{x}({\Mm_\R})$ at the point $z$. Importantly, it  is
determined solely by $\Mm_{\R}$. From \refeq{DirichletGreenEquation} it
follows that the difference $|\ssp_z -\bar{x}_z|$ is bounded by
\[
|\ssp_z -\bar{x}_z|
= \sum_{z''\in\partial \R}\gd_{z\bar{z}''}\ssp_{z''}
  \leq  |\partial \R|\,  \gd_{z\bar{z}''_{0}}
\,,
\]
with $\bar{z}''_{0}$ being the boundary point of $\R$
(\ie, adjacent to $\partial \R$), where the function $\gd_{z
\bar{z}''}$ attains its maximum value along $\partial \R$ (for a fixed
$z$).
For an illustration,
consider a $[\ell_1\!\times\!\ell_2]$ rectangular domain
\begin{equation}
\R^{[\ell_1\times\ell_2]}=
            \{(i, j)|\,i=0,\cdots,\ell_1\!-\!1,\;j=0,\cdots,\ell_2\!-\!1\}
\,,\label{Rectangular}
\end{equation}
with  $\ell_1$,  $\ell_2$ even (see \reffig{fig:block2x2}\,(a)),  and take
the point $z$ at  the rectangular center.
As the  Green's function
$\gd_{z\bar{z}''}$  decays exponentially with $|z-\bar{z}''|$
(see \ref{sect:Green2D}), the distance $|\ssp_z -\bar{x}_z|$ is  of the
order $e^{-\nu\ell_{\min}}$ for a large
$\ell_{\min}=\min{\{\ell_1/2,\ell_2/2\}}$, where the exponent
$\nu$ is defined by  $\cosh \nu =s/{2}$.

We determine next the measure $\Msr(\MmR)$ of the cylinder set
corresponding to $\MmR$. Take
\(
\R = \R^{[\ell_1\times\ell_2]}
\)
to be a  rectangular domain \refeq{Rectangular}. In what follows it is
convenient to distinguish   points in the interior of \R\  from the
points  which belong to the boundary $\partial \R$. While in principle
the boundary \statesp\ points $\ssp_{z''}\in\partial \R$ are labelled by
the symbol pair $z''=(n,t)$, we find it more convenient to label them by
a single index that indicates their position along the border, $\ssp_{i}
= \ssp_{z''}$, where $i$ runs from $1$ to
$\cl{\partial\R}=2(\ell_1+\ell_2)$. For examples, see
\reffig{fig:block2x2} and \refsects{exam:block1x1}{exam:block2x2}. Both
the boundary \statesp\ points $\ssp_i, i=1,\dots \cl{\partial \R}$ and
the internal points $\ssp_{z}, z\in\R$ must lie within the unit interval.

\begin{theorem}\label{catLattTheorem}
Given a \brick\ of symbols $\MmR$ on rectangle $\R$, the measure
$\Msr(\MmR)$ can be factorized into product
\beq
 \Msr(\MmR)=d(\R)\,|\Pol(\MmR)|
 \,,
\ee{CoupledFactorization}
where  $|\Pol(\MmR)|$ is the volume of the $\cl{\partial \R}$-dimensional
polytope $\Pol(\MmR)$, defined by the following inequalities
\begin{eqnarray}
& 0 & \leq  \ssp_{i}<1, \qquad i=1,\dots, \cl{\partial \R}
            \,,\label{polytop1}\\
& 0 & \leq  \bar{x}_z + \sum_{i=1}^\cl{\partial \R}
    \gd_{z\bar{z}_i}\ssp_{i}<1, \qquad z\in\R
\label{polytop}
\,
\end{eqnarray}
and   the factor $d(\R)$ depends only on the sizes $\ell_1,\ell_2$  of
\R, but not on the symbolic content of $\MmR$.
\end{theorem}

\begin{proof}
Since for all interior points $z\in \R$ one has  $0 \leq x_z<1$,
\refeq{DirichletGreenEquation} implies that  the {\admissible} set of
boundary points  $\ssp_i, i=1,\dots ,\cl{\partial \R}$  satisfy
inequalities \refeq{polytop1}  and \refeq{polytop}. Essentially, the
inequalities    \refeq{polytop}  cut   out the polytope $\Pol(\MmR)$ out
of the $\cl{\partial \R}$\dmn\ unit hypercube defined by
\refeq{polytop1}.  As a result, the measure    $\Msr(\MmR)$ is given by
the product of the  $\Pol(\MmR)$ volume and the Jacobian $d(\R)$ of
the linear transformation  between boundary  coordinates \refeq{polytop1}
and the  set  of $2(\ell_1+\ell_2)$  coordinates
\[
\{(x_{n  t_0}, x_{n t_1}) \, | \, n
 = -\lfloor\ell_2/2\rfloor,\dots ,-\lfloor\ell_2/2\rfloor+\ell_1+\ell_2-2\}
\]
at the  two  consecutive   times $t_0=\lfloor\ell_1/2\rfloor$, $t_1=t_0+1$.
Since  the Jacobian of this transformation is independent of any
particular \brick\ $\MmR$,  the
factor $d(\R)$ depends only on \R, but  not
on its symbolic content.
\end{proof}

As was the case for the single cat map theorem~\ref{catTheorem}, the
theorem yields  a simple result  for symbol  \brick s composed only
of the interior alphabet symbols:

\begin{corollary}
If all symbols in $\MmR$ belong to the interior alphabet \Ai\
\refeq{2dCatLattAlph}, then
\beq
  \Msr(\MmR)={d(\R)}
\ee{msrMmR}
is independent of the symbolic content of $\MmR$.
\end{corollary}

\begin{proof}
Note that the inequalities \refeq{polytop} are satisfied if all symbols
from $\MmR$ belong to the interior alphabet \Ai\
\refeq{2dCatLattAlph}. This follows from the positivity of the Green's
function $\gd_{zz'}$, and the identity
\[
  1 = 2(s-2)\sum_{z'\in \R}\gd_{zz'}
      + \sum_{z''\in\partial \R}\gd_{z\bar{z}''}, \qquad  z\in \R
\,,
\]
obtained by substituting the \catlatt\ \refeq{CoupledCats}  constant
field solution $\ssp_z=1$, $\Ssym{z}=2(s-2)$ into the Green's function
\refeq{DirichletGreenEquation} - see discussion following
\refeq{2dCatLattAlph5}.
As a result, for any {\brick} $\MmR$ of interior symbols   $\Pol({\MmR})$
is just a hypercube with   $|\Pol({\MmR})|=1$, and  \refeq{msrMmR}
follows immediately.
\end{proof}

\subsection{Evaluation of measures}
\label{sect:catLattFreqEval}

The evaluation of measures $\Msr(\MmR)$ for the \catlatt\   boils down to
the evaluation of  the polytope volumes $|\Pol({\MmR})|$, determined by
the inequalities \refeq{polytop1} and \refeq{polytop}. By the rationality
of every element $\gd_{zz'}$, $|\Pol({\MmR})|$ is given by a rational
number for any $\MmR$. This allows for  exact evaluation of
$|\Pol({\MmR})|$ by integer arithmetic. Once the volumes are found for
all {\admissible} \brick s $\MmR$, the constant factor $d(\R)$ can be
extracted from the normalization condition, by  summing up all volumes:
\beq
 1/d(\R) = \sum|\Pol({\MmR})|
\,.
\ee{CoupledPrefactor}
We were unable to derive any explicit formulas for $d(\R)$. However, its
asymptotic form  in the limit of large domains can be related to the
{\spt} metric entropy, as discussed in \refappe{sect:catLattEntropy}.

Before looking at specific examples of measure calculation we
  list  the symmetry  properties of  the \catlatt\ measures
$\Msr(\MmR)$.

\paragraph{Symmetries.}
Besides the invariance under shifts in time and space directions, \catlatt\
\refeq{CoupledCats} is separately invariant under the space and time
reflections $n\to -n$, $t\to -t$, as well as under exchange
$n\longleftrightarrow t$ of space and time.
{\catLatt} thus has all the symmetries of the square lattice:
\begin{itemize}
  \item 2 discrete translation symmetries
  \item
the  group $D_4$ composed of  rotations by $k\pi/2$, $k=1,2,3$ and
reflection across $x$-axis, $y$-axis, diagonal $a$, diagonal $b$:
\beq
C_{4v} = \Dn{4} = \{
E, C_{4z}^+, C_{4z}^-, C_{2z},
\sigma_{y}, \sigma_{x},
\sigma_{da},\sigma_{db},
\}
\,.
\ee{eq:C4v}
\end{itemize}
In the international crystallographic notation\rf{Dresselhaus07}, this
point group  is referred to as $p4mm$.
In addition,  the transformation
\[
x_{nt}\to 1-x_{nt},
  \qquad
\Mm=\{\Ssym{nt}\}\to \bar{\Mm}=\{\bar{m}_{nt} \}, \quad \bar{m}_{nt}
   = {2(s-2)}-\Ssym{nt}
\,,
\]
leaves eq.~\refeq{CoupledCats} invariant. All together, the measure is
invariant under
\[
 \Msr(\MmR)=\Msr(\sigma\circ\MmR), \qquad  \Msr(\MmR)=\Msr(\bar{\Mm}_\R),
\]
where $\sigma$ is an element of space group $p4mm$.
As an  example,
consider ${s=7/2}$ \catlatt, with alphabets \refeq{2dCatLattAlph}
\beq
  \Ai=\{0,1,2,3\},   \quad
  \Ae=\{\underline{3},\underline{2},\underline{1}\}\cup \{4,5,6\}
\,.
\ee{2dCatLattAlph7}
By the \Dn{4} symmetries $\Msr(\MmR)=\Msr(\sigma\circ\MmR)$ the measures of
the following eight {\brick s} are equal:
    \[
        \left[\begin{array}{cc}
 1  & 2 \\
 3 & 4\\
 5 & 6
              \end{array}\right], \qquad  \left[\begin{array}{ccc}
 2  & 4 & 6 \\
 1 & 3 & 5
              \end{array}\right], \qquad  \left[\begin{array}{cc}
 6  & 5 \\
 4 & 3\\
 2 & 1
              \end{array}\right],\qquad
      \left[\begin{array}{ccc}
 5  & 3 & 1 \\
 6 & 4 & 2
              \end{array}\right]
              \]
     \[
        \left[\begin{array}{cc}
 2  & 1 \\
 4 & 3\\
 6 & 5
              \end{array}\right], \qquad  \left[\begin{array}{ccc}
 6  & 4 & 2 \\
 5 & 3 & 1
              \end{array}\right], \qquad  \left[\begin{array}{cc}
 5  & 6 \\
 3 & 4\\
 1 & 2
              \end{array}\right],\qquad
      \left[\begin{array}{ccc}
 1  & 3 & 5 \\
 2 & 4 & 6
              \end{array}\right]
\,.
\]
In addition, the measures of \brick s such as
  \[
        \left[\begin{array}{cc}
 2  & 1 \\
 4 & 3\\
 6 & 5
              \end{array}\right]
 \quad\Leftrightarrow\quad
        \left[\begin{array}{cc}
 1  & 2 \\
 \underline{1} & 0\\
 \underline{3} & \underline{2}
              \end{array}\right]
\]
(eight additional {\brick s} in all) are equal by
$\Msr(\MmR)=\Msr(\bar{\Mm}_\R)$ symmetry.

While for a {cat map} $\cl{\partial \R}$ is always $2$, \ie, the boundary
of interval $\R$ consists of the two end points, for the {\catlatt}  the
number  $\cl{\partial \R}$ of  boundary points grows with the domain
size. The complexity of $|\Pol({\MmR})|$ calculation for a \catlatt\ thus
grows with $|\R|$, as well. We illustrate this with calculations for $\R
= [1\!\times\!1]$ and $[2\!\times\!2]$  symbol \brick s.

\subsubsection{Example: $\R = [1\!\times\!1]$
               measure.}
\label{exam:block1x1}

Consider a $\R=[1\!\times\!1]$ {\spt} domain, with a single symbol
{\brick} $\Mm$, together with the four \statesp\ points
$\ssp_i=\ssp_{z}\in\partial\R$ comprising its boundary, \reffig{fig:block2x2}\,(b).
We need to evaluate the volume of the 4\dmn\ polytope $\Pol(m)$ for
each $m\in \A$.   $\Pol(m)$ is contained with the hypercube
\[
 0 \leq \ssp_{i} <1,  \qquad i =1,2,3,4,
\]
bounded by the inequalities
\beq
 -m\leq \ssp_{1}+\ssp_{2}+ \ssp_{3} + \ssp_{4} <{2s}-m
 \,.
\ee{hyperplane}
The polytope volume $|\Pol(m)|$ depends on $m$. For the
interior letters $m\in \Ai$ the hyperplane  \refeq{hyperplane}  does not
intersect the hypercube and the volume $|\Pol(m)|=1$. For
\(
m\in\{\underline{3},\underline{2},\underline{1}\}
\)
in the exterior $\Ae$ alphabet \refeq{2dCatLattAlph}, the corresponding
volumes $|\Pol(m)|$ are $1/4!$,   $1/2$, and $23/4!$, respectively.  The
normalization condition \refeq{CoupledPrefactor}  then yields
{$d=1/(2s)$}.
Thus the measures for the symbols from the exterior alphabet \Ae\ are
\bea
\Msr(\underline{3})&=&\Msr({2s\!-\!1})= {1}/{{(2\cdot4!\,s)}}
    \continue
\Msr(\underline{2})&=&\Msr({2s}\!-\!2)= {1}/{(4s)}
    \continue
\Msr(\underline{1})&=&\Msr({2s\!-\!3})=
                            {1}/{(2s)}-{1}/{{(2\cdot4!\,s)}}
    \continue
\Msr(m)&=&{1}/{{2s}} \mbox{ for the  }{2s-3}
               \mbox{ interior letters } m\in \Ai
\,,
\label{exactBlock1x1}
\eea
with the total measure satisfying
\(
\sum_{\Ssym{}} \Msr(\Ssym{}) =1
\,.
\)
The numerical estimates of \reffig{fig:RJsymbol} confirm these analytic
results.

\begin{figure}	
(a)\;\includegraphics[width=0.33\textwidth]{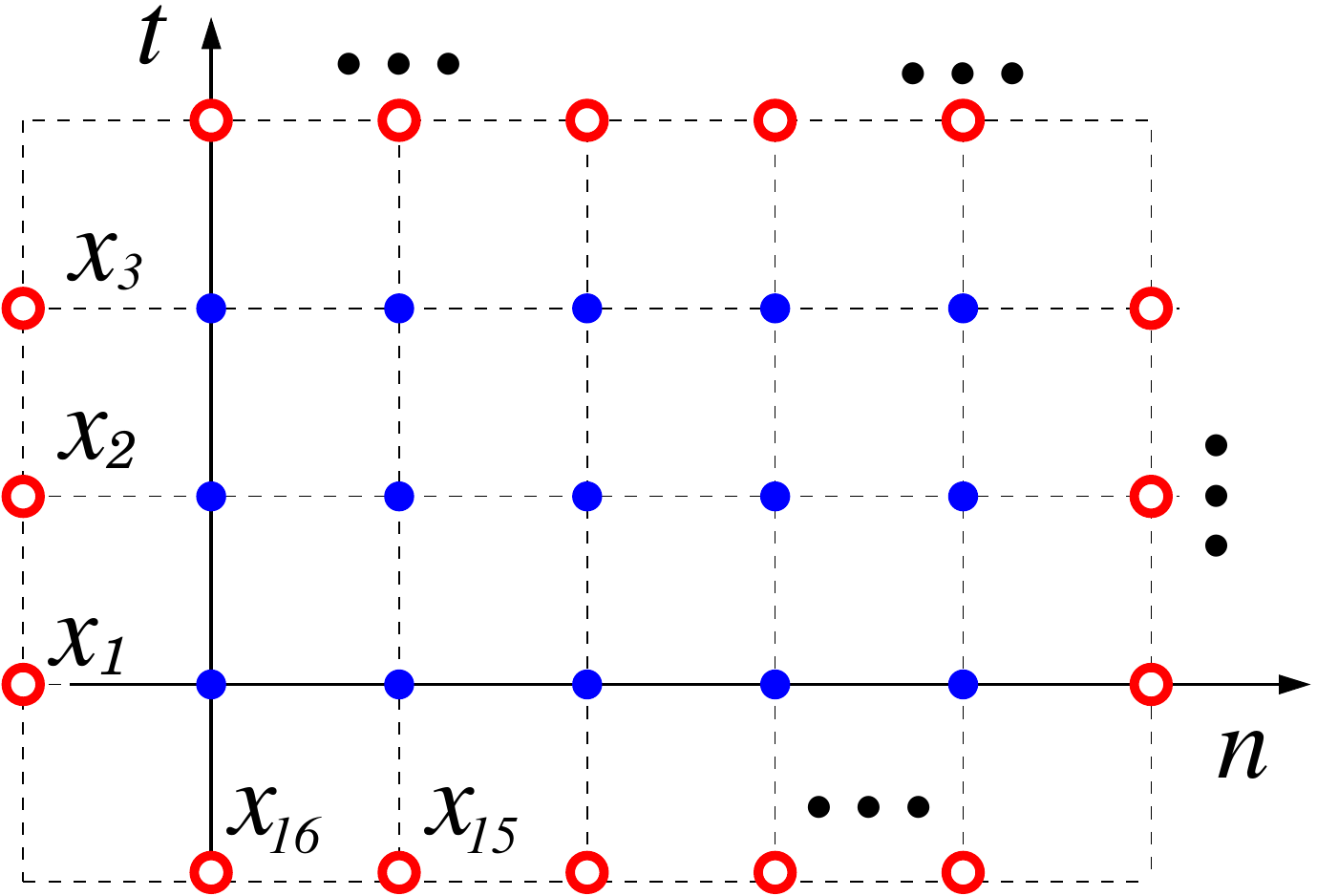}
\hspace{2mm}
(b)\;\includegraphics[width=0.23\textwidth]{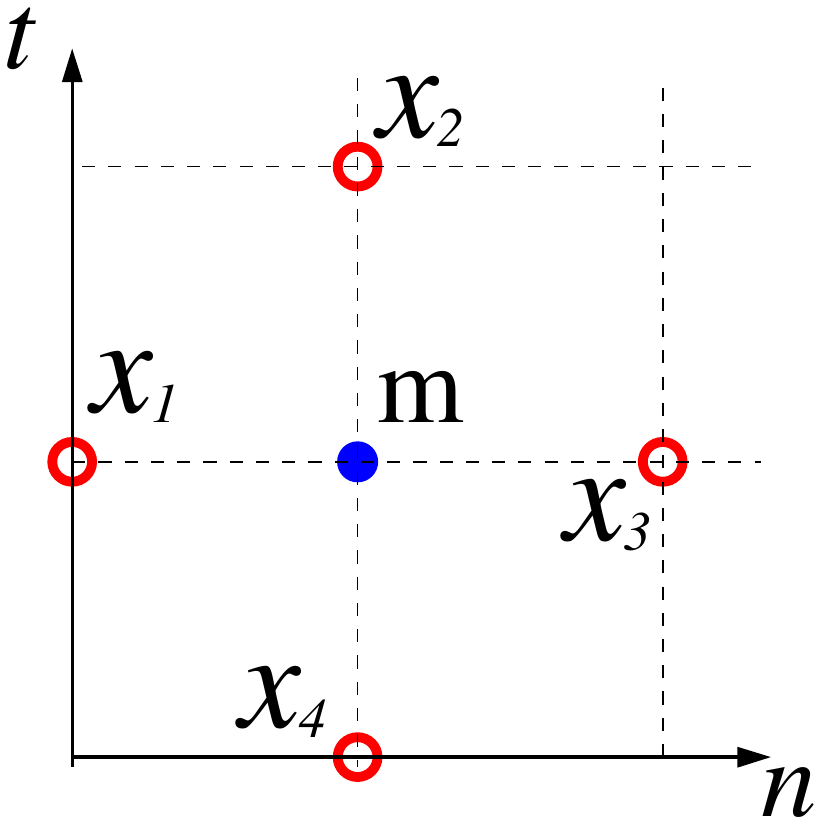}
  \hspace{2mm}
(c)\;\includegraphics[width=0.23\textwidth]{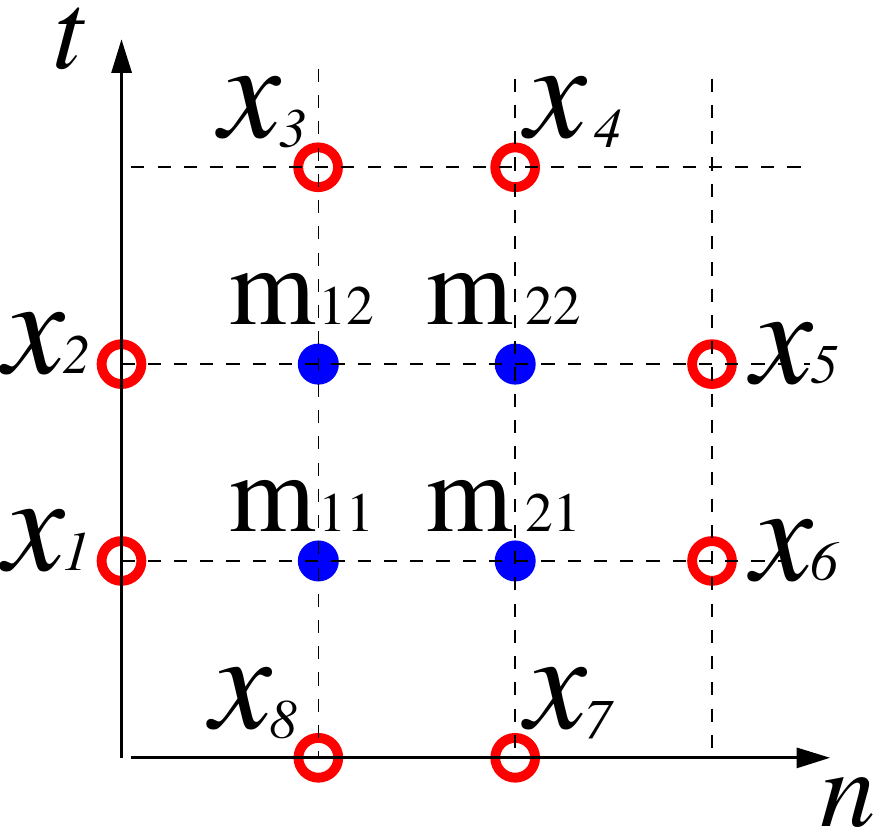}
	\caption{\label{fig:block2x2}
(Color online)
(a) A $[5\!\times\!3]$ domain $\R$ on the
2\dmn\ lattice. The red (open) circles indicate the boundary
$\partial \R$. For each  point $x_i\in \partial\R $, there exists a unique
adjacent point within the domain $\R$.
{\Brick s} $\MmR$ for
(b) $\R = [1\!\times\!1]$,  and
(c) $\R = [2\!\times\!2]$  domains, together with the corresponding
boundary points
$\partial\R=\{\ssp_{1}, \ssp_{2},  \ssp_{3}, \ssp_{4}\}$  and
$\partial\R=\{\ssp_{1}, \ssp_{2},  \cdots, \ssp_{8}\}$, respectively.
   }
\end{figure}

\subsubsection{Example: $\R = [2\!\times\!2]$ measure.}
\label{exam:block2x2}

For the {\brick}
\beq
\Mm=
        \left[\begin{array}{c}
\Ssym{12} \Ssym{22}   \\
\Ssym{11} \Ssym{21}
              \end{array}\right]
\ee{eq:block2x2}
the 8\dmn\ polytope $\Pol_{\Mm}$ is  parametrized by the boundary
$\partial \R$   points, \reffig{fig:block2x2}\,(c).  They  satisfy $0
\leq \ssp_{i} <1,  i=1,\dots 8$, supplemented by the four inequalities:
\beq
 0< P_k(x_1,\dots,x_8) \leq {8\,s(s^2-1)}, \qquad k=1,2,3,4,
\ee{eq:block2x2Inequalities}
where
\begin{eqnarray*}
\fl P_1=
{2(2s^2-1)}(\ssp_1 +\ssp_{8}+\Ssym{11})
 +{2s}(\ssp_3+\ssp_2+\ssp_6+\ssp_7+\Ssym{12}+\Ssym{21})+2(\ssp_4+\ssp_5+\Ssym{22})\nonumber\\
\fl P_2={2(2s^2-1)}(\ssp_2 +\ssp_3+\Ssym{12})
 +{2s}(\ssp_1+\ssp_{8}+\ssp_4+\ssp_5+\Ssym{11}+\Ssym{22})+2(\ssp_7+\ssp_6 +\Ssym{21})\nonumber\\
\fl P_3={2(2s^2-1)}(\ssp_7 +\ssp_6+ \Ssym{21})
 +{2s}(\ssp_1+\ssp_{8}+\ssp_4+\ssp_5 +\Ssym{22}+\Ssym{11})+2(\ssp_2+\ssp_3+\Ssym{12})\nonumber\\
\fl P_4={2(2s^2-1)}(\ssp_4 +\ssp_5 +\Ssym{22})
 +{2s}(\ssp_3+\ssp_2+\ssp_6+\ssp_7 + \Ssym{12}+\Ssym{21})+2(\ssp_1+\ssp_{8} +\Ssym{11})
\,.
\end{eqnarray*}

These inequalities lead to analytical expressions for $\Pol_{\Mm}$
volumes. A general $\Pol_{\Mm}$ volume is a four-dimensional integral,
whose calculation is lengthy and unilluminating, so we skip it here.
Instead, we evaluate $|\Pol_{\Mm}|$ for cases where some of the
inequalities (\ref{eq:block2x2Inequalities}) are satisfied  for all
points of the hypercube  $\Pol_0=\{x_1, \dots,x_8 \in [0,1)\}$. If all
letters of $\Mm$ belong to the interior alphabet, then all inequalities
hold, and $|\Pol_{\Mm}|=1$. Another easy case is the one where  three out
of four inequalities hold for all points in $\Pol_0$. For example,
consider
\[
       \Mm= \left[\begin{array}{cc}
        {s}-2    & {s}-2 \\
        \underline{2} & {s}-1
              \end{array}\right]
\]
for an even ${s>2}$. Then only the first inequality,
$0<P_1(x_1,\dots,x_8)$, is a non-trivial one, while the rest are
satisfied for all points in $\Pol_0$. Since the center of the hypercube
$\{x_i=1/2, i=1,\dots,8\}$ belongs to  the hyperplane
$P_1(x_1,\dots,x_8)=0$, the polytope $\Pol_{\Mm}$ has the same volume as
half of  the hypercube $\Pol_0$, \ie,  $|\Pol_{\Mm}|=1/2$.

It is also possible to find all {\inadmissible} symbol {\brick s}, with
$|\Pol_{\Mm}|=0$. For {\inadmissible} symbol {\brick s}, one of the
inequalities (\ref{eq:block2x2Inequalities}) must be violated for all
points of the hypercube $\Pol_0$. In particular, any combination of
symbols for $[2\!\times\!2]$ {\brick} that satisfies condition
\[
(2+\Ssym{11}){(2s^2-1)}+(4+\Ssym{12}+\Ssym{21}){s}+ (2+\Ssym{22})\leq 0
\]
is forbidden, \ie, $|\Pol_{\Mm}|=0$. This implies that symbol {\brick s}
  \[
  \left[\begin{array}{cc}
        \underline{2} & \underline{2}  \\
        \underline{2} &  \underline{2}
              \end{array}\right],  \qquad
        \left[\begin{array}{cc}
        \Ssym{12} & \Ssym{22}  \\
        \underline{3} &  \Ssym{21}
              \end{array}\right]
\]
are  {\inadmissible} if either $\Ssym{12}+\Ssym{21}\leq{2s}-6$ and
arbitrary  $\Ssym{22}$, or
$\Ssym{12}+\Ssym{21}={2s}-5,\Ssym{22}\leq{s}-3$.
Other forbidden $[2\times2]$ {\brick}s  are obtained by
application of the symmetry operations.

\subsubsection{Numerics.}
\label{sect:2Dnumerics}

\begin{figure}	
\begin{center}
(a)\;\includegraphics[width=0.35\textwidth]{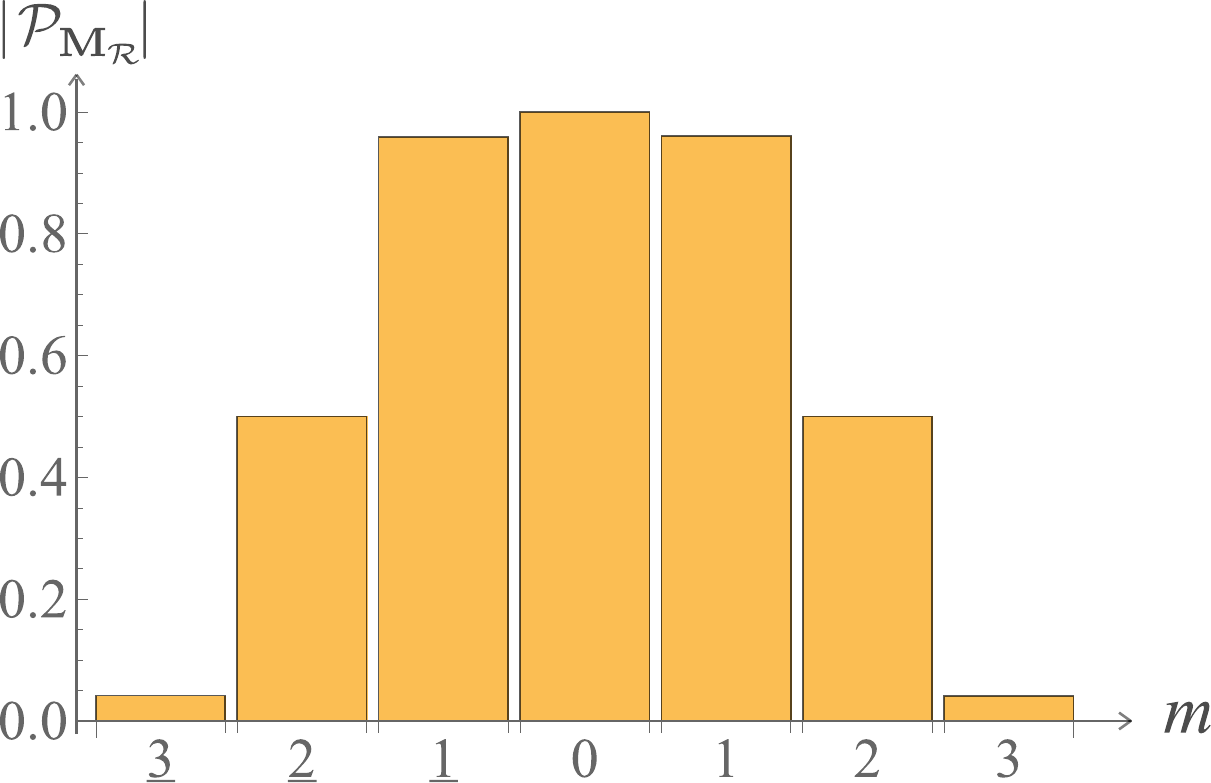}
  \hspace{4mm}
(b)\;\includegraphics[width=0.35\textwidth]{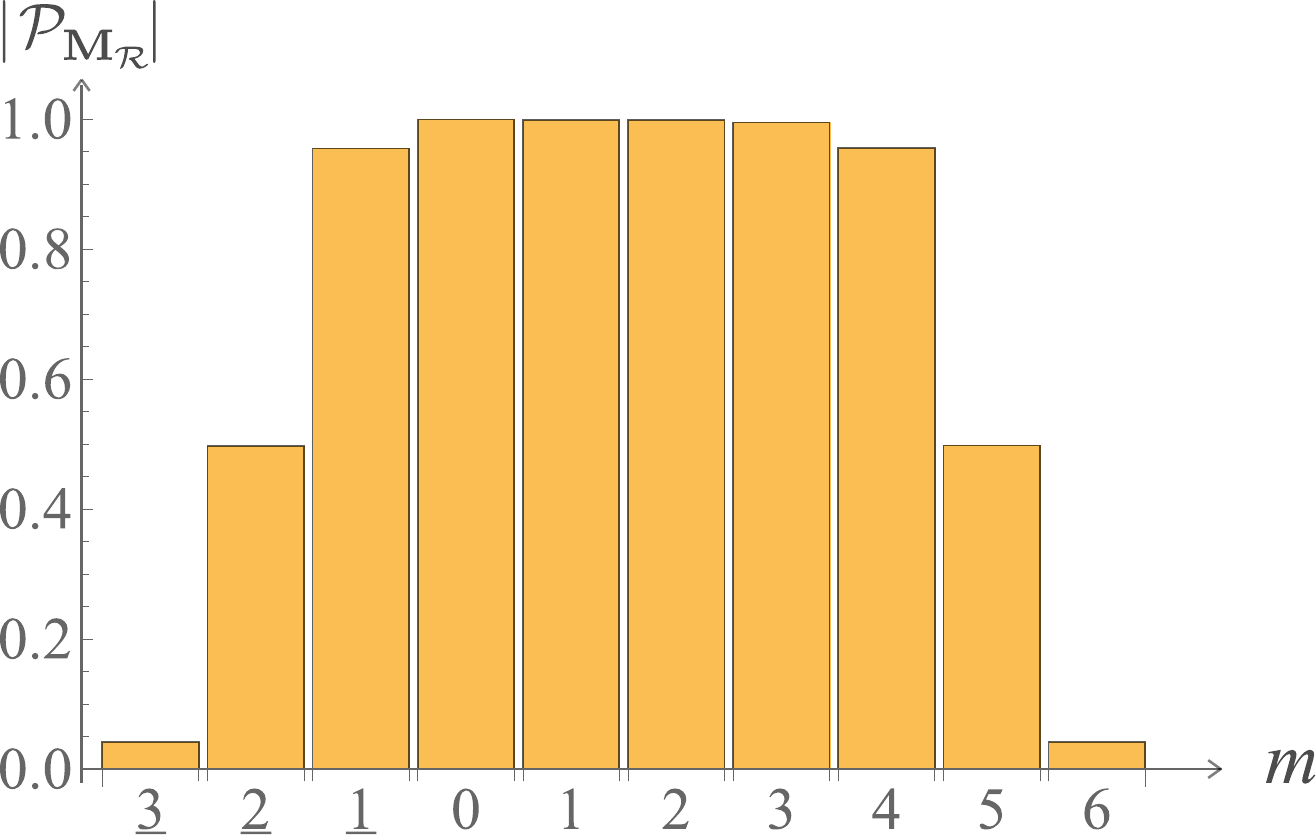}\\
  \vspace{4mm}
(c)\;\includegraphics[width=0.36\textwidth]{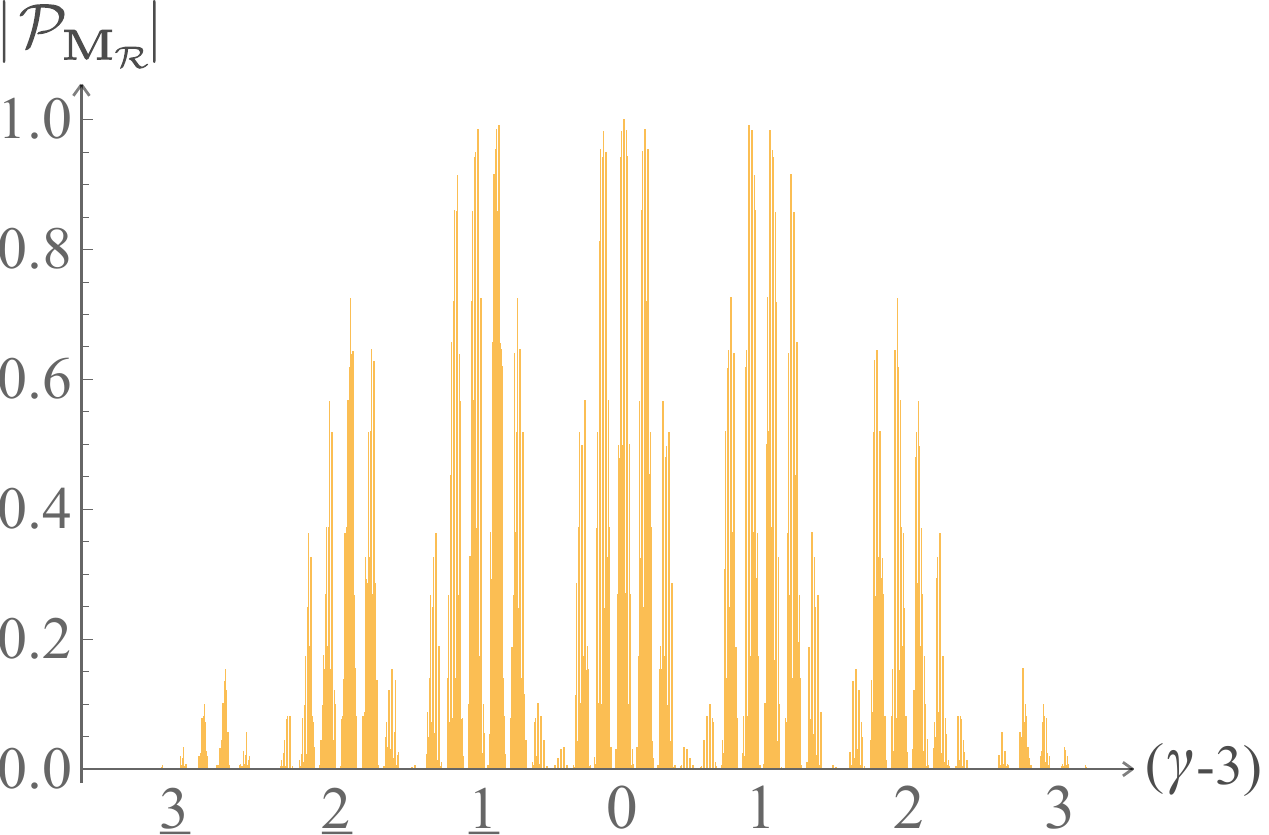}
  \hspace{4mm}
(d)\;\includegraphics[width=0.36\textwidth]{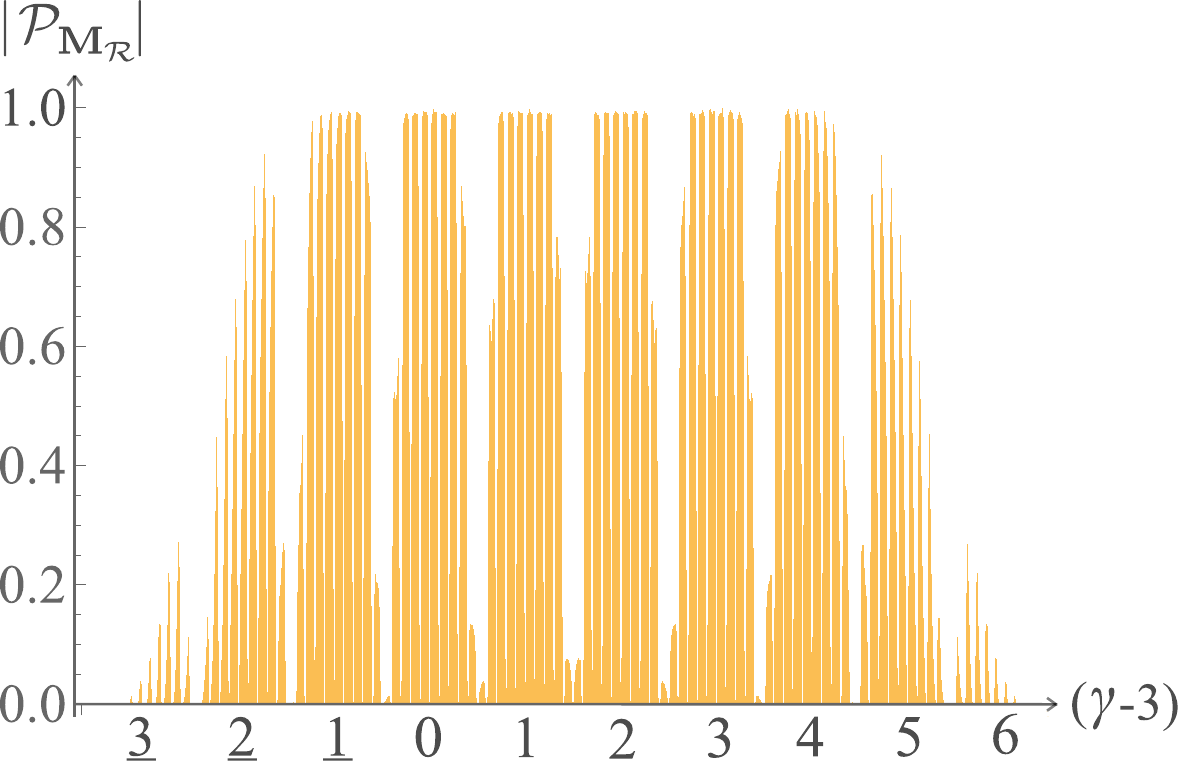}
\end{center}
\caption{\label{fig:RJsymbol}
Relative frequencies $|\Pol_{\Mm_\R}| = f({\Mm_\R})/{d_\R}$ of {\brick s}
$\MmR$ obtained from long-time numerical simulations ($\sim10^7$
iterations, rounded off to two significant  digits) are in agreement with
the exact formulas of \refeq{exactBlock1x1}.
    (a) $\R=[1\!\times\!1]$ domain, \reffig{fig:block2x2}\,(b), ${s=2}$.
The interior alphabet \refeq{2dCatLattAlph} consists of only one letter
$\Ai=\{0\}$, with the corresponding polytope of maximal volume,
$|\Pol(0)|=1$.
    (b) $\R=[1\!\times\!1]$, ${s=7/2}$.
The maximum frequencies are attained for the four letters
\refeq{2dCatLattAlph7} from the interior alphabet $\Ai=\{0,1,2,3\}$.
The {\brick s} $\MmR$
over the $\R=[2\!\times\!2]$ {\spt} domain, \reffig{fig:block2x2}\,(c),
can be represented as numbers  \refeq{block2x2bases=3}
in base ${2s}+3$.
    (c) For $\R=[2\!\times\!2]$,  ${s}=2$,
the only combination of interior symbols that attains the maximum measure
is $\Ssym{11}=\Ssym{12}= \Ssym{21}=\Ssym{22}=0$.
    (d) For $\R=[2\!\times\!2]$, ${s=7/2}$,
the interior alphabet \Ai\ \refeq{2dCatLattAlph7} has $4$ letters, so
there are $4^4$ \brick s that attain the maximum measure.
    }
\end{figure}

The  volumes  of $\Pol_{\Mm}$ evaluated analytically are found to be
consistent with the measure  of a given {\brick} \Mm\ obtained by
numerical simulations of trajectories with random initial conditions.

While in the $\R=[1\!\times\!1]$ case it was possible to plot the single
symbol {\brick} $\Mm$ measures $\Msr(\Ssym{j})$ along a single, integer
$j$ labelled axis, as in \reffig{fig:RJsymbol}\,(a) and (b), the $\R =
[2\!\times\!2]$ has four sites
\(
z\in\{11,12,21,22\}
\,.
\)
A way to map the array \refeq{eq:block2x2} onto a line is to write it as
\beq
\gamma(\MmR) = \gamma_1.\gamma_2 \gamma_3 \gamma_4
\ee{block2x2bases=3}
in base ${2s}+3$, where
$\gamma_k\in\{0,1,\cdots,{2s}+2\}$ are the symbols $\Ssym{ij}$
shifted into nonnegative integers,
\[
(\gamma_1,\gamma_2,\gamma_3,\gamma_4)
= (\Ssym{11}+3,\Ssym{12}+3,\Ssym{21}+3,\Ssym{22}+3)
\,.
\]
Estimates of the corresponding measures $\Msr(\MmR)$  from long-time
numerical Hamiltonian simulations, on a spatially periodic domain of
extent $L=20$, are displayed in this way in \reffig{fig:RJsymbol}\,(c)
and (d), and are in full agreement with the available analytical data. In
particular, whenever all symbols $\Ssym{ij}$ belong to the interior
alphabet, the numerics is consistent with relative frequency
$|\Pol_{\MmR}|=1$.

\section{Families of {\twots}}
\label{sect:twots}

\begin{figure}
\begin{center}
(a) \includegraphics[width=0.45\textwidth]
{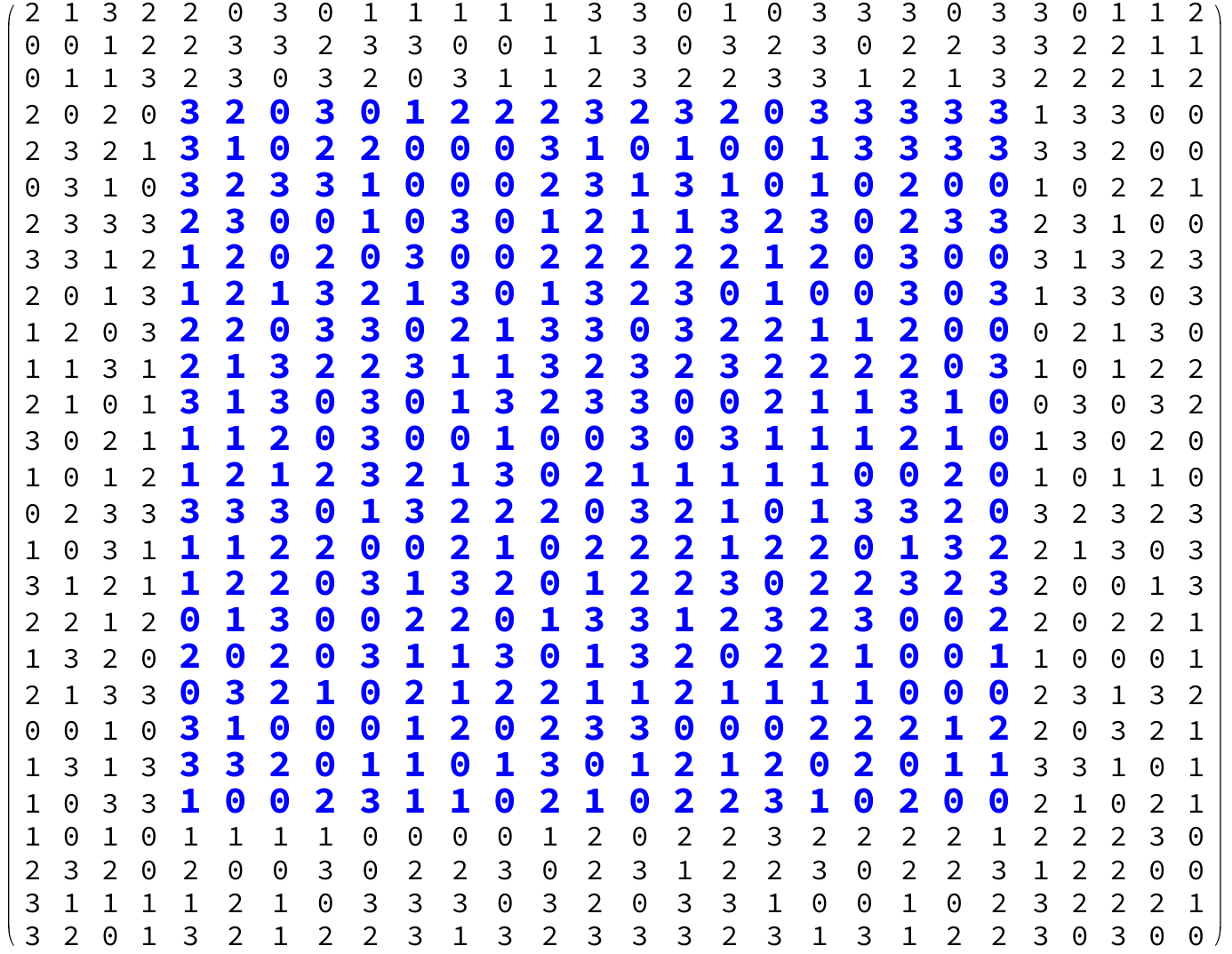} 
(b) \includegraphics[width=0.45\textwidth]
{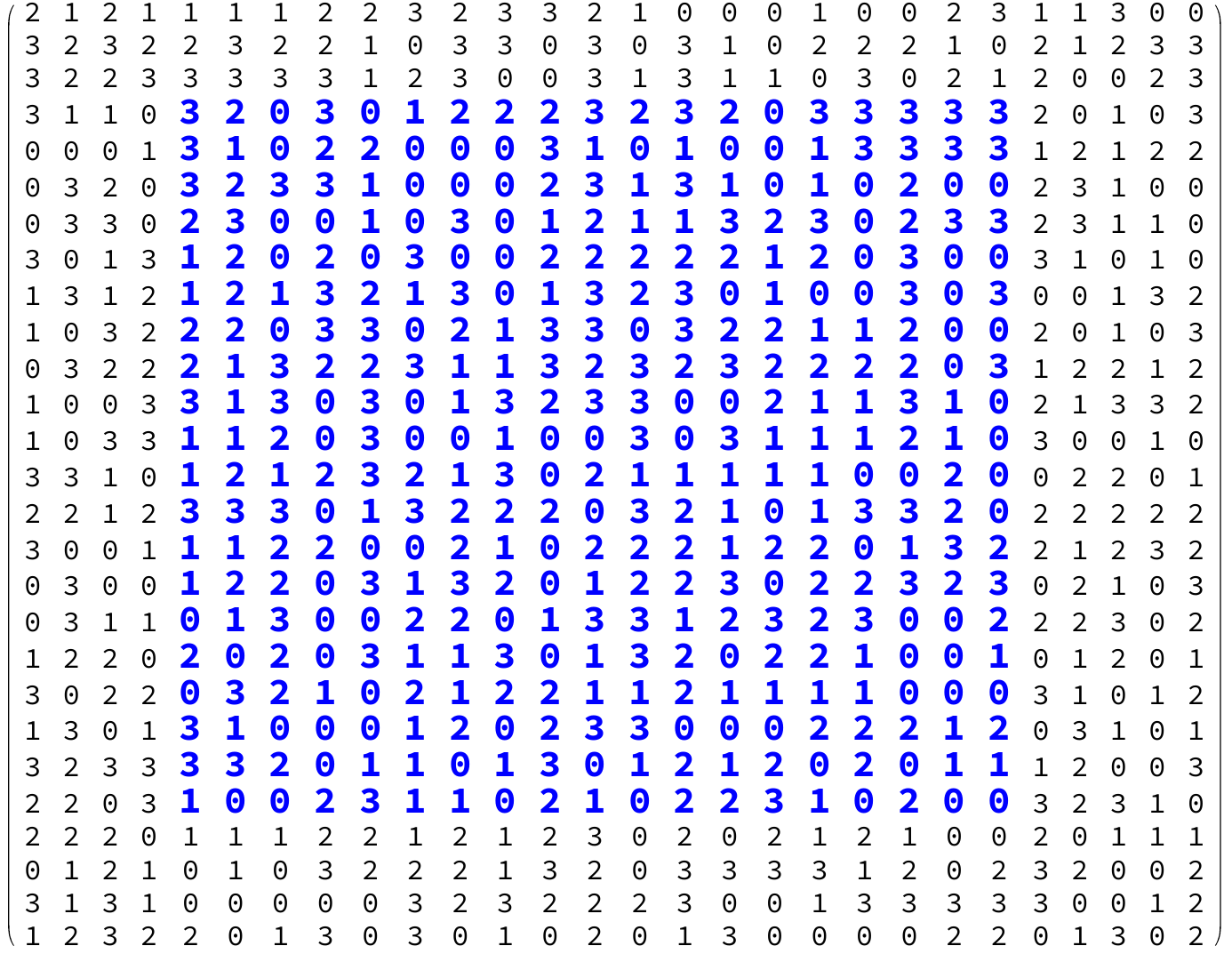}                 
\end{center}
\caption[]{
(Color online)
{\Brick s} $\Mm$, $\Mm'$ encoding  two $[L\!\times\!T] = [28 \times 27]$
{\twots} $\Xx$, $\Xx'$ that shadow each other within the $[19\times 20]$
domain $\R\subset \ZLT$ (bold/blue).
The symbols  over the $\R$ are  drawn randomly from the ${s=7/2}$ interior
alphabet $\Ai=\{0,1,2,3\}$ and are the same for both symbolic {\brick s}
$\Mm$, $\Mm'$. The symbols outside $\R$, also drawn randomly from \Ai,
differ for $\Mm$, $\Mm'$.
}
\label{fig:BGcloseActSymb}  
\end{figure}

\begin{figure} 
\begin{center}
(a) \includegraphics[width=0.45\textwidth]
{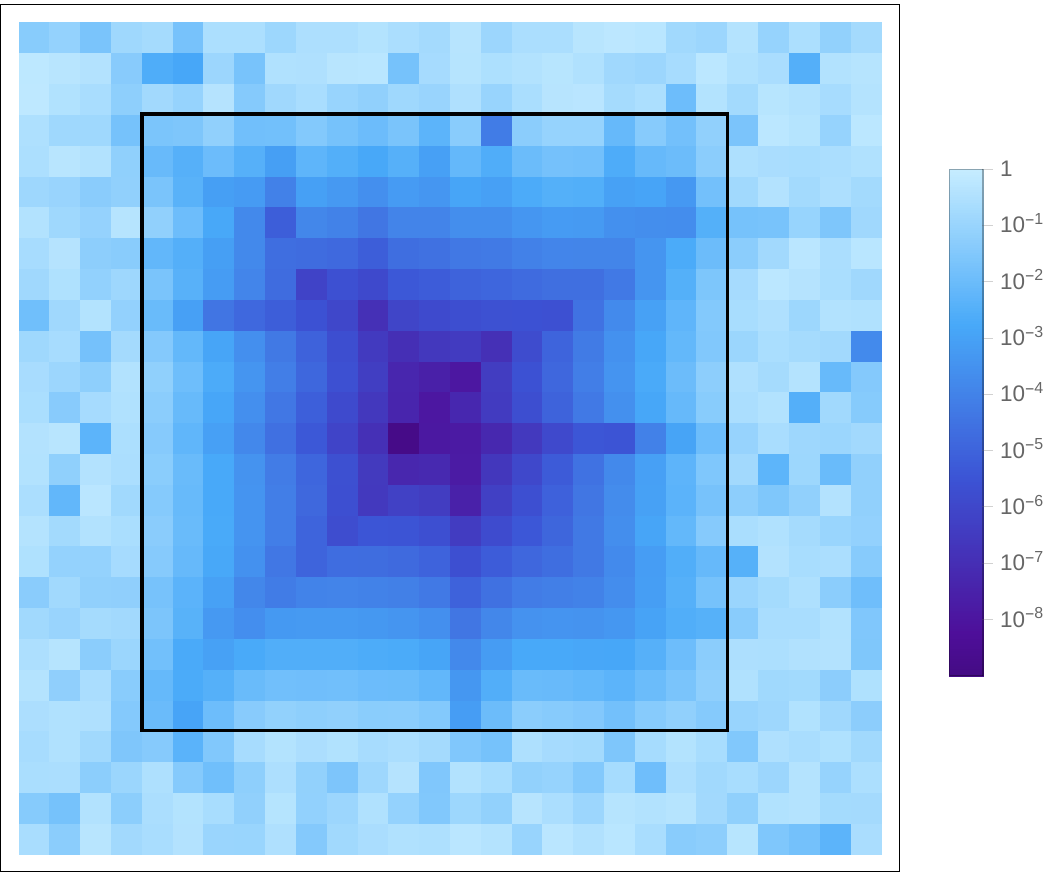} 
(b) \includegraphics[width=0.45\textwidth]
{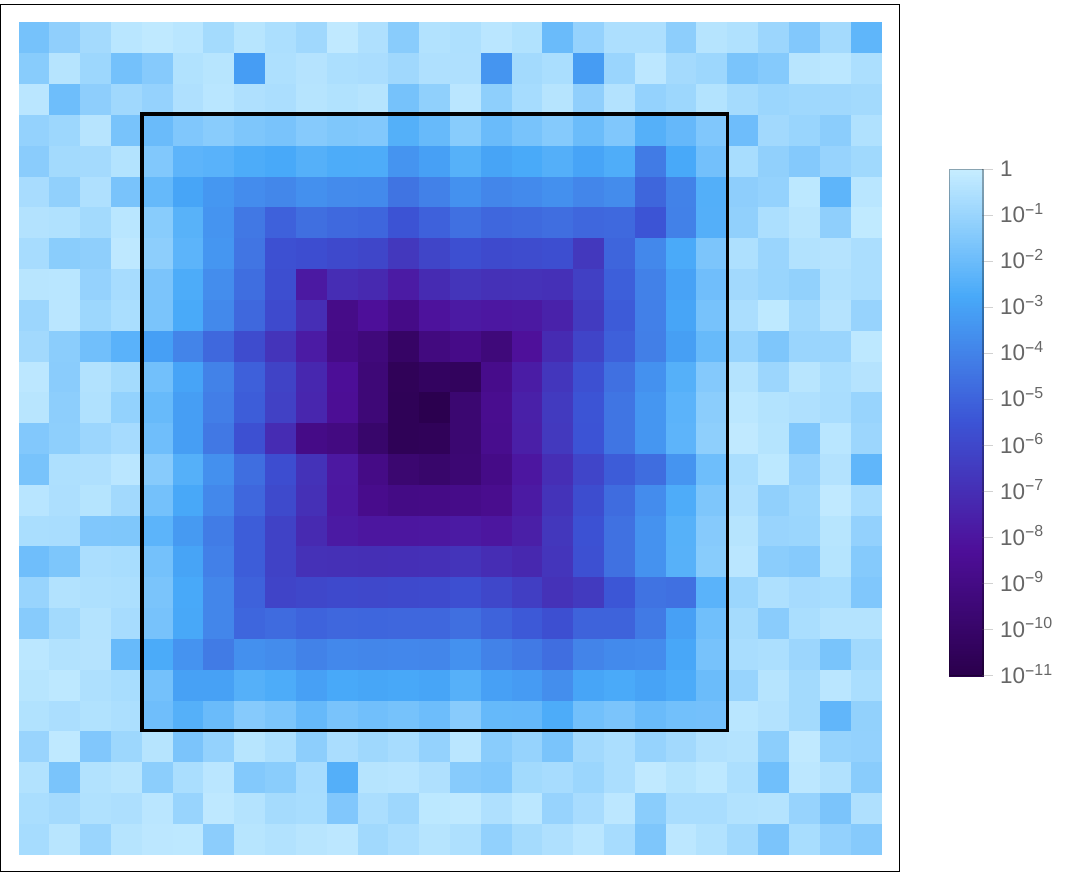} 
\end{center}
\caption[]{\label{fig:AKSLPS12}
(Color online)
The plots of the logarithm of the site-wise distance
\(|x_{z}-x'_{z}|\) of the states $\Xx$, $\Xx'$ of
\reffig{fig:BGcloseActSymb} for
(a)
${s=7/2}$  and
(b) ${s=13/2}$
illustrate the exponential fall-off of the site-wise distances as one
approaches the center of the shared symbol domain ${\R}$.  The
exponential fall-off towards the center of ${\R}$ is faster for $s=13$.
Outside of the shared domain \R\ the distances are of the order $1$.
}
\end{figure}

Linear encoding makes it an easy task to obtain \catlatt\
\refeq{CoupledCats} solutions. Of particular interest are the doubly
periodic solutions (the \spt\ analogs of $d\!=\!1$ cat map \po s), the
\twots\ with periods $L$ and $T$,
\[
          \ssp_{nt} = \ssp_{n+L,t+T},
\qquad    n=1,\cdots ,L,\;\; t=1,\cdots,T
\,,
\]
invariant under  spatial $\Spshift^L$  and temporal $\Tshift^T$ shifts.

Since the interior alphabet \Ai\ corresponds to a full shift
dynamical system,  any $[L\!\times\!T]$ {\brick} of interior symbols
$\Mm= \{\Ssym{z}\in\Ai|z\in \ZLT \}$
over the rectangular \spt\ domain
\[
\ZLT=\{z=(n,t)| n=1,\cdots,L,\;\; t=1,\cdots,T\}
\,,
\] is {\admissible} and generates an \twot\ solution of \refeq{CoupledCats}.
The corresponding {\spt} state  $\Xx =\{\ssp_z, z\in \ZLT\}$ (restricted to the domain $\ZLT$)
is obtained
by taking the inverse of \refeq{CoupledCats}:
\begin{equation}
   \ssp_z=\sum_{z'\in \ZLT} \gp_{zz'} \Ssym{z'}, \qquad    \Ssym{z'}\in \Ai
\,,
\end{equation}
where $\gp_{zz'}$ is the Green's function with periodic boundary
conditions, see \rf{CL18}.
We next use the obtained solutions
to test shadowing  properties of \twots.

\subsection{Partial shadowing}
\label{sect:partShade}

As the first application, we show  in \reffig{fig:BGcloseActSymb} two
$[L\!\times\!T]$ \brick s $\Mm$, $\Mm'$ composed of interior symbols
$m_z\in \Ai$, which coincide within a   rectangular  domain
$\R\subset\ZLT$. In  \reffig{fig:AKSLPS12}, we  show the distances between
the corresponding {\spt} states $\Xx$, $\Xx'$. In agreement with the
results of \refsect{sect:CCMmeasBrick}, the distances between
$\ssp_z\in\Xx$ and $\ssp_z'\in\Xx'$ shrink exponentially as $z$
approaches the center of the domain \R. In other words, \Xx\ and $\Xx'$
shadow each other within the domain \R.

\begin{figure}
\begin{center}
 \includegraphics[width=0.48\textwidth]{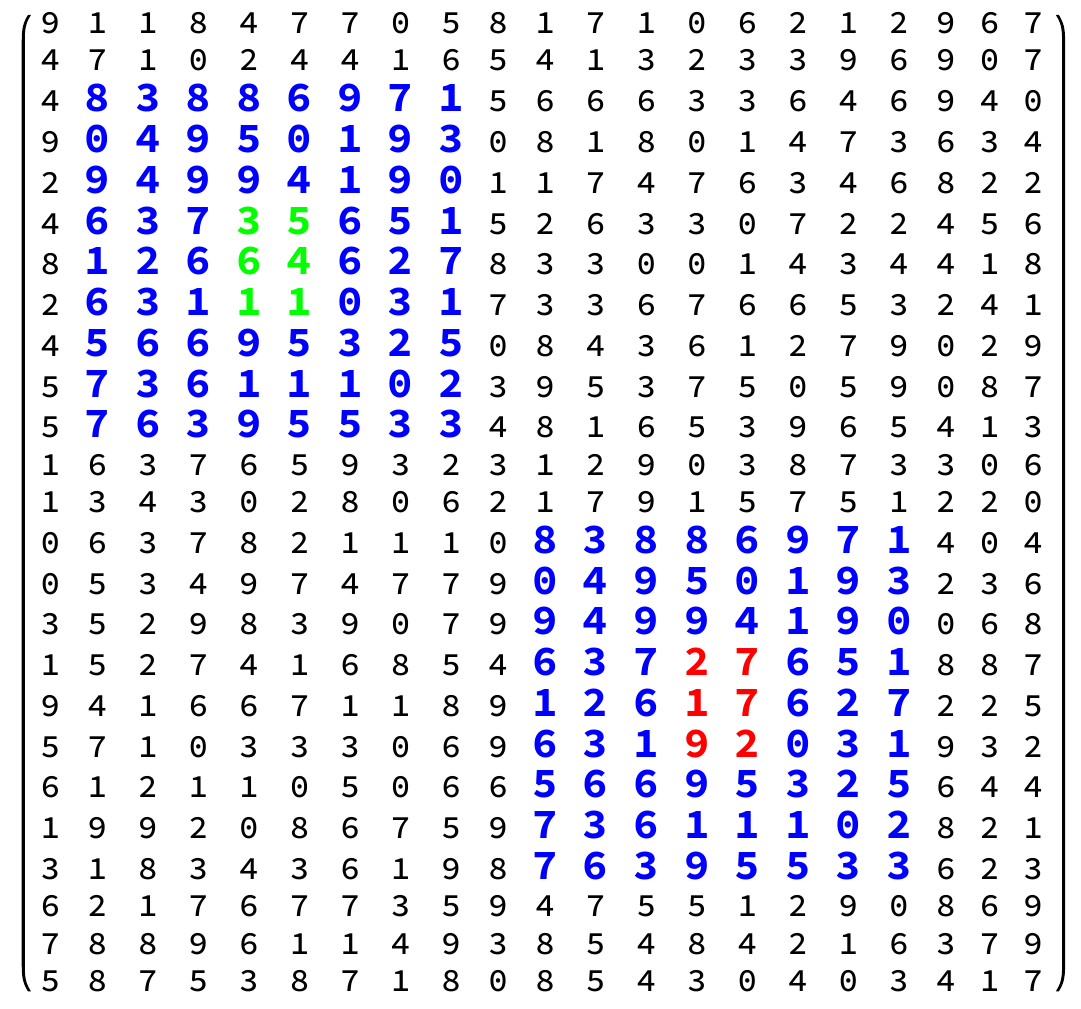}
 \includegraphics[width=0.48\textwidth]{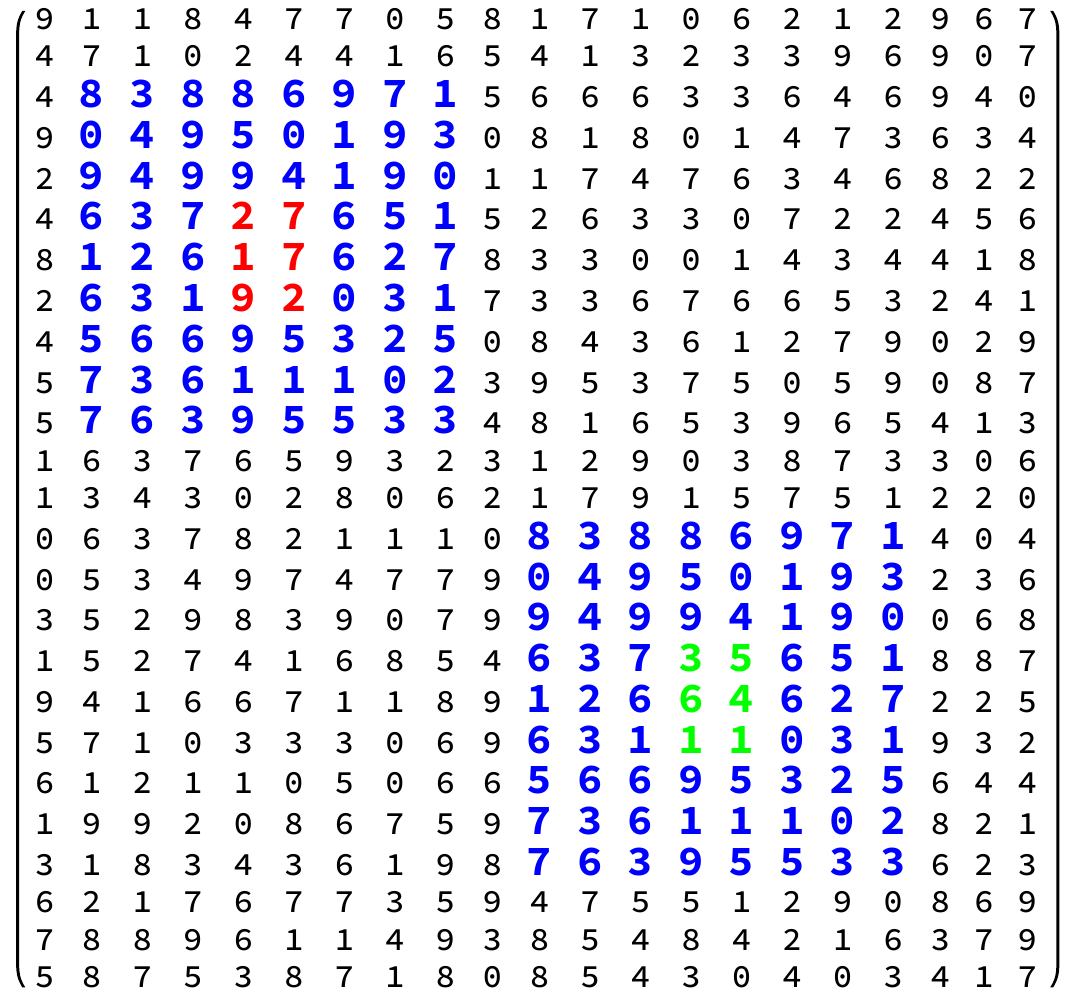}
\end{center}
\caption[]{\label{fig:AKSs13TwoBlock}
(Color online)
Symbol \brick s $\Mm_1$, $\Mm_2$ of two $[25\times25]$ \twots\ $\Xx_1$,
$\Xx_2$, with
annular
(encounter) domains $\R_1$ and $\R_2$ indicated
in blue (bold).
The symbols are drawn randomly from the interior alphabet $\Ai$ for
${s=13/2}$. The two \brick s $\Mm_1$, $\Mm_2$ are related by the permutation
of symbol \brick s over  the interior domains $A_1$, $A_2$ (red and green,
respectively). Any $[3\times3]$ symbol {\brick} $\Mm^{[3\times 3]}$
appears the same number of times in both $\Mm_1$ and $\Mm_2$, and the two
\twots\ $\Xx_1$ and $\Xx_2$ shadow each other at every point.
    }
\end{figure}

\begin{figure}
\begin{center}
   \includegraphics[width=0.48\textwidth]
{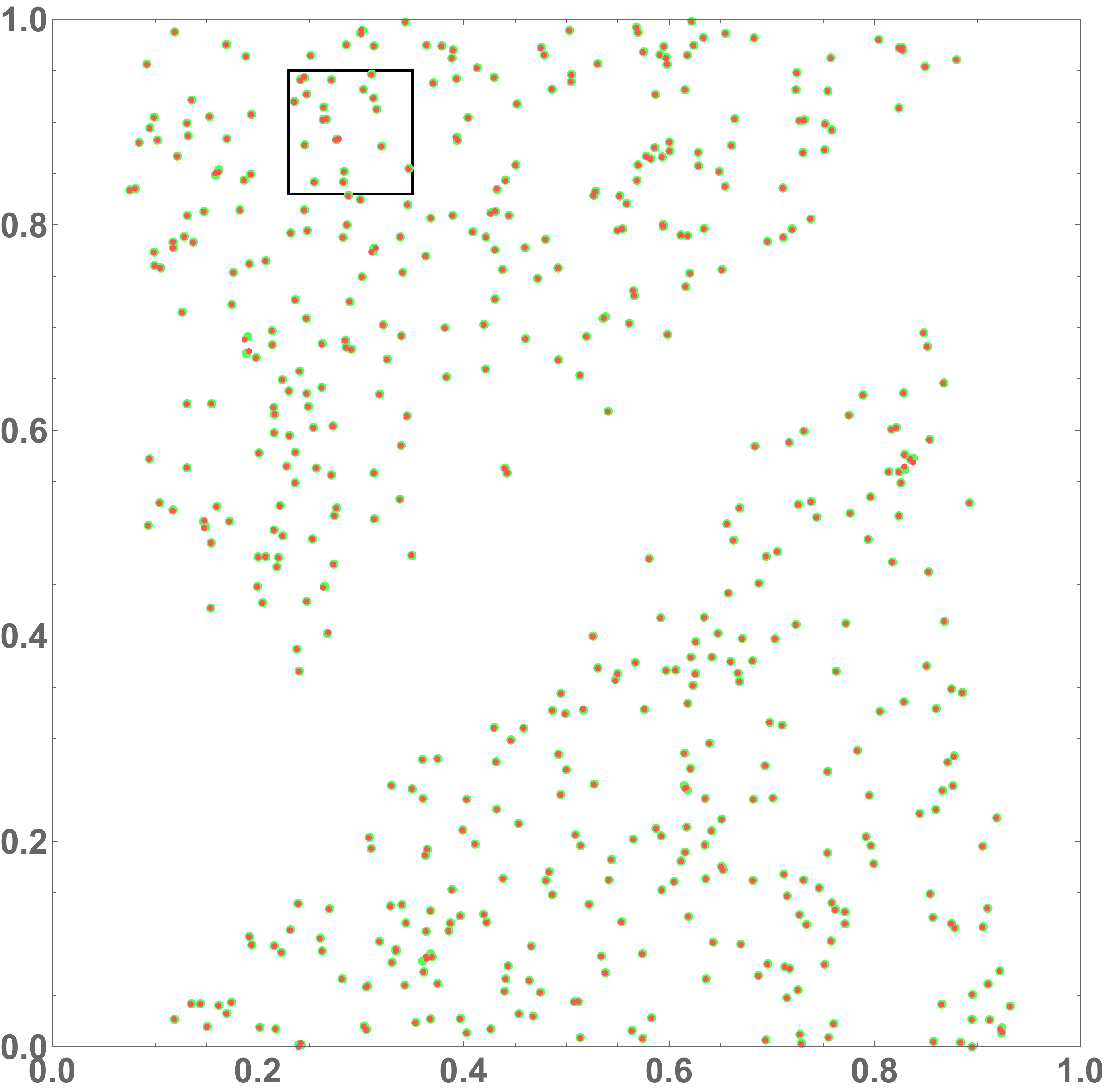}
   \includegraphics[width=0.48\textwidth]
{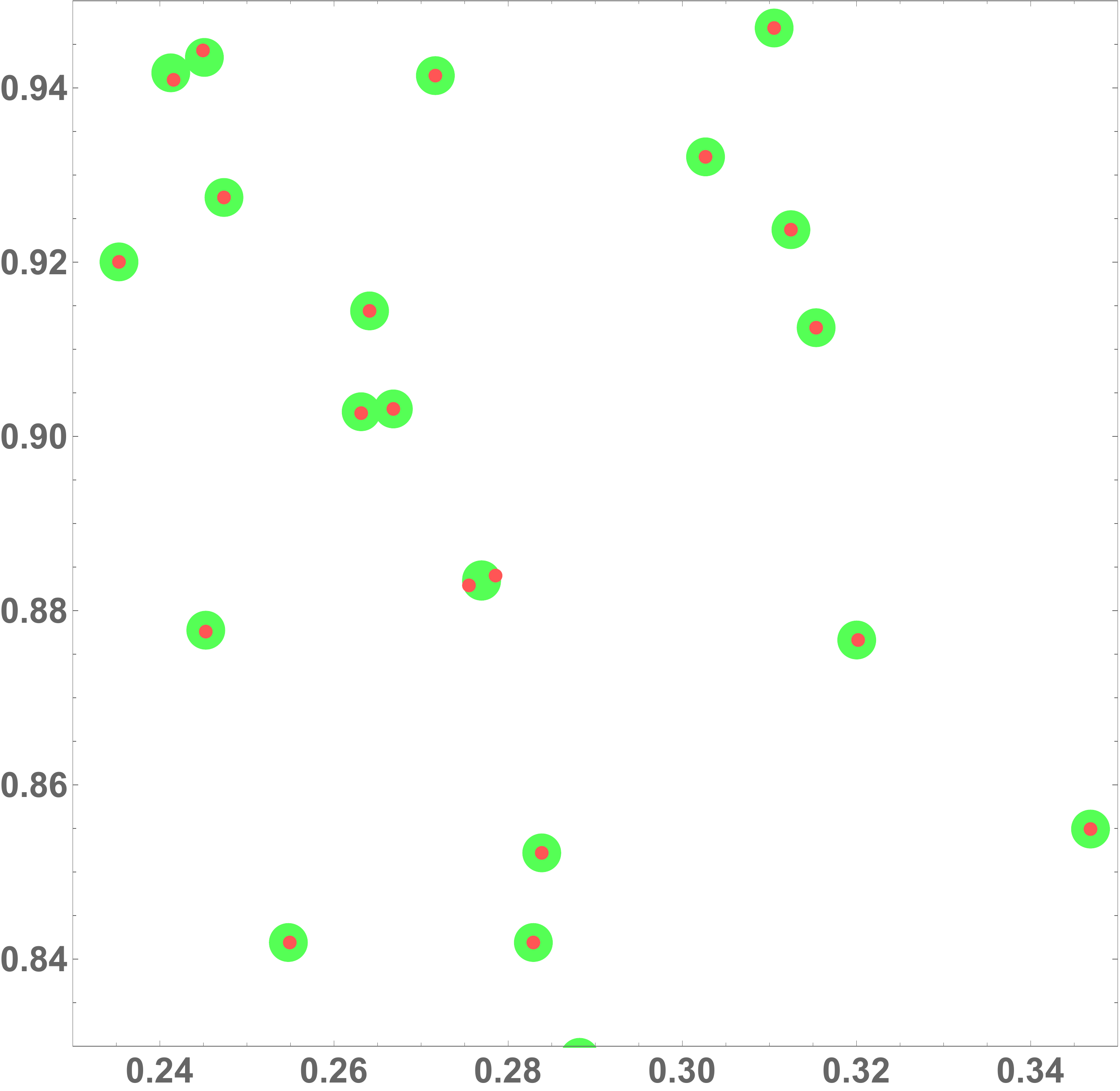}
\end{center}
\caption[]{\label{fig:AKSs13TwoBlocks}
(Color online)
(left)
Hamiltonian coordinate-\-momentum representation of the two \twots\
$\Xx_1, \Xx_2$ of \reffig{fig:AKSs13TwoBlock}. This Hamiltonian
representation is explained in \refappe{sect:HamiltonCatLatt}.
(right)
An enlargement of the square in (left).  The small (red)  circles indicate
points  $(q^{(1)}_z, p^{(1)}_z)$  of   $\Xx_1$. The large  (green)
circles correspond to the  points $(q^{(2)}_z, p^{(2)}_z)$  of   $\Xx_2$.
The solutions pair nearly perfectly, except for the points
from the encounter domains $z\in \R_1\cup \R_2$,  where some small
deviations can be  observed.
    }
\end{figure}

\subsection{Full shadowing}
\label{sect:fullShade}

\begin{figure}
    \begin{center}
\includegraphics[width=0.48\textwidth]{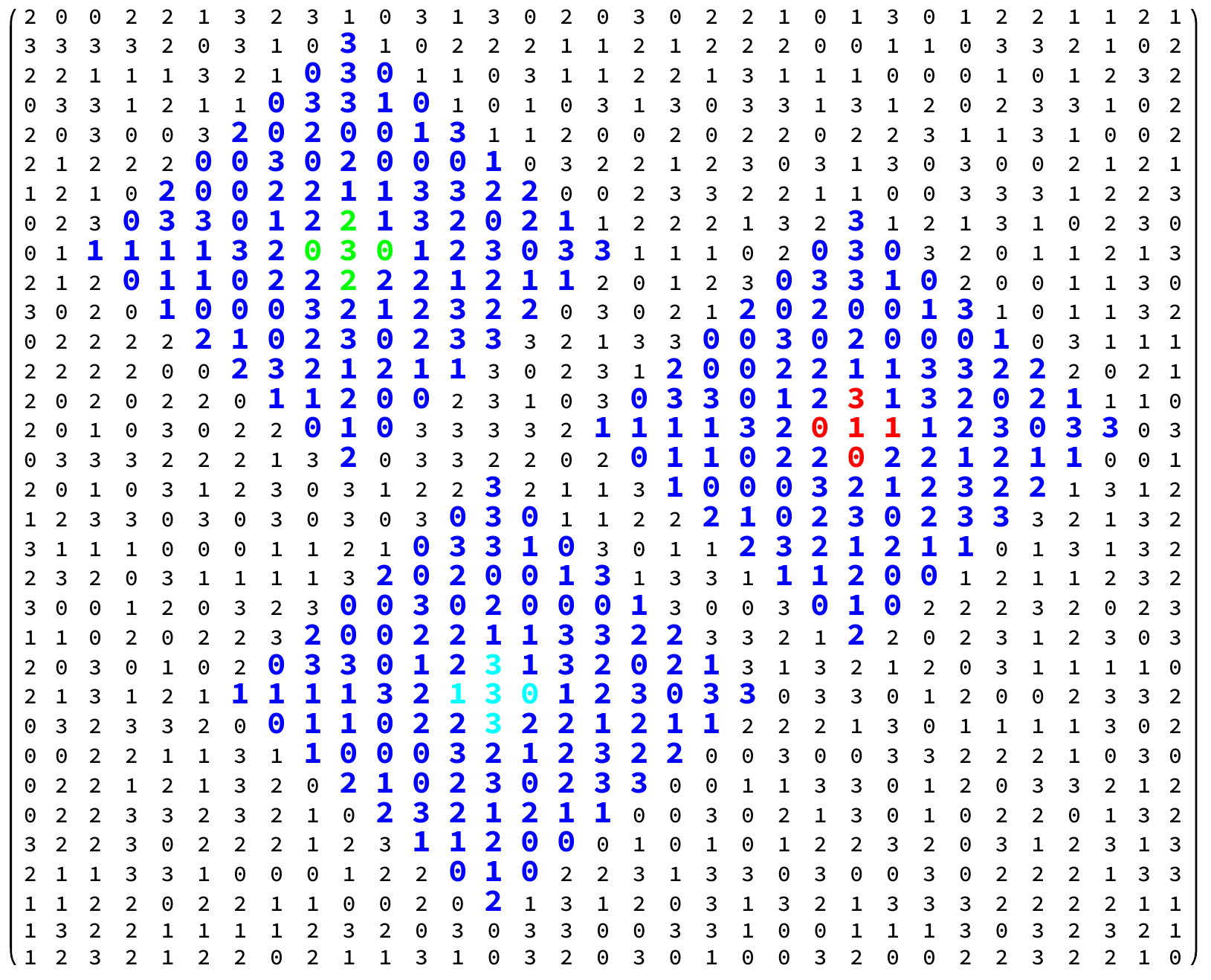}
\includegraphics[width=0.48\textwidth]{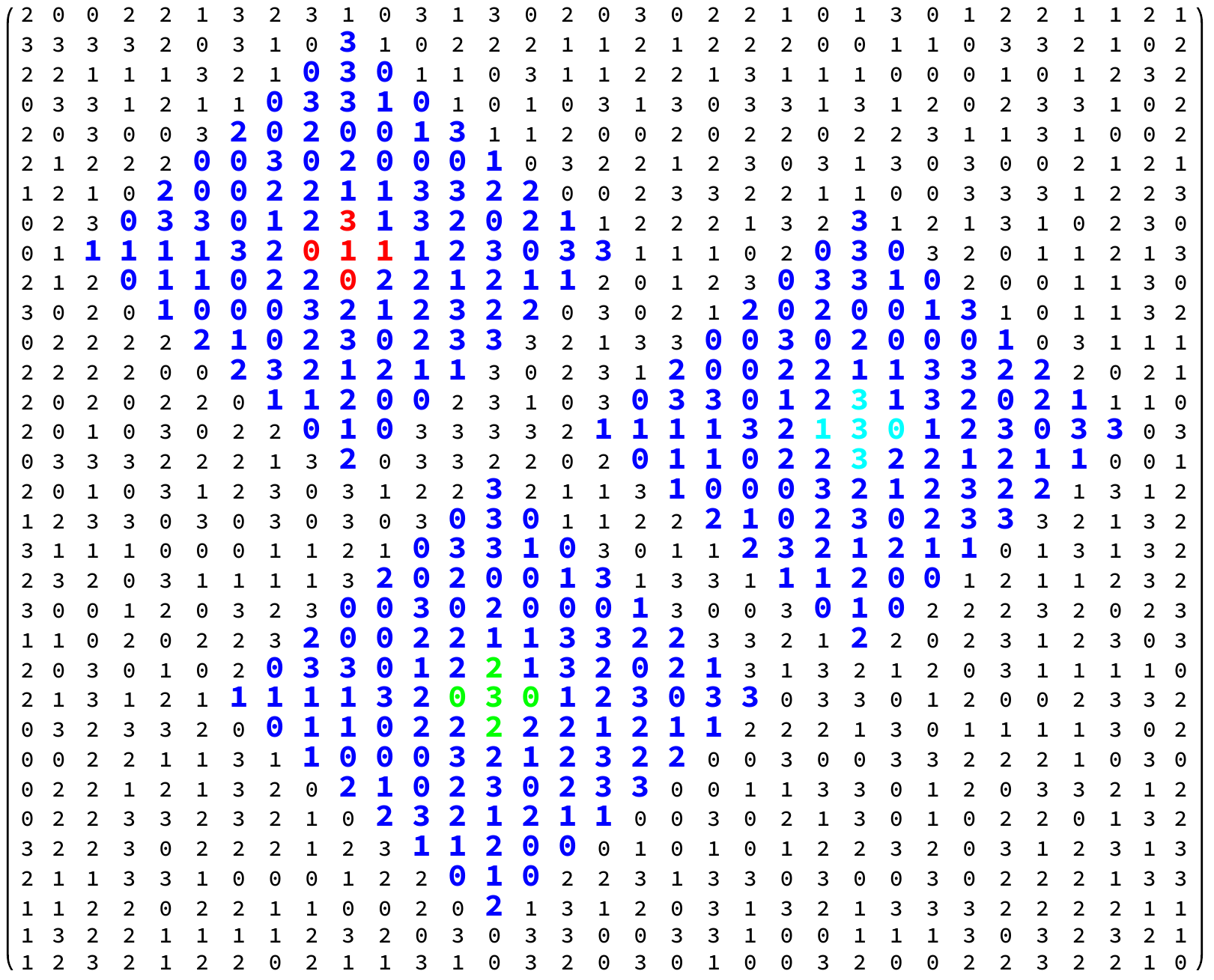}~\hspace{0.3cm}%
\includegraphics[width=0.48\textwidth]{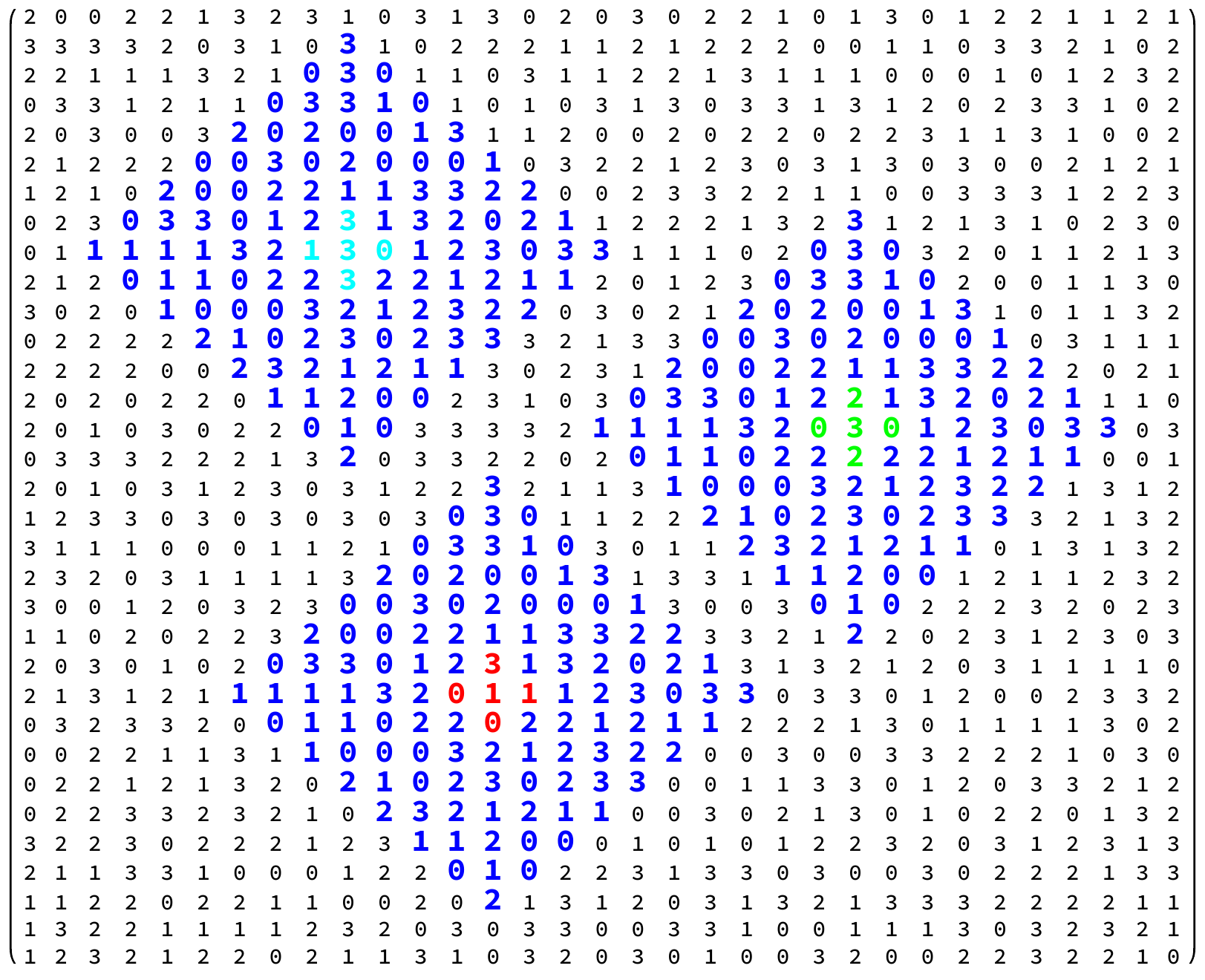}
    \end{center}
\caption[]{
(Color online)
Symbol \brick s $\Mm_1, \Mm_2, \Mm_3$ of three $[33\times33]$ \twots\
$\Xx_1, \Xx_2, \Xx_3$, with the 3-encounter diamond-shaped
domains $\R_1, \R_2, \R_3$  indicated in blue (bold). The symbols are
drawn randomly from the interior alphabet $\Ai$ for ${s=7/2}$.
$\Mm_1\to\Mm_2\to\Mm_3$ are related by the cyclic permutation of the
symbol \brick s  within  the interior domains (red, green, light  blue).
Any distinct $[4\!\times\!4]$ symbol {\brick} $\Mm^{[4\times 4]}$ appears
the same number of times in each \twot\ (or not at all).
        }
\label{fig:AKSs13TwoBlock2}
\end{figure}

\begin{figure}
\begin{center}
\includegraphics[height=0.48\textwidth]
{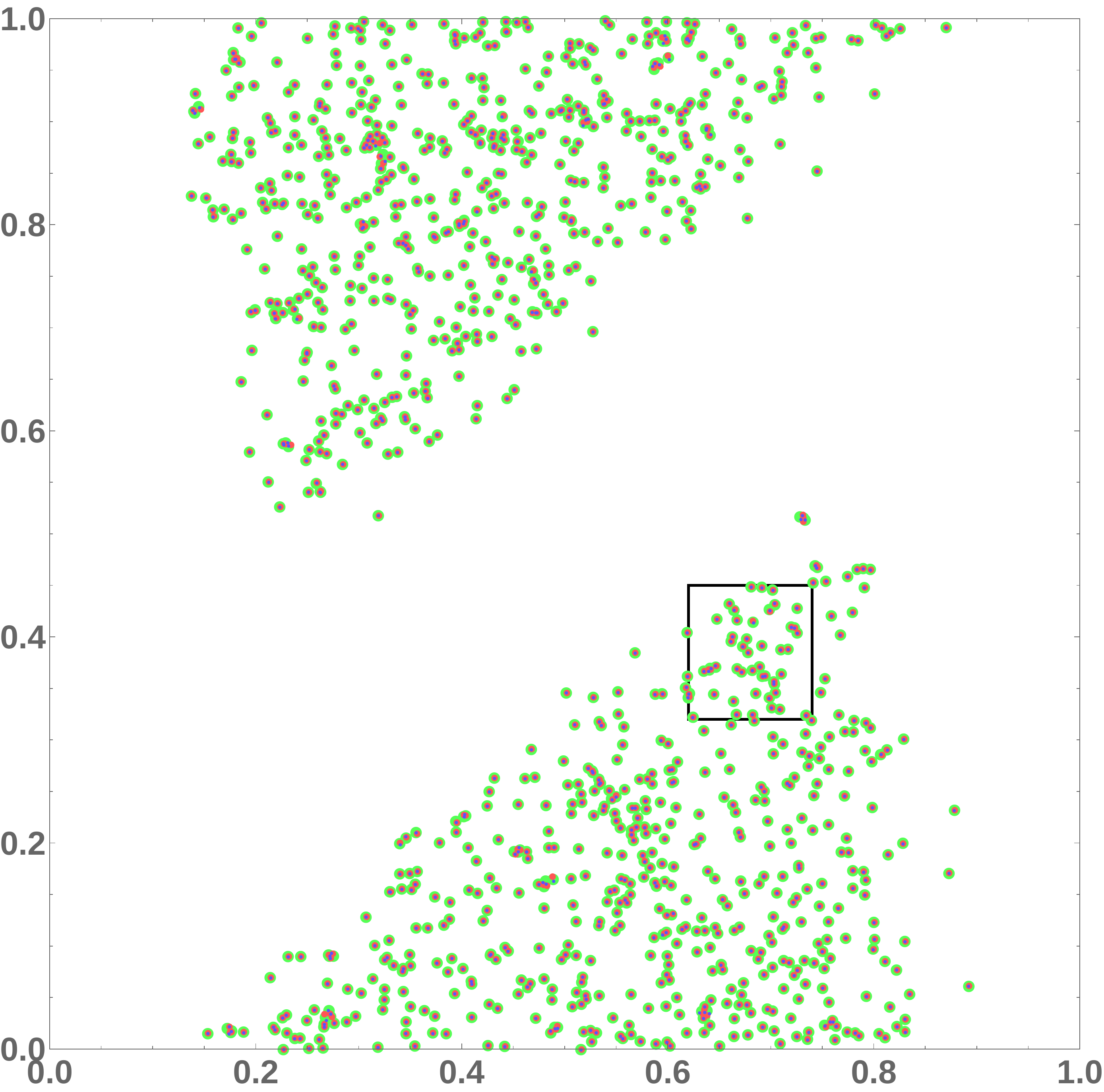}
   \includegraphics[height=0.48\textwidth]
{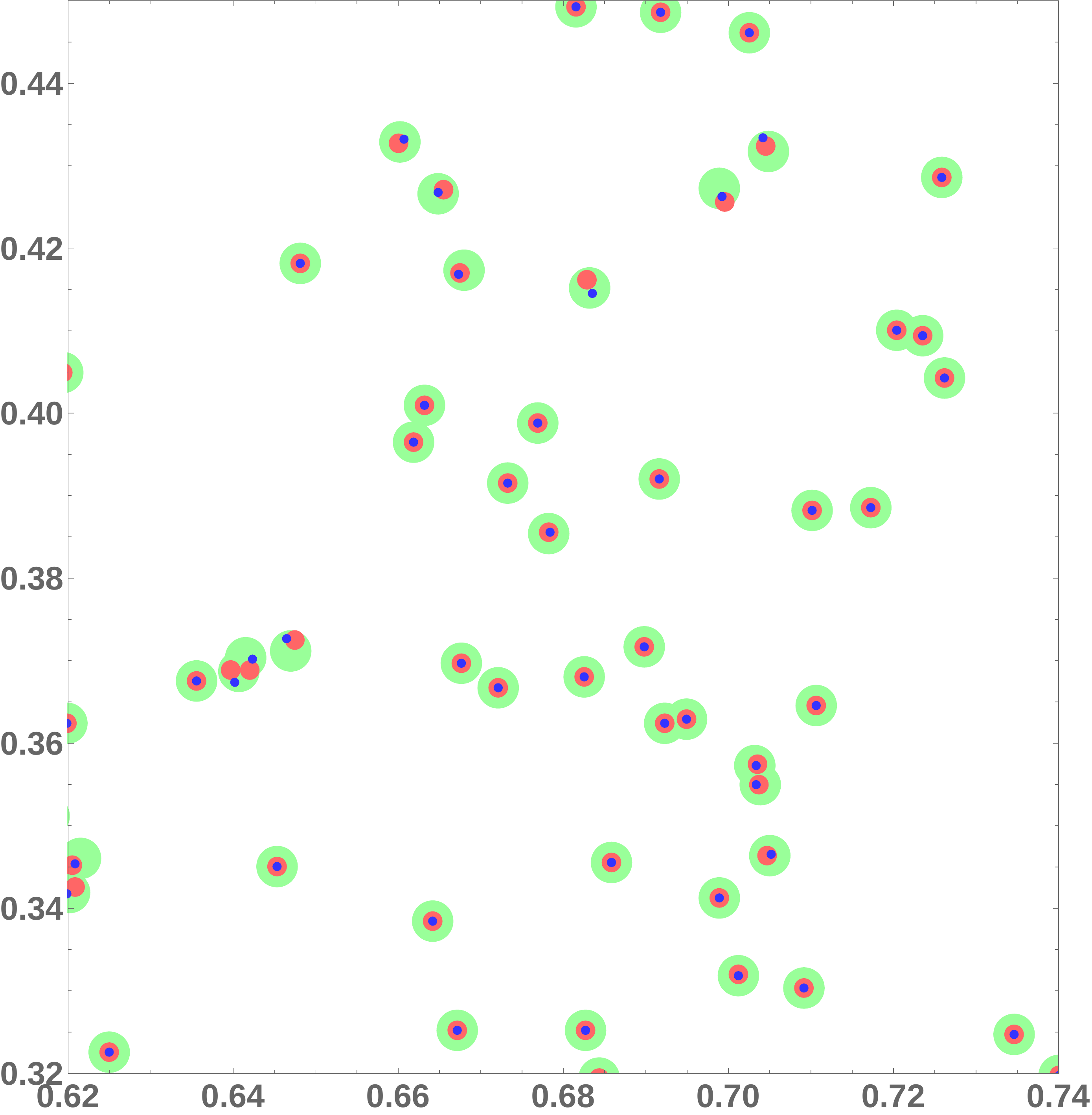}
\end{center}
\caption[]{\label{fig:AKSs13TwoBlocks1}
(Color online)
(left) Hamiltonian coordinate-\-momentum representation $(q^{(i)}_z,p^{(i)}_z)$,
$i=1,2,3$
of the three \twots\  $\Xx_i$,
$i=1,2,3$  of \reffig{fig:AKSs13TwoBlock2},
shown as centers of the green, red and blue circles. (right)
An enlargement of the (left) \statesp.  Note that the three \twots\ $\Xx_i$
form almost perfect triplets except for the points
from the encounter domains, $z\in \R_1\cup\R_2\cup\R_3$  where some
deviations can be  observed.
        }
\end{figure}

The \twots\ $\Xx$ and $\Xx'$ of the above example shadow each other only
partially - outside of the subdomain $\R$ their points are not paired. On the
other hand, it turns out to be possible to find different \twots\
solutions of \refeq{CoupledCats} which shadow each other at every point
of $\ZLT$. Such solutions, referred as partners,  play an important role
in the semiclassical treatment of the corresponding quantum {problem}, since
their action differences are small, see
\refrefs{SieRic01,MuHeBrHaAl04,GutOsi13a,GutOsi13,GutOsi15}. We briefly
recall here the construction of partner solutions in  $d=2$ setting,
following \refref{GutOsi15}.

Let $\{\R_1, \dots, \R_{m}\},\; \R_j\subset \R$, be ${m}$ non-overlapping
domains obtained  by spacetime shifts $(n,t)\to(n+n_i,t+t_i,)$,
$i=2,\dots {m}$ of a given domain $\R_1$,
\[
\R_i=\Spshift^{n_i}\Tshift^{t_i}\,\R_1
\,,\qquad
\R_i\cap\R_j=\emptyset \mbox{ for } i\neq j
\,.
\]
Examples of such domains (referred to as ``${m}$-encounters'')
are given in \reffigs{fig:AKSs13TwoBlock}{fig:AKSs13TwoBlock2}. Assume
that the domain $\R_1$  has  an annular shape  with a non-empty interior
$A_1$, and an exterior $C_1=\ZLT \setminus{A}_1\cup\R_1$. We define the
width  $\ell$ of ``annulus'' $\R_1$   by determining the largest
$[\ell\!\times\!\ell]$ square domain $\R^{[\ell\times\ell]}$ that can fit
into $\R_1$ i.e., $\R^{[\ell\times\ell]}$ that has no simultaneous
intersections with both the hole and the exterior. In other words, $\ell$
is the maximum integer such that either
$A_1\cap\R^{[\ell\times \ell]}=\emptyset$, or
$C_1\cap\R^{[\ell\times \ell]}=\emptyset$,
for any translation of $\R^{[\ell\times \ell]}$. All encounter domains
$\R_i$ have the same width $\ell$, as they are translations of each
other.

Now,  let \brick\
\(
\Mm\equiv\Mm_1=\{\Ssym{z}\in\Ai|z\in\R\}
\)
be a  $[L\!\times\!T]$ symbolic representation of a \catlatt\
\refeq{CoupledCats} \twot\ state $\Xx\equiv\Xx_1$, such that it contains the same
{\brick} of symbols over each of the   subdomains
$\R_1,\dots,\R_{m}$,
\[
\Mm_{\R_1}=\Mm_{\R_2}= \dots = \Mm_{\R_{m}}
\,,
\]
where $\Mm_{\R_i}$ stands for the restriction of $\Mm$ to the domain
$\R_i$. Provided that the $m$  interior domains $A_1,\dots A_m$ have
different symbolic content, $\Mm_{A_i}\neq\Mm_{A_j}$ for $i\neq{j}$, we
can generate $m!$ distinct {\brick s} $\Mm_a, a=1,\dots m!$ by permuting
symbol \brick s $\Mm_{A_i}$ over the {domains} $A_1,\dots A_m$. This leads to
${m}!$ distinct \twots\
$\Xx_1, \dots \Xx_{{m}!}$
whose symbol {\brick s} $\Mm_{1}, \dots \Mm_{{m}!}$ share the following
property: each distinct $[\ell\!\times\!\ell]$ square symbol {\brick}
appears one and the same number of times in all $\Mm_{a}, \, a=1,\dots,
m!$, with  $\ell$ being  the width of the encounter domains $\R_i$. In
other words, if a   $[\ell\!\times\!\ell]$ symbol {\brick}
$\Mm^{[\ell\times\ell]}$
appears in $\Mm_1$ a number of times (or zero times),  it must appear the
same number of times in  the encoding $\Mm_{a}$ of any other \twot\
$\Xx_a, a\neq1$.

By the shadowing property, this in turn implies that all $\Xx_1, \dots
\Xx_{{m}!}$  pass through approximately the same points of the \statesp\
but in a different {\spt} `order'. The degree of their closeness is
controlled by the parameter $\ell$. The larger $\ell$ is, the closer two
different $\Xx_a$,  $\Xx_b$ come  to each other in the \statesp. In
\reffigs{fig:AKSs13TwoBlocks}{fig:AKSs13TwoBlocks1} we illustrate these
pairings for families of \twot\ solutions which symbolic representations
are shown in \reffigs{fig:AKSs13TwoBlock}{fig:AKSs13TwoBlock2},
respectively.

\section{Summary and discussion}
\label{sect:summary}

In this paper, we have analyzed the {\catlatt} \refeq{LinearConn}
linear encoding. We now summarize our main findings.

The finite alphabet of symbols \A\ encoding system's dynamics has been shown
to split into interior \Ai\ and exterior \Ae\ parts, where only symbolic
\brick s containing external symbols require non-trivial grammar
rules. \Brick s  composed of only interior symbols are {\admissible}
and   attain one and the same measure for a given {domain \R}.
Furthermore, the measure of a general {\brick} factorizes into product of a
constant factor $d_\R$ and the geometric one  $|\Pol_\R|$ which  can be
interpreted as a volume of certain type of polytope  in the Euclidean space
whose dimension is determined by the length of the boundary of \R. While
$d_\R$ is fixed by  \R, $|\Pol_\R|$ depends on the symbolic content and attains
maximum value $1$ for \brick s  of interior symbols. In addition, it has been
shown that a  local  {\brick}  of symbols determines approximate positions of
the corresponding \statesp\ points within an error
decreasing exponentially with  the size of the {\brick}.

The  number  of letters, ${2s-3}$, of the interior alphabet  \Ai\ grows
linearly with the increasing stretching parameter $s$, while the exterior
alphabet \Ae\ always consists of the 6 letters. In
the limit of large $s$,  the \brick s of a finite size affected by the
exterior letter pruning rules can be neglected,  as their total measure
tends to zero.   This implies that as $s$ grows, the dynamics of the
{\catlatt} map approaches a Bernoulli process. In particular, by
\refeq{SpatiotemporalEntropy} the ratio between the metric   entropy
$h_\Msr$ of the {\catlatt} and the topological entropy $\log ({2s-3})$  of
the full $({2s\!-\!3})$-shift   converges  to $1$, as
${s}\to\infty$.

As an application of these results we have constructed several examples
of partner \twots\ composed of interior symbols which visit the same
regions in the \statesp, but in different {\spt} ``orders''. Such
\twots\ shadowing each other are expected to play an important role in
the semiclassical treatment of the corresponding quantum models.

\subsection{Discussion and future directions}

Remarkably, as far as the linear encoding is concerned, the above results
hold both for the {cat map} and its coupled lattice generalization. In
both cases, the  proofs rely only upon ellipticity of the operator $\Box$
and the linearity of the equations. It is very plausible  that the same
results hold for the lattices $\integers^d$  of an arbitrary dimension $d$.
Indeed, the companion paper\rf{CL18} makes claim that its
\po\ theory formulation applies to \catlatt\ in any dimension.
Furthermore, the restriction to
the integer valued matrices in the definitions of maps appears unnecessary. Cat
map is a smooth version of the sawtooth map, defined by the
same equation \refeq{OneCat}, but for a real (not necessarily integer) value
of $s$. The linear encoding for the {saw map} has been analyzed in
\refref{PerViv} and its extension  to a coupled $\integers^d$ model along the
lines of the present paper seems to be straightforward. Also,  in the current
paper we stuck to  the Laplacian  form of $\Box$.   Again this seems to be
too restrictive and extension to other elliptic  operators of higher order
should be possible. Such operators necessarily appear within the models
with higher range of interactions.

A physically necessary extension of  current setting would be addition of an
external periodic potential ${V}$ to \refeq{LinearConn}, rendering this a
nonlinear problem,
 \begin{equation}
 (\Box -{d(s-2)}+ {V}'(\ssp_z)) \ssp_z = \Ssym{z}, \qquad z\in \integers^{d}
 \,. \label{LinearConnPerturbed}
\end{equation}
As long as the perturbation ${V}$ is sufficiently weak, this lattice map can
be conjugated to the linear {\catlatt}, with ${V}=0$.
This approach has been used in \refref{GutOsi15} to construct partner
{\twots} for perturbed cat map lattices.
On the other hand, for a sufficiently strong perturbation, such a conjugation
to linear system is no longer possible. Also, let us note that the lattice
models like  \refeq{LinearConnPerturbed} can be seen  as discretized versions
of PDEs arising from the Hamiltonian field theories.
In this respect, it would be of  interest  to study whether our results  can
be extended to the continuous, PDE setting.

Finally, there are many intriguing open questions about the quantum
properties of {\catlatt}s. Starting with the work of Hannay and
Berry\rf{HanBer80}, quantum cat maps have provided deep insights into
``single particle'' quantum chaos, see e.g. \refref{Keating91, Dana02,
Creagh94, ValSar99, KurRud00}. In the same spirit, quantization of the
{\catlatt} could serve as an inspiring model  for ``many-body'' quantum
chaos\rf{GutOsi15, AWGBG16, AWGBG17, AGBWG18}. One of its  most striking
features is the  dynamical self-duality, with spatial and temporal
evolutions being completely equivalent. Due to their remarkable features,
quantum lattice models with   such properties are currently attracting
much interest. Such self-dual models  exhibit properties of
``maximally-chaotic''  quantum systems\rf{BeKoPr18, AWGG16, BWAGG19}, and
yet turn out to be exactly solvable at the thermodynamical limit, with
correlations between local operators\rf{BeKoPr19-4}, the local-operator
entanglement\rf{BeKoPr19-9, GBAWG20, ClaLam20}, and  the time evolution of
the entanglement entropies\rf{BeKoPr19-1, GopLam19} exactly calculable.
It would be interesting to investigate whether similar results can be
established for the {\catlatt} map.

\ack
Work of B.~G. and P.~C. was supported by the family of late G. Robinson,
Jr..  B.~G. acknowledges  support from  the Israel Science
Foundation through grant No.~2089/19.

\appendix

\section{Lattice Green's functions}
\label{sect:Green}

\subsection{Green's  function for 1\dmn\ lattice}

Consider the {cat map} equation \refeq{OneCat} with a delta function
source term
\begin{equation}
 (-\Box+s-2) \gd_{ij}=\delta_{ij}, \qquad i,j \in \integers^1
\,.
\label{GreenFun0}
\end{equation}
The corresponding free, infinite lattice Green's function $\gd_{ij}=\gd_{i-j, 0}$
is given by\rf{varcyc,PerViv}
\beq
\gd_{t 0}=\frac{1}{\pi}\int_{0}^{\pi}\frac{\cos(tx)}{s-2\cos x } dx
   = \ExpaEig^{-|t|}/(\ExpaEig-\ExpaEig^{-1})
\,,
\ee{GreenFun00}
with $s=\ExpaEig+\ExpaEig^{-1}$, $\ExpaEig>1$, as may be verified by
substitution.

\medskip

\paragraph{Dirichlet boundary conditions:}
In this case the Green's function $\gd$  satisfies
\refeq{GreenFun0}, but,  in addition, is subject to the  Dirichlet
boundary conditions:
\[ \gd_{0j} = \gd_{i0}= \gd_{\ell +1,j} = \gd_{i,\ell+1}=0
\,.
\]

It is possible to determine $\gd$ in two different ways.
The first one  is to use  the fact that $ \gd_{ij}=(\D^{-1})_{ij}$, where
$ \D$ is tridiagonal $[\ell\!\times\!\ell]$ matrix
\[\D= \left(\begin{array}{ccccccc}
s & -1 & 0 & 0 &\dots &0&0 \\
-1 & s & -1 & 0 &\dots &0&0 \\
0 & -1 & s & -1  &\dots &0 & 0 \\
\vdots & \vdots &\vdots & \vdots & \ddots &\vdots &\vdots\\
0 & 0 & \dots &  \dots &\dots &-1 & s  \end{array} \right )
\,.
\]
Since $\D$ is of a tridiagonal form,  its inverse can be  found explicitly.

An alternative  way to evaluate   $\gd_{ij}$   is to use {the free
Green's function \refeq{GreenFun00}} and take anti-periodic sum (similar
method can be used for  periodic and Neumann  boundary conditions)
\[   \gd_{ij}=\sum_{n=-\infty}^{\infty} \gd_{i,j+2 n(\ell+1)}- \gd_{i,-j+2 n(\ell+1)}. \]
This approach has an advantage of being  easily extendable  to
$\integers^2$ case. After substituting $\gd$ and taking the sum one
obtains
\begin{equation}
 \gd_{ij}=  \begin{array}{ll}
        \frac{U_{i-1}(s/2)U_{\ell-j}(s/2)}{U_{\ell}(s/2)} \qquad &\mbox{for } i\leq j\\
        \frac{U_{j-1}(s/2)U_{\ell-i}(s/2)}{U_{\ell}(s/2)} \qquad &\mbox{for } i> j .
        \end{array}
\,,
\label{BGtempCatGF}
\end{equation}
where \(
U_n(s/2)=\frac{\sinh (n+1)\Lyap}{\sinh \Lyap}
\,
\)
are Chebyshev polynomials of the second kind, and  $e^\Lyap=\ExpaEig$.
Note that the Green's function is strictly positive for both boundary
conditions.

\subsection{Green's  function for 2\dmn\ square lattice}
\label{sect:Green2D}

The free Green's function
$\gd(z,z')\equiv \gd(z-z',0)\equiv \gd_{z z'}$
solves  the equation
\begin{equation}
 (-\Box+{2s}-4)\gd_{z z'}=\delta_{zz'}
\,,\qquad
  z=(n,t) \in \integers^2
\,.
\label{GreenFun000}
\end{equation}
The solution is given by the double integral\rf{Martin06}
\beq
 \gd_{z0}=\frac{1}{{2}\pi^2}\int_{0}^{\pi}\int_{0}^{\pi}
           \frac{\cos(nx)\cos(ty)}{{s-\cos x -\cos y}} dx dy
 \,,
\ee{GreenFun1}
which, in turn, can be recast into single integral form,
\bea
 \gd_{z0} &=& \frac{1}{2\pi^3}\int_{-\infty}^{+\infty}d\eta\int_{0}^{\pi}\int_{0}^{\pi}
   \frac{\cos(nx)\cos(ty)}{({s}-2\cos x -i\eta)({s} -2\cos y+i\eta) } dx dy
    \continue
    &=&
 \frac{1}{2\pi}\int_{-\infty}^{+\infty}d\eta\,
 \frac{\mathcal{L}(\eta)^{-n}\mathcal{L}^*(\eta)^{-t}}{|\mathcal{L}(\eta)-\mathcal{L}(\eta)^{-1}|^2 }
\,,
\label{GreenFun2}
\eea
where
\beq
\mathcal{L}(\eta)+\mathcal{L}(\eta)^{-1}= {s}+i\eta, \qquad  | \mathcal{L}(\eta)|>1
\,. \ee{Lfunc}

The above equation can be thought as the integral over a product of two
$\integers^1$ functions:
\beq
 \gd_{z0} =
 \frac{1}{2\pi}\int_{-\infty}^{+\infty}d\eta\,
                 \gd_{n 0}({s}+i\eta) \gd_{t 0}({s}-i\eta)
\,.
\ee{GreenFun3}
An alternative representation is given by modified Bessel
functions $I_n(x)$ of the first kind\rf{Martin06}:
\beq
\gd_{z0} =\int_{0}^{+\infty}d\eta\,e^{-{s\eta/2}} I_n(\eta) I_t (\eta)
\,,
\ee{GreenFun4}
which demonstrates that $ \gd_{z z'}$ is
positive for all $z=(n,t)$.
The representation (\ref{GreenFun4})
enables explicit evaluation of the $n=t$ diagonal elements
in terms of a Legendre function,
\[
 \gd_{z0}=\frac{1}{2\pi i}Q_{n-1/2}({s^2/2}-1 ),\qquad
        {s^2/2}-1 >1, \quad z=(n,n)
 \,.
\]

\paragraph{Dirichlet boundary conditions.}
Consider next the Green's function $\gd_{zz'}$ which  satisfies
\refeq{GreenFun000} within the rectangular domain
\(
\R=\{ (n,t) \in
\integers^2 |1 \leq n\leq \ell_1 , 1 \leq t\leq \ell_2    \}
\)
and vanishes at its boundary $\partial \R$, see \reffig{fig:block2x2}(a).
By applying the same method as in the case of 1\dmn\ lattices we  get
\begin{eqnarray*}
\gd_{zz'}
&=&\sum_{j_1,j_2=-\infty}^{+\infty}\!\!
\gd_{n-n'+2j_1(\ell_1+1), t-t'+2j_2(\ell_2+1)} +
\gd_{n+n'+2j_1(\ell_1+1),t+t'+2j_2(\ell_2+1)}\\
& &\qquad -\gd_{n-n'+2j_1(\ell_1+1), t+t'+2j_2(\ell_2+1)}
  -\gd_{n+n'+2j_1(\ell_1+1),t-t'+2j_2(\ell_2+1)}
\,,
\end{eqnarray*}
where $\gd_{zz'}$ is the free Green's function \refeq{GreenFun1}.
Substituting  \refeq{GreenFun3} yields
the \spt\ Green's function as a convolution of the two 1\dmn\
Green's functions \refeq{BGtempCatGF}
\begin{equation}
 \gd_{z z'}  =\frac{1}{2\pi}\int_{-\infty}^{+\infty}d\eta\,
              \gd_{nn'}({s}+i\eta)  \gd_{tt'}({s}-i\eta)
\,.
\label{DirGreenFun1}
\end{equation}

\paragraph{Exponential decay of Green's function.}
From \refeq{DirGreenFun1} we have a bound on the magnitude of {the}
Green's function,
\bea
    &&|\gd_{z z'}|  \leq \frac{1}{2\pi}\int_{-\infty}^{+\infty}d\eta\,
              |\gd_{nn'}({s}+i\eta)|\,|\gd_{tt'}({s}-i\eta)|=
    \continue
    &=&
\int_{-\infty}^{+\infty} \frac{d\eta}{2\pi}\,
 \left( \frac{|\mathcal{L}|^{-|n-n'+1|}|\mathcal{L}|^{-|t-t'+1|} }{|\mathcal{L}-\mathcal{L}^{-1}|^2 }\right)\left(\frac{ \K_{n}K_{\ell_1 -n'+1}\K_{t}\K_{\ell_2 -t'+1} }{ \K_{\ell_1} \K_{\ell_2}}\right)
\,,
\label{GreenFunIneq}
\eea
where  $\mathcal{L}(\eta)$ is the root of the  equation (\ref{Lfunc})
with the largest absolute value, and
\[ 
 \K_{j}(\eta) =|1-\mathcal{L}(\eta)^{-2j}|, \qquad j=1,2,\dots \,.
\] 
We now show that the first factor in the integrand of  (\ref{GreenFunIneq})
decays  exponentially with increasing $n-n'$, $t-t'$, while the second
one is bounded by a constant, and consequently
the Green's function $\gd_{zz'}$ decays exponentially
with increasing \spt\ distance between lattice points $z$ and $z'$.

By  (\ref{Lfunc}) we have
\[
  |\mathcal{L}(\eta)-\mathcal{L}(\eta)^{-1}|^2 =|({s}+i\eta)^2-4|
  \,.
\]
A lower bound on $|\mathcal{L}(\eta)|$ follows from the observation that
for ${s>2}$  the  minimum of
$|\mathcal{L}(\eta)|$ is achieved at $\eta=0$. Indeed, from  the identity
\[
\left(\frac{{s}}{|\mathcal{L}(\eta)|+|\mathcal{L}(\eta)|^{-1}}\right)^2
+ \left(\frac{\eta}{|\mathcal{L}(\eta)|-|\mathcal{L}(\eta)|^{-1}}\right)^2=1
\]
we obtain
\[
\frac{s/2}{|\mathcal{L}(\eta)|+|\mathcal{L}(\eta)|^{-1}}\leq 1,
\]
which, in turn, implies
\begin{equation}
|\mathcal{L}(\eta)|\geq \mathcal{L}(0)= e^\nu >1, \qquad  \cosh \nu ={s/2}
\,.
\label{LBound}
\end{equation}
The lower and upper bounds  on functions $\K_{j}(\eta)$ follow,
 \beq
 2>1 + |\mathcal{L}(\eta)^{-2j}| \geq \K_{j}(\eta) \geq 1 - |\mathcal{L}(\eta)^{-2j}| > 1- e^{-2\nu}.
\ee{KBound}
Inserting  (\ref{KBound}), (\ref{LBound}) into  (\ref{GreenFunIneq}) yields
 \beq
 |\gd_{zz'}|< C\exp(\nu|n-n'|+\nu|t-t'|)
\,,
\ee{GreenFunIneqFinal}
where
\begin{equation}
C=\left(\frac{2}{\sinh \nu}\right)^2 \int_{-\infty}^{+\infty}
        \frac{d\eta}{2\pi}\,|({s}+i\eta)^2-4|
\,.
\end{equation}

\subsection{Lattice Green's identity}
\label{sect:Green2Dident}

Consider
    \begin{equation}
   (-\Box+{(s-2)d})\,\ssp(z)= \m(z)
\,,
\label{CatMapContinues}
\end{equation}
where $\ssp(z)$,  $\m(z)$  are $C^2$ functions of continuous coordinates
$z\in \mathbb{R}^d$. For ${s<2}$ this is the inhomogeneous Helmholtz
equation, whose general solution is a sum of complex exponentials. For
${s>2}$, the case studied here, the equation is known as the screened
Poisson equation\rf{FetWal03}, or Yukawa equation, whose general solution
is a sum of exponentials.

Let $g(z,z')$, $z,z' \in \R$ be the corresponding
Green's function  on  a  bounded, simply connected domain $\R\subset
\mathbb{R}^d$,
\begin{equation}
 (-\Box+{(s-2)d})\,g(z,z')=\delta^{(d)}(z-z'),
\,
\label{GreenFunContinues}
\end{equation}
satisfying some boundary condition  (e.g., periodic, Dirichlet or Neumann) at
$\partial \R$.
The Green's function identity allows us to connect the values of  $\ssp_{z}$
inside of $\R$ with the ones attained at the boundary:
\begin{eqnarray}
 x(z) &=& \int_{\R} g(z,z')\m(z') dz'\nonumber \\
 &-& \int_{\partial \R} \nabla_n\,g(z,z'')x(z'')\,d z''
  +  \int_{\partial \R} \nabla_n\,x(z'') g(z,z'')\,d z''
\,.
 \label{GreensTheor}
\end{eqnarray}
The analogous theorem holds in the discrete setting as well.  For the
sake of simplicity, we will  restrict our considerations  to  $d=1,2$.

\noindent\emph{1\dmn\ lattice.}
Let $\gd_{tt'}$ be a Green's function on $\integers^1$  satisfying
\refeq{GreenFun0} and some boundary condition at the end points $0$,
$\ell+1$.   To prove {Green's theorem} for a solution $\ssp_{t}$ of
\refeq{OneCat}  we  multiply  each  side of  this equation by
$\gd_{tt'}$  and sum up over the index $t$ running  from $1$ to $\ell$.
In a similar way, the two sides of     \refeq{GreenFun0} can be
multiplied with  $x_t$ and summed up over the same interval. After
subtraction of two equations, we obtain
\begin{eqnarray*}
 \ssp_{t}&=&\sum_{t'=1}^{\ell}\Ssym{t'}\gd_{t't} -\ssp_{\ell} \gd_{\ell+1,t}+\ssp_0 \gd_{1t} -\ssp_1 \gd_{0t}+\ssp_{\ell+1} \gd_{\ell t}, \\
 &=&\sum_{t'=1}^{\ell}\Ssym{t'}\gd_{t't} - \ssp_{\ell} \partial_n \gd_{\ell t}-\ssp_1 \partial_n \gd_{1t}   +\partial \ssp_1 \gd_{1t} +\partial\ssp_{\ell} \gd_{\ell t},
\end{eqnarray*}
with $\partial_n \phi_{\ell} := \phi_{\ell+1}-\phi_{\ell} $,
$\partial_n \phi_{1} := \phi_{0}-\phi_{1} $
being the normal derivatives  at the two boundary points. This equation is
of the exactly same form as \refeq{GreensTheor}.  For the Green's function
$\gd$ with the Dirichlet boundary condition it  simplifies further  to:
\begin{equation}
 \ssp_{t}=\sum_{t'=1}^{\ell}\Ssym{t'}\gd_{t't}
          +\ssp_0 \gd_{1t}
          +\ssp_{\ell+1} \gd_{\ell t}
\,.
\label{GreensTheor1DLattice}
\end{equation}

\noindent\emph{2\dmn\ lattice.}
Let $\ssp_{z}$ be a solution of  \refeq{CoupledCats} within   a
domain  $\R$.  After multiplication of two sides of
\refeq{CoupledCats} with {the} Green's function  satisfying \refeq{GreenFun000}
and summing up over the domain $\R$  we get
\[
\sum_{z' \in\R}\gd_{zz'}\left(\sum_{i=1}^4 \ssp_{z'+e_i}-{2s}\ssp_{z'}\right)
  =\sum_{z' \in\R}\gd_{zz'}\Ssym{z'}
  \,,
\]
with $e_{1,2}=(0,\pm 1), e_{3,4}=(\pm 1,0)$  being the four vectors connecting the neighboring sites
of the lattice. In the same way, we obtain from \refeq{GreenFun000}
\[
\sum_{z' \in\R}\ssp_{z'}\left(\sum_{i=1}^4 \gd_{z+e_i,z'}-{2s}\gd_{zz'}\right)
=\sum_{z' \in\R}\delta_{zz'}\ssp_{z'}
=-\ssp_{z}
\,.
\]
The subtraction of two equations yields
\begin{equation}
 x_z = \sum_{z' \in\R} \gd_{zz'}\Ssym{z'}
 -  \sum_{z'' \in\partial \R}\nabla_n \gd_{zz''}\ssp_{z''}
  + \sum_{z'' \in\partial \R} \nabla_n \ssp_{z''} \gd_{zz''}
\,,
\label{GreensTheor2DLattice}
\end{equation}
with  $\nabla_n \phi_z:= \sum_{e_i\notin \R}\phi_{z+e_i} -\phi_{z}$.
For a  Green's function  satisfying Dirichlet  boundary conditions,
\refeq{GreensTheor2DLattice}  can be simplified  further, as illustrated by:

\paragraph{Example: Rectangular domain.}
Given a rectangular domain $\R=\R^{[\ell_1\times\ell_2]}$, let
$\gd_{zz''}$ be the corresponding Dirichlet Green's function vanishing
at the  boundary $\partial\R$. Since  $\gd_{zz''}=0$ for all
$z''\in\partial \R$, equation~\refeq{GreensTheor2DLattice}  can be
written down in the following form
\begin{equation}
 x_z = \sum_{z' \in\R} \gd_{zz'}\Ssym{z'}
 +  \sum_{z'' \in \partial\R}\gd_{z\bar{z}''}\ssp_{z''}
\,,
\end{equation}
where $\bar{z}'' \in  \R$ is a    point of  $\R$,    adjacent   to $z''
\in \partial\R$, see  \reffig{fig:block2x2}(a). Note that  the
rectangular corners are associated with    two  points     of the
boundary ${\partial\R}$.  Otherwise, the relationship between
$\bar{z}''$ and $z''$ is one-to-one.

\section{\catLatt, Hamiltonian formulation}
\label{sect:HamiltonCatLatt}

The Hamiltonian setup of the \catlatt\ \refeq{CoupledCats} is discussed in
detail in \refref{GutOsi15}. In this paper, we use it to generate
{\spt}ly chaotic patterns by time evolution of random initial
conditions on a cylinder infinite in time direction, but $L$-periodic in the
space direction.

In one spatial dimension the momentum
field at the lattice site ($n$'th ``particle'') $q_n$ is given by
\(
p_{nt}= \ssp_{nt} - \ssp_{n,t-1}\,,
\)
and the Hamiltonian cat map  \refeq{eq:CatMap} at each lattice site is
coupled to its nearest neighbors by
 \bea
 x_{n,t+1}&=& p_{nt}-x_{n+1,t}+a\,x_{nt}-x_{n-1,t}-m^x_{n,t+1}
\continue
 p_{n,t+1}&=&
   bp_{nt}+(ab-1)\,x_{nt}-{b}\left(x_{n+1,t}+x_{n-1,t}\right)- m^p_{n,t+1}
 \;,
\label{eqmotion}
\eea
where $(x_{nt},p_{nt})$ are the coordinate and momentum of the $n$'th
``particle'' at the discrete time $t$, and $m^x_{nt},  m^p_{nt}$ are the
corresponding winding numbers necessary to bring $(x_{n,t+1},p_{n,t+1})$
to the unit interval.
As for the {cat map}, integers $a$ and $b$ are arbitrary; the Lagrangian
form \refeq{eq:CatMapNewton2} of the map only depends on their sum
${2s}=a+b$. Hamiltonian winding numbers\rf{GutOsi15} are connected
to the Lagrangian ones by
$m_{nt}=-b\,m^x_{nt}+m^x_{n,t+1}-m^p_{nt}$.

In the Hamiltonian setup, the $d=2$ \catlatt\ is thus viewed as a
$\integers^1$ chain of  linearly coupled cat maps acting on the \statesp\
$V$, a direct product of the  2\dmn\  tori $V=\otimes_{n}\mathbb{T}_n^2$,
${n\in\integers^1}$. Each torus $\mathbb{T}_n^2$ is equipped with the
\statesp\ coordinate pair $(x_{n},p_{n}) \in (0,1]\times(0,1]$
corresponding to the position and momentum  of the $n$'th ``particle''.
The law of time evolution \refeq{eqmotion} is time as well as space
translation invariant under shifts $\Tshift$ and $\Spshift$, respectively.
Along  with the infinite chain, in
\reffigs{fig:AKSs13TwoBlocks}{fig:AKSs13TwoBlocks1} we use the spatial
finite setup,
with $L$
Hamiltonian cat maps coupled cyclically, and \refeq{eqmotion} subject to
the periodic boundary conditions $(x_{n},p_{n}) = (x_{n+L},p_{n+L})$.
This  defines a linear map
\(
V_{\scriptscriptstyle L}\to V_{\scriptscriptstyle L}
\)
on the $2L$\dmn\  \statesp\ $V_{\scriptscriptstyle L}=\otimes_{n=1}^L
\mathbb{T}_n^2$:
\beq
Z_{t+1}=\B_{\scriptscriptstyle L} Z_{t}\,\;\; \mod 1
\,, \qquad
Z_{t}
= (\ssp_{1,t},p_{1,t}, \dots  \ssp_{{\scriptscriptstyle{L}},t},
   p_{{\scriptscriptstyle{L}},t})^T
\,,
\ee{LperHamiltonian}
where $Z_{t}\in{V_{\scriptscriptstyle L}}$,
and $\B_{\scriptscriptstyle L}$ is the $[2L\!\times\!2L]$  matrix
\[
 \B_{\scriptscriptstyle L}= \left(\begin{array}{cccccc}
      A     & B     &     0  & \dots  & 0 & B\\
      B     & A     &     B      &\dots  &  0 & 0\\
       0 & B     &    A      &\dots  &  0 & 0\\
      \vdots&\vdots &   \vdots & \ddots &\vdots  &\vdots \\
       0 &  0 &    0  & \dots  &A  & B\\
      B     &  0 &    0  & \dots  &B  & A\\
     \end{array} \right)
 \,,\qquad\begin{array}{l}
A= \left(\begin{array}{cc} a& 1\\
ab-1 & b
\end{array} \right)
 \\
B=-\left(\begin{array}{cc} 1& 0\\
b & 0
\end{array} \right)
          \end{array}
\,.
\]
The spectrum of the Lyapunov exponents (linear stability of the Hamiltonian
\catlatt\ \refeq{eqmotion}, posed as $t=0$ initial problem, evolving in time)
is given by the $\B_{\scriptscriptstyle L}$ eigenvalues (see eq.~(3.6) of
\refref{GutOsi15}):
\beq
 \ExpaEig_k+\ExpaEig_k^{-1}={2s}-2\cos(2\pi k/L), \qquad k=1,\dots L
\,,
\ee{qudreq}
${2s}=a+b \in \integers$.
Accordingly, the map  is fully hyperbolic iff ${|s|>2}$, with all Lyapunov
exponents $\Lyap_k^{\pm} = \pm\log|\ExpaEig_k|$ paired as $\Lyap_k^+=
-\Lyap_k^-$, $\Lyap_k^+ >0$, for all $k$. Since the matrix
$\B_{\scriptscriptstyle L}$ is symplectic,  the map
\refeq{LperHamiltonian}  preserves the measure $d\Msr_{\scriptscriptstyle
L}=\prod_{i=1}^L dx_i dp_i$ on $V_{\scriptscriptstyle L}$. In the limit
$L\to\infty$, $\Msr_{\scriptscriptstyle L}$ induces the corresponding
measure $\Msr$  on $V$, invariant under both discrete time evolution and
discrete spatial lattice translations.

\section{Metric entropy}
\label{sect:metricEntropy}
 \begin{figure}	
 	\centering
	\includegraphics[width=0.85\textwidth]{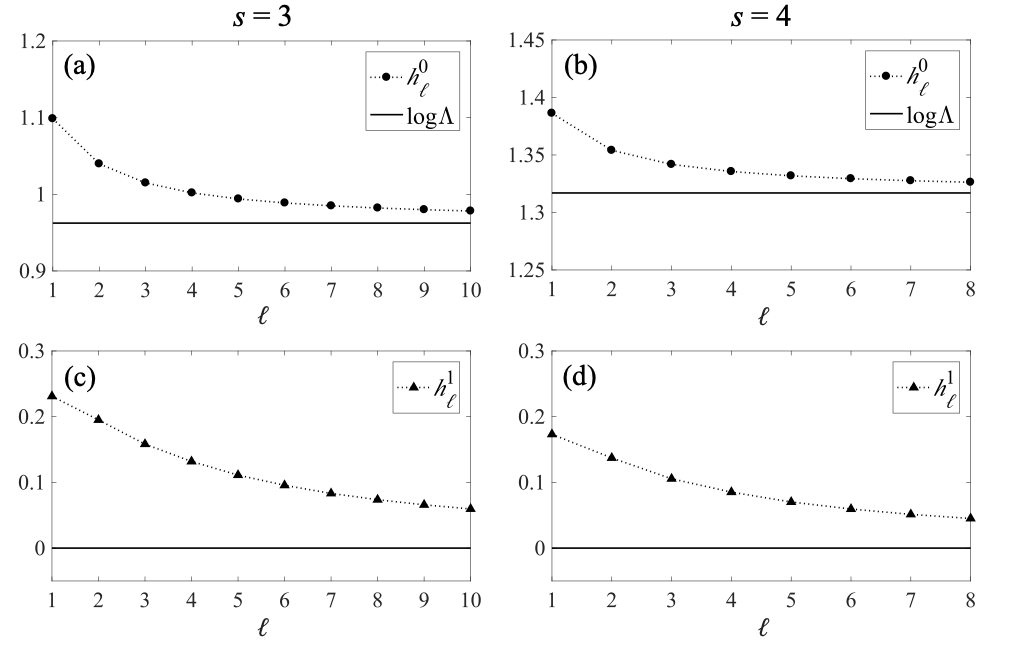}
        \hspace{0.1\textwidth}
	\caption{\label{fig:RJentropyPlot}
Geometric, $h_\ell^1$, and interior, $h_\ell^0$, parts of entropy vs. the
length $\ell$ of
    a cat map symbol \brick\
for
(a), (c) $s=3$, and
(b), (d) $s=4$.
The interior
    part of entropy
$h_\ell^0$ converges to the exact value of the metric entropy, which is
$\ln(3+\sqrt{5})/2$ for $s=3$ and $\ln(2+\sqrt{3})$ for $s=4$,
respectively.  The geometric part $h_\ell^1$ converges to $0$.
    }
\end{figure}

Since the Lyapunov exponents of {\catlatt} are known explicitly for any
finite lattice, its  metric entropy is known exactly, through the Pesin
entropy formula\rf{Pesin77}. On the other hand, the metric entropy can be
represented through  sums of symbol {\brick s} $\Mm_\R$ measures  in the
limit of growing domain size $|\R|$. As we show below, this connection
can be used to extract the asymptotics of the Jacobian $d(\R)$ which, in
turn, provides the maximum  of measures for symbol {\brick s} in a given
domain $\R$.

\subsection{Cat map entropy}
\label{sect:catEntropy}

Given  measures  $\Msr({b})$ of {\brick s} $ {b}$ of an arbitrary finite
length $\cl{b}=\ell$ one can estimate observables of the dynamical system
with increasing  precision as $\ell$ grows. In the following we apply
this  to estimation of the metric entropy of the cat map.

The metric entropy  of the cat map with respect to measure $\Msr$ can be
represented as the limit
\begin{equation}
h_\Msr =\lim_{{\ell}\to\infty} {h}_\ell,
\label{EntropyEigenvalue}
\end{equation}
where
\[
h_\ell = -\frac{1}{\ell} \sum_{|{b}|=\ell}{\Msr({b}) \log{\Msr({b})}}
\,.
\]
By using the decomposition \refeq{FreqDecomp}, $h_\ell$  can be split as
$h_\ell=h_\ell^{1}+h_\ell^{0}$ into ``geometric'' and ``internal'' parts:
\[
h_\ell^{1}=
-\frac{\sum_{|{b}|=\ell}{|\Pol_{{b}}| \log{|\Pol_{{b}}|}}}{\ell \sum_{|{b}|=\ell}{|\Pol_{{b}}| }},
\qquad
h_\ell^{0}=-\frac{1}{\ell}\log d_\ell
          = \frac{1}{\ell}\log U_{\ell}({s}/{2})
\,.
\]
The ``interior'' part $h_\ell^{0}$ converges to $\log \ExpaEig$ with the
rate $O(1/\ell)$. Since $h_\Msr=\log \ExpaEig$, the cat  map  metric
entropy\rf{Sinai59},  the geometrical part must vanish,
$\lim_{{\ell}\to\infty}h_\ell^{1}=0$. We {illustrate} this numerically for
several values of $s$ in \reffig{fig:RJentropyPlot}. In other words,
metric entropy  is determined solely by the asymptotics of the Jacobian
$d_\ell$.

\subsection{\catLatt\ entropy}
\label{sect:catLattEntropy}

By the linearity of the \catlatt\ automorphism, every \catlatt\
$[\speriod{}\!\times\!\period{}]$ {\twot} has the same spectrum of  the
Lyapunov exponents  $\Lyap_k$, $k=1,\dots,2\speriod{}$, given by the
eigenvalues \refeq{qudreq} $\ExpaEig_k =e^{\Lyap_k}$ of the matrix
$\B_{\scriptscriptstyle L}$.

Accordingly, the map  is fully hyperbolic iff ${2s}=a+b>4$.    In this
case all solutions   of \refeq{qudreq} are paired such that $\Lyap_k^+=
-\Lyap_k^-$ and  $\Lyap_k^+ >0$ for all $k$.  The metric entropy  of
$\Phi_{\speriod{}}$ for a finite $\speriod{}$ is given by the sum of all
positive exponents:
\[
h(\Phi_{\speriod{}})= \sum_{k=1}^\speriod{}  \Lyap_k^+
\,.
\]
 For the   infinite lattice the corresponding  {\spt} entropy of
 $\Phi$ with respect to $\Msr$ is given by  the limit
$
   h_\Msr(\Phi) =\lim_{L\to\infty}\frac{1}{L} {h}(\Phi_{\scriptscriptstyle L})
$.
Alternative, more convenient way to evaluate $h_\Msr(\Phi)$ is to
use numbers $\mathcal{N}_{\speriod{}\period{}}$ of  periodic tori with
spatial-temporal periods $\speriod{},\period{}$, see \refref{GutOsi15,CL18}.
For the linear hyperbolic automorphisms the topological and metric entropies
coincide, leading to
\bea
h_\Msr(\Phi) &=&\lim_{\speriod{},\period{}\to\infty}\frac{1}{LT} \log \mathcal{N}_{\speriod{}\period{}}
\continue
&=&\lim_{\speriod{},\period{}\to\infty}\frac{1}{\speriod{}\period{}} \log \mbox{det} (I-\B_{\speriod{}}^\period{})
=  \lim_{\speriod{},\period{}\to\infty}\frac{1}{\speriod{}\period{}} \log \mbox{det} (I-\B_{\period{}}^\speriod{})
\,.
\label{TopologicalEntropy}
\eea
This yields a closed formula:
\beq
h_\Msr(\Phi)=\frac{1}{\pi^2} \int_0^{\pi}\int_0^{\pi}\,dx\,dy
\log({2s}-4+4\sin^2 x +4\sin^2 y)
 \,.
\ee{SpatiotemporalEntropy}

Recall that for the (temporal) cat map the constant $d(\R)$  was given
explicitly  by  \refeq{SingleCatJacobian}. In the case of
{\catlatt} we were unable to derive any explicit formula for $d(\R)$.
On the other hand, by the same argument as in the {cat map} case, in the
large domain \R\ limit the asymptotics of $d(\R)$ is determined  by the
metric entropy,
\[
-\lim_{|\R|\to\infty} \frac{1}{|\R|}\log d(\R) =  h_\Msr
\,,
\]
where $h_\Msr$ is the {\spt}  metric entropy \refeq{SpatiotemporalEntropy}
of the {\catlatt}. Therefore, asymptotically we expect
\beq
d(\R) \sim
e^{-|\R| h_\Msr}
\,.
\ee{AsymJacobian}

\printbibliography[
heading=bibintoc,
title={References}
				  ] 
\end{document}